\documentclass[10pt,aps,pra,twocolumn,preprintnumbers,longbibliography,superscriptaddress,floatfix,nobalancelastpage]{revtex4-1}

\usepackage[T1]{fontenc}
\usepackage{graphicx}  
\usepackage{dcolumn}   
\usepackage{bm}        
\usepackage{amssymb}
\usepackage{mathtools}
\usepackage{amsthm}
\usepackage{ragged2e}
\usepackage{verbatim}
\usepackage[colorlinks=true,linkcolor=blue,citecolor=blue,urlcolor=blue,plainpages=false,pdfpagelabels]{hyperref}
\usepackage{color,xcolor,colortbl}
\usepackage{multirow}
\usepackage{bbm}
\usepackage{lipsum}
\usepackage{array}
\usepackage{etoolbox}

\usepackage{times}
\usepackage{txfonts}
\usepackage[mathcal]{euscript}

\providecommand{\U}[1]{\protect\rule{.1in}{.1in}}
\newtheorem{theorem}{Theorem}

\newtheorem{lemma}[theorem]{Lemma}

\newtheorem{proposition}[theorem]{Proposition}

\newcommand{\e}{\mathrm{e}}
\newcommand{\I}{\text{i}}

\newcommand{\ket}[1]{| #1 \rangle}
\newcommand{\bra}[1]{\langle #1 |}

\def\U{\mathrm{U}}
\def\A{\mathcal{A}}

\def\E{\mathcal{E}}

\def\H{\mathcal{H}}
\def\U{\mathcal{U}}
\def\L{\mathcal{L}}
\def\M{\mathcal{M}}
\def\N{\mathcal{N}}

\newcommand{\Tr}{\operatorname{Tr}}

\newcommand{\id}{\operatorname{id}}

\newcommand{\norm}[1]{\lVert#1\rVert}
\newcommand{\Norm}[1]{\left\lVert#1\right\rVert}
\renewcommand{\t}{{\scriptscriptstyle\mathsf{T}}}

\allowdisplaybreaks

\begin{document}
\title{Information-theoretic aspects of the generalized amplitude damping channel}
\author{Sumeet Khatri}\email{skhatr5@lsu.edu}
\affiliation{Hearne Institute for Theoretical Physics, Department of Physics and Astronomy,
	Louisiana State University, Baton Rouge, Louisiana 70803, USA}
\author{Kunal Sharma}\email{ksharm7@lsu.edu}
\affiliation{Hearne Institute for Theoretical Physics, Department of Physics and Astronomy,
Louisiana State University, Baton Rouge, Louisiana 70803, USA}
\author{Mark M. Wilde}
\affiliation{Hearne Institute for Theoretical Physics, Department of Physics and Astronomy,
Louisiana State University, Baton Rouge, Louisiana 70803, USA}
\affiliation{Center for Computation and Technology, Louisiana State University, Baton
Rouge, Louisiana 70803, USA}

\begin{abstract}

The generalized amplitude damping channel (GADC) is one of the sources of noise in superconducting-circuit-based quantum computing. It can be viewed as the qubit analogue of the bosonic thermal channel, and it thus can be used to model lossy processes in the presence of background noise for low-temperature systems. In this work, we provide an information-theoretic study of the GADC. We first determine the parameter range for which the GADC is entanglement breaking and the range for which it is anti-degradable. We then establish several upper bounds on its classical, quantum, and private capacities. These bounds are based on data-processing inequalities and the uniform continuity of information-theoretic quantities, as well as other techniques. Our upper bounds on the quantum capacity of the GADC are tighter than the known upper bound reported recently in [Rosati \textit{et al}., Nat.~Commun.~9, 4339 (2018)] for the entire parameter range of the GADC, thus reducing the gap between the lower and upper bounds. We also establish upper bounds on the two-way assisted quantum and private capacities of the GADC. These bounds are based on the squashed entanglement, and they are established by constructing particular squashing channels. We compare these bounds with the max-Rains information bound, the mutual information bound, and another bound based on approximate covariance. For all capacities considered, we find that a large variety of techniques are useful in establishing bounds.

\end{abstract}

\date{\today}
\maketitle

\section{Introduction}

	One of the main goals of quantum information theory is to determine the optimal rate of sending information (classical or quantum) through quantum channels \cite{H13book,MH06,W17,Wat18}. Quantum channels model the noisy evolution that quantum states undergo when they are transmitted via some physical medium.

	Depending on the message and the availability of resources, communication protocols over quantum channels can be divided into different categories. In particular, classical communication, entanglement-assisted classical communication, private classical communication, and quantum communication are some of the communication protocols that have been studied in the last few decades (see \cite{H13book,MH06,W17,Wat18} for reviews). The notion of the capacity of a channel defined by Shannon \cite{Shannon1948} can be extended to the quantum domain for these different communication protocols (see Sec.~\ref{sec:capacity-of-qc} for formal definitions). 

	The optimal rate (capacity) of any communication protocol depends on the properties of the quantum channel. In general, the best characterization of the capacities of a quantum channel is given by an optimization over regularized information quantities over an unbounded number of copies of the channel. Hence, it appears to be generally difficult to calculate the quantum and private capacities of quantum channels \cite{CEMOGS15,ES15} except for a special class of quantum channels that are degradable (see definitions in Sec.~\ref{sec:prelim}), in which case the regularized quantities reduce to simpler formulas that are functions of only one copy of the channel \cite{DS04,GS08}. Recently, however, it was shown that one can calculate quantum capacity for some channels that are not degradable \cite{FW07,GJL18}. Furthermore,  recent progress in estimating and understanding the quantum capacity of low-noise and some other channels has been reported in~\cite{LDS18, LLS18,LLSb18,BL18}. 
	
	Remarkably, even in the qubit case, very little is known when it comes to exact, computable expressions for the communication capacities of quantum channels. For example, two of the most widely considered noise models in quantum information and communication are the depolarizing channel and the amplitude damping channel. The classical capacity of the qubit depolarizing channel is known \cite{King02,K03}, but its quantum capacity (for its entire parameter range) is not. Similarly, the quantum capacity of the amplitude damping channel is known \cite{GF05}, but its classical capacity (for its entire parameter range) is not. These are two of the most significant open problems in quantum Shannon theory.
	
	In general, the difficulty in obtaining exact expressions for the communication capacities of quantum channels has led to a wide body of work on obtaining lower and upper bounds on these quantities. With the recent developments in quantum communication technologies, it is important to study different physically motivated noisy communication processes (quantum channels) and to establish lower and upper bounds on their communication capacities in terms of the channel parameters. Moreover, these communication rates also play a critical role in the context of distributed quantum computing between remote locations and in benchmarking the performance of quantum key distribution and quantum networks.

	In this work, we provide an information-theoretic study of the generalized amplitude damping channel (GADC). As the name suggests, the GADC is indeed a generalization of the amplitude damping channel. Specifically, the GADC is a qubit-to-qubit channel, and it models the dynamics of a two-level system in contact with a thermal bath at non-zero temperature. It can be used to describe the $T_1$ relaxation process due to the coupling of spins to a system that is in thermal equilibrium at a temperature higher than the spin temperature \cite{NC10,MKTSKIMW00,PhysRevA.62.053807}. The GADC is also one of the sources of noise in superconducting-circuit-based quantum computing \cite{CB08}. It can additionally be used to characterize losses in linear optical systems in the presence of low-temperature background noise \cite{Pan17}. In the case that the thermal bath is at zero temperature, the GADC reduces to the amplitude damping channel, which arises naturally as a noise model in spin chains \cite{Bose03,GF05}.
	
	The GADC can be thought of as the qubit analogue of the bosonic thermal channel, which is used to model loss in quantum optical systems and is particularly relevant in the context of communication through optical fibers or free space  \cite{YS78,S09,RGRKVRHWE17}. Moreover, in the context of private communication, tampering by an eavesdropper can be modeled as the excess noise realized by a thermal channel \cite{NH04,LDTG05}. A lower bound on the quantum capacity of a bosonic Gaussian thermal channel was proposed in \cite{HW01}. Recently, several upper bounds on the energy-constrained quantum and private capacities of a thermal channel have been established in \cite{SWAT} (see also \cite{NAJ18} in the context of lower and upper bounds on the energy-constrained quantum capacity). Moreover, the unconstrained quantum capacity of a thermal channel has been studied in \cite{PLOB17, NAJ18, RMG18, WTB17, SWAT}. However, the communication capacities of a qubit thermal channel, i.e., the GADC, have not been studied extensively. 
	
	Some prior works have established bounds on the various capacities of the GADC. Since it is not a degradable channel for nearly all parameter values, determining its quantum capacity exactly appears to be a difficult task. It is worth noting, however, that it is degradable in the special case that it reduces to the amplitude damping channel, and thus the quantum and private capacities of the amplitude damping channel are  simply given by its coherent information \cite{GF05},  due to the additivity of the coherent and private information for degradable channels \cite{DS04,GS08}. An upper bound on the quantum capacity of the GADC in general was established in \cite{RMG18} by using the notion of weak degradability. Furthermore, lower and upper bounds on the classical capacity of the GADC have been established in \cite{F18} (see also \cite{FFK18}). In \cite{LM07}, the mutual information of the GADC was calculated, thus establishing its entanglement-assisted classical capacity \cite{BSST99,ieee2002bennett,Hol01a}, which is in turn an upper bound on its unassisted classical capacity. In general,  half the mutual information of a quantum channel is  an upper bound on its two-way assisted quantum and private capacities \cite{TGW14a, TGW14b, GEW16}. Thus, one can infer from \cite{LM07} and \cite{TGW14a, TGW14b, GEW16} an upper bound on the two-way assisted quantum and private capacities of the GADC.

\section{Summary of Results}

	In this paper, we study the GADC in detail by first deriving its intrinsic information-theoretic properties, such as necessary and sufficient conditions for entanglement breakability \cite{HSR03} and anti-degradability \cite{CG06}. We then consider several upper bounds on the classical, quantum, and private capacities of the GADC; see Table~\ref{table-summary} for a summary.
	
	We start with the classical capacity of the GADC. A first upper bound, known as $C_\beta$, is based on the no-signalling and PPT-preserving codes for classical communication over a quantum channel \cite{WXD18}. In particular, we find an analytical expression for $C_\beta$ of the GADC that depends only on the channel parameters. Another upper bound from \cite{WXD18} on the classical capacity of any quantum channel is the quantity $C_{\zeta}$. We prove that $C_{\zeta}=C_{\beta}$ for the GADC. Two other upper bounds on the classical capacity of the GADC are established by using the notion of $\varepsilon$-entanglement-breakability and $\varepsilon$-covariance \cite{LKDW18}. We also compare these upper bounds with the entanglement-assisted classical capacity upper bound for the GADC \cite{LM07}. 
	
	\begin{table*}
		\renewcommand{\arraystretch}{2.1}
		\centering
		\begin{tabular}{|c || c | c | c |}
			\hline
			\multirow{2}{*}{Capacity} & \multirow{2}{*}{Lower Bounds} & \multicolumn{2}{c|}{Upper Bounds} \\ \cline{3-4}
			& & Quantity & Technique \\ \hline\hline
			\multirow{5}{*}{Classical} & \multirow{5}{*}{\shortstack{$\chi$\\Holevo Information\\(Eq. \eqref{eq-GADC_Hol_inf})}} & $C_{\beta}$ (Eq. \eqref{eq-C_beta_GADC}) & No-signalling and PPT-preserving codes \cite{WXD18} \\ \cline{3-4}
			& & $C_{\text{cov}}^{\text{UB}}$ (Eq. \eqref{eq-CCap_UB2_cov}) & Approximate covariance \cite{LKDW18} \\ \cline{3-4}
			& & $C_{\text{EB}}^{\text{UB}}$ (Eq. \eqref{eq-CCap_UB1_EB}) & Approximate entanglement-breakability \cite{LKDW18} \\ \cline{3-4}
			& & $C_{\text{Fil}}^{\text{UB}}$ (Eq. \eqref{eq-GADC_CCap_UB_Fil}) & Approximate unitality \cite{F18,FFK18}\\ \cline{3-4}
			& & $C_E$ (Eq. \eqref{eq-GADC_mut_inf}) & Entanglement-assisted classical capacity \cite{ieee2002bennett,BSST99,Hol01a} \\ \hline
			\multirow{5}{*}{Quantum} & \multirow{5}{*}{\shortstack{$I_{\text{c}}$\\Coherent Information\\(Eq. \eqref{eq-GADC_LB1})}} & $Q_{\text{DP},1-4}^{\text{UB}}$ (Eq. \eqref{eq-GACD_Qcap_UB_DP_1}--\eqref{eq-GACD_Qcap_UB_DP_4}) & Data processing \cite{WP07,SS08} \\ \cline{3-4}
			& & $Q_{\text{deg},1-2}^{\text{UB}}$ (Eq. \eqref{eq-GACD_Qcap_UB_DP_5}, \eqref{eq-GACD_Qcap_UB_DP_6}) & Approximate degradability \cite{SSWR17} \\ \cline{3-4}
			& & $Q_{\text{a-deg}}^{\text{UB}}$ (Eq. \eqref{eq-GADC_QCap_UB7_Adeg}) & Approximate anti-degradability \cite{SSWR17} \\ \cline{3-4}
			& & $Q_{\text{Rains}}^{\text{UB}}$ (Eq. \eqref{eq-GADC_QCap_UB8}) & PPT-preserving codes \cite{TWW17} \\ \cline{3-4}
			& & $Q_{\text{RMG}}^{\text{UB}}$ (Eq.~\eqref{eq-QTN_QCap_UB_RMG}) & Degradability and data processing \cite{RMG18} \\ \hline
			\multirow{4}{*}{\shortstack{Two-Way\\Assisted\\Quantum}} & \multirow{2}{*}{\shortstack{$I_{\text{c}}$\\Coherent Information\\(Eq. \eqref{eq-GADC_LB1})}} & $Q_{\text{MI}}^{\leftrightarrow,\text{UB}}$ (Eq. \eqref{eq-GADC_Q2cap_UB1}) & One-half mutual information \cite{TGW14a,TGW14b,GEW16} \\\cline{3-4}
			& & $Q_{\text{sq},1-2}^{\leftrightarrow,\text{UB}}$ (Eq. \eqref{eq-GADC_Esq_UB1}, \eqref{eq-GADC_Esq_UB2}) & Squashed entanglement \cite{GEW16,DSW18} \\ \cline{2-4}
			& \multirow{2}{*}{\shortstack{$I_{\text{rc}}$\\Reverse Coherent\\Information\\(Eq. \eqref{eq-GADC_RCI})}} & $Q_{\text{max-Rains}}^{\leftrightarrow,\text{UB}}$ (Eq. \eqref{eq-GADC_Emax_Rmax_UB}) &  PPT-preserving assisted codes \cite{BW18}\\ \cline{3-4}
			& & $Q_{\text{cov}}^{\leftrightarrow,\text{UB}}$ (Eq. \eqref{eq:approx-tele-sim-bnd}) & Approximate covariance \cite{KW17}\\\hline
		\end{tabular}
		\caption{Summary of the lower and upper bounds on the classical, quantum, and two-way assisted quantum capacities of the GADC that we consider in this work. The classical capacity upper bounds are established in Sec.~\ref{sec:c-cap-bounds}. The quantum and private capacity upper bounds are established in Sec.~\ref{sec:q-cap-bounds}. The two-way assisted quantum and private capacities are established in Sec.~\ref{sec:two-way-q-cap-bounds}. We obtain analytic expressions for the quantities $C_{\beta}$ (Proposition~\ref{prop-C_beta}), $Q_{\text{max-Rains}}^{\leftrightarrow,\text{UB}}$ (Proposition~\ref{prop-GADC_Emax}), and $Q_{\text{cov}}^{\leftrightarrow,\text{UB}}$ in this work.}\label{table-summary}
	\end{table*}

	We employ a variety of techniques to establish upper bounds on the quantum and private capacities of the GADC. The first four upper bounds are established, related to the approach of \cite{WP07,SS08} (see \cite{SWAT,RMG18,NAJ18} for bosonic channels), by decomposing any GADC into a serial concatenation of two amplitude damping channels. Since the quantum capacity of an amplitude damping channel is known \cite{GF05}, upper bounds on the quantum capacity of the GADC follow from the data processing property \cite{SN96} of the coherent information of a quantum channel. We call these bounds the ``data-processing bounds.'' We also consider three other upper bounds by using the notion of approximate degradability and anti-degradability, recently developed in \cite{SSWR17}. We call these bounds the ``$\varepsilon$-degradable bound'', ``$\varepsilon$-close-degradable bound,'' and   ``$\varepsilon$-anti-degradable bound.'' We finally employ the Rains information strong converse upper bound from \cite{TWW17} and the relative entropy of entanglement strong converse upper bound from \cite{WTB17} in order to bound the quantum and private capacities of the GADC, respectively.
	

	We compare these upper bounds on the quantum capacity of the GADC with the known coherent information lower bound, and we find that for certain parameter values, the gap between the lower bound and the upper bounds is relatively small. Moreover, we compare these upper bounds with the upper bound established in \cite{RMG18}, and we find that two of our data-processing upper bounds are tighter than the bound in \cite{RMG18} for all parameter values of the channel. Furthermore, the strong converse bounds from \cite{TWW17,WTB17} can be even tighter for certain parameter values.

	We also consider four different upper bounds on the two-way assisted (i.e., feedback-assisted) quantum and private capacities of the GADC. The first two upper bounds are based on the fact that the squashed entanglement of a quantum channel is an upper bound on the two-way assisted quantum and private capacities of any channel \cite{TGW14a,TGW14b,MMW16}. For the third upper bound, we employ the max-Rains information \cite{WD16, WFD18} and the max-relative entropy of entanglement \cite{CMH17}, which are known to be upper bounds on the two-way assisted quantum  \cite{BW18} and private \cite{CMH17} capacities, respectively, for any quantum channel. In fact, for this third upper bound, we have found an analytical expression that establishes that the max-Rains information and max-relative entropy of entanglement are equal for the GADC. We found this analytical expression by analytically solving the semi-definite programs associated to max-Rains information and max-relative entropy of entanglement. The fourth upper bound is based on the notion of approximate covariance. A comparison of these four upper bounds with the mutual information upper bound leads to the conclusion that all four upper bounds are significantly tighter than the mutual information upper bound. 


	The rest of the paper is structured as follows. We begin by summarizing relevant definitions and prior results in Sec.~\ref{sec:prelim}. We derive necessary and sufficient conditions for entanglement breakability and anti-degradability of the GADC in Sec.~\ref{sec:ent-break} and Sec.~\ref{sec:anti-deg}, respectively. We then establish several upper bounds on the classical capacity and the quantum capacity of the GADC in Sec.~\ref{sec:c-cap-bounds} and Sec.~\ref{sec:q-cap-bounds}, respectively. In Sec.~\ref{sec:two-way-q-cap-bounds}, we establish several upper bounds on the two-way assisted quantum and private capacities of the GADC. Finally, we summarize our results and conclude in Sec.~\ref{sec:conclusion}. 
	
	All codes in Mathematica, Matlab, and Python used to assist with the analytical derivations, numerical computations, and the creation of plots are available as ancillary files with the arXiv posting of this paper \footnote{Some of the Matlab code makes use of the Quantinf package \cite{cubitt_matlab} as well as the package QETLAB \cite{qetlab}.}. The Mathematica files contain the code used in the proofs of \eqref{eq-GADC_ent_break}, Proposition \ref{prop-C_beta}, Proposition \ref{prop-GADC_Emax}, and \eqref{eq-GADC_REE}. The Matlab and Python files have been used to compute all the bounds stated in the paper, and the plots have been generated in the included Jupyter notebooks using Python.

\section{Preliminaries}\label{sec:prelim}

	In this section, we review some definitions and prior results relevant for the rest of the paper. We point readers to \cite{MH06, H13book, W17, Wat18} for details and further background. 

Let $\H$ denote a finite-dimensional Hilbert space. The tensor product of two Hilbert spaces $\H_A$ and $\H_B$ corresponding to the quantum systems $A$ and $B$ is denoted by $\H_{AB}\equiv\H_A \otimes \H_B$. We let $d_A$ denote the dimension of $\mathcal{H}_A$. Let $D(\H)$ denote the set of density operators (positive semi-definite operators with unit trace) acting on a Hilbert space $\H$. An extension of a state $\rho_A \in D(\H_A)$ is some state $\rho_{RA} \in D(\H_R \otimes \H_A)$  such that $\Tr_R[\rho_{RA}] = \rho_A$.  Similarly, a purification of a state $\rho_A\in D(\H_A)$ is some pure state $\ket{\phi}_{RA} \in \H_R \otimes \H_A$ such that $\Tr_R[\ket{\phi} \bra{\phi}_{RA}] = \rho_A$. 

	The quantum entropy of a quantum state $\rho \in D(\H)$ is defined as $H(\rho)\equiv -\Tr[\rho \log_2\rho]$. The binary entropy $h_2(x)$ is defined for $x \in [0, 1]$ as 
	\begin{align}
		h_2(x) \equiv - x \log_2(x) - (1-x)\log_2(1-x). 
	\end{align}
	Moreover, throughout the paper we use the bosonic entropy $g(x)$ for $x\geq0$:
	\begin{align}
		g(x) &\equiv (1+x)\log_2(1+x) - x\log_2x\\
		&= (1+x)h_2\left(\frac{x}{1+x}\right).
	\end{align}
	The quantum mutual information of a bipartite state $\rho_{AB} \in D(\H_A \otimes \H_B)$ is defined as
	\begin{align}
		I(A;B)_{\rho} \equiv H(\rho_A) + H(\rho_B) - H(\rho_{AB}). 
	\end{align}

	Let $L(\H)$ denote the space of linear operators acting on $\H$. Quantum channels are completely positive and trace preserving maps from $L(\H_A)$ to $L(\H_B)$ and denoted by $\N_{A\to B}$. An isometric extension or Stinespring dilation $U: \H_A \to \H_B \otimes \H_E$ of a quantum channel $\N_{A\to B}$ is a linear isometry such that for all $\rho_A \in L(\H_A)$, the following holds: $\Tr_E[U \rho_A U^{\dagger}]= \N(\rho_A)$. A complementary channel $\N^c_{A\to E}$ of $\N_{A\to B}$ is defined as $\N^c_{A\to E}(\rho_A) = \Tr_B[U \rho_A U^{\dagger}]$. The Choi state of a quantum channel $\N_{A\to B}$ is given by
	\begin{align}\label{eq:choi-state}
		\rho_{AB}^{\N} \equiv (\id_A \otimes \N_{A'\to B}) \left( \Phi_{AA'}^+ \right),
	\end{align}
	where $\Phi_{AA'}^+$ denotes the maximally entangled state, i.e.,
	\begin{equation}
		\Phi_{AA'}^+ \equiv \frac{1}{d_A}\sum_{i,i'=1}^{d_A} \ket{i}\bra{i^{\prime}}_A \otimes \ket{i}\bra{i^{\prime}}_{A'}.
	\end{equation}
	We let
	\begin{equation}
		\Gamma_{AB}^{\mathcal{N}}\equiv d_A\rho_{AB}^{\mathcal{N}}
	\end{equation}
	denote the Choi matrix of the channel $\mathcal{N}$.
	
	According to the Choi-Kraus theorem, the action of a quantum channel $\N_{A\to B}$ on any $X_A \in L(\H_A)$ can be represented in the following way:
	\begin{align}
		\N_{A\to B}(X_A) = \sum_{i=1}^r V_i X_A V_i^{\dagger}~,
	\end{align}
	where the so-called Kraus operators $V_i:\H_A\to\H_B$, $i \in \{1, \dots, r   \}$, satisfy $\sum_{i=1}^r V^{\dagger}_i V_i = \mathbbm{1}_A$, and $r$ need not exceed $d_Ad_B$, with a minimal choice being $r = \text{rank}(\Gamma_{AB}^{\mathcal{N}})$.

A quantum channel $\N_{A\to B}$ is entanglement breaking if the Choi state as in \eqref{eq:choi-state} of the channel is separable \cite{HSR03}. 

	A quantum channel $\mathcal{N}_{A\to B}$ is called \textit{degradable} if there exists a channel $\mathcal{D}_{B\to E}$ such that 
	\begin{equation}
		(\mathcal{D}_{B\to E}\circ\mathcal{N}_{A\to B})(X_A ) =\mathcal{N}^c_{A \to E}(X_A),
	\end{equation}
for all $X_A \in L(\H_A)$ \cite{DS04}. A channel $\mathcal{N}_{A\to B}$ is called \textit{anti-degradable} if its complementary channel $\mathcal{N}^c_{A\to E}$ is degradable, i.e., if there exists a channel $\mathcal{E}_{E \to B}$ such that
	\begin{equation}\label{eq-anti_degrade}
		(\mathcal{E}_{E\to B}\circ\mathcal{N}^c_{A\to E})(X_A)=\mathcal{N}_{A\to B}(X_A)
	\end{equation}
	for all $X_A \in L(\H_A)$ \cite{CG06}.
	
	For any Hermiticity-preserving map $\mathcal{M}_{A\to B}$, its diamond norm $\norm{\mathcal{M}}_{\diamond}$ is defined as \cite{Kit97}
	\begin{equation}\label{eq-diamond_norm}
		\norm{\mathcal{M}}_{\diamond}=\max_{\psi_{RA}}\norm{\mathcal{M}_{A\to B}(\psi_{RA})}_1,
	\end{equation}
	where the optimization is over all pure states $\psi_{RA}$, with the dimension of the reference system $R$ equal to the dimension of $A$, and $\norm{X}_1$ denotes the trace norm of the matrix $X$, which is defined as the sum of the singular values of $X$.

\subsection{Capacities of quantum channels}\label{sec:capacity-of-qc}

	For any quantum channel $\mathcal{N}$, its classical capacity $C(\mathcal{N})$ is defined to be the highest rate at which classical information can be sent over many uses of the channel with an error probability that converges to zero as the number of channel uses increases. It holds that \cite{Hol73,SW97,Hol98}
	\begin{equation}\label{eq-classical_capacity}
		C(\mathcal{N})=\lim_{n\to\infty}\frac{1}{n}\chi(\mathcal{N}^{\otimes n}),
	\end{equation}
	where $\chi(\mathcal{N})$ is the Holevo information of the channel $\mathcal{N}$, which is defined as
	\begin{equation}
		\chi(\mathcal{N})=\max_{\rho_{XA}}I(X;B)_{\omega},
		\label{eq:Holevo-info-channel}
	\end{equation}
	where $\omega_{XB}=\mathcal{N}_{A\to B}(\rho_{XA})$, and the maximization is with respect to all classical-quantum states, i.e., states of the form
	\begin{equation}\label{eq-cq_state}
		\rho_{XA}\equiv\sum_xp_X(x)\ket{x}\bra{x}_X\otimes\rho_A^x.
	\end{equation}

	For any quantum channel $\mathcal{N}$, its quantum capacity $Q(\mathcal{N})$ is defined to be the highest rate at which quantum information can be sent over many uses of the channel with a fidelity that converges to one as the number of channel uses increases. It has been shown \cite{BS96,SN96,BNS98,BKN20,Llo97,capacity2002shor,Dev05} that
	\begin{equation}
		Q(\mathcal{N})=\lim_{n\to\infty}\frac{1}{n}I_{\text{c}}(\mathcal{N}^{\otimes n}),
	\end{equation}
	where the function $I_{\text{c}}$ is the channel coherent information, which is defined for any quantum channel $\mathcal{N}$ as
	\begin{equation}\label{eq:coherent-info-q-chan}
		I_{\text{c}}(\mathcal{N})\equiv\max_{\rho}I_{\text{c}}(\rho,\mathcal{N}),
	\end{equation}
	where $\rho \in D(\H)$, and 
	\begin{equation}
		I_{\text{c}}(\rho,\mathcal{N})\equiv H(\mathcal{N}(\rho))-H(\mathcal{N}^c(\rho)).
	\end{equation}
	If the channel $\mathcal{N}$ is anti-degradable \cite{CG06}, then its coherent information in \eqref{eq:coherent-info-q-chan} vanishes, which means that anti-degradable channels have zero quantum capacity.

	The private capacity $P(\N)$ of a quantum channel $\N$ is defined to be the maximum rate at which a sender can reliably communicate classical messages to a receiver by using the channel many times, such that the environment of the channel obtains negligible information about the transmitted message. The private capacity $P(\N)$ is equal to the regularized private information of the channel $\N$ \cite{Dev05, CWY04}, i.e.,
	\begin{align}
		P(\N) = \lim_{n\to\infty}\frac{1}{n} P^{(1)}(\N^{\otimes n})~,
	\end{align}
	where the private information of the channel is defined as
	\begin{align}\label{eq-priv_inf}
		P^{(1)}(\N) \equiv \max_{\rho_{XA}}\bigg[ I(X;B)_{\omega} - I(X;E)_{\omega}\bigg].
	\end{align}
	The maximization here is with respect to all states $\rho_{XA}$ as in \eqref{eq-cq_state}, and $\omega_{XABE}=\mathcal{U}^{\mathcal{N}}_{A\to BE}(\rho_{XA})$, with $\mathcal{U}^{\mathcal{N}}_{A\to BE}$ being an isometric channel extending $\mathcal{N}$.
	
	In general, the quantum and private capacities of a channel $\mathcal{N}$ are related as follows \cite{Dev05}:
	\begin{equation}\label{eq-quant_priv_cap_ineq}
		Q(\mathcal{N})\leq P(\mathcal{N}).
	\end{equation}
	For degradable channels $\N$ and $\M$, the coherent information is known to be additive \cite{DS04} in the following sense:
	\begin{align}
	I_{\text{c}}(\N \otimes \M) = I_{\text{c}} (\N)+I_{\text{c}} (\M).
	\end{align}
	Moreover, the private information of a degradable channel is equal to its coherent information \cite{GS08}. Therefore, both the quantum and private capacities of a degradable channel are given by its coherent information.
	
	Two-way assisted communication capacities are defined as the highest achievable rate of communication for protocols involving local operations by the sender and receiver and classical communication in both directions between the sender and receiver \cite{BGPSSW96,BDSW} (see also \cite{TGW14a}). We denote the two-way assisted quantum and private capacities of a quantum channel $\mathcal{N}$ by $Q^{\leftrightarrow}(\mathcal{N})$ and $P^{\leftrightarrow}(\mathcal{N})$, respectively. As in the unassisted case, we have that $Q^{\leftrightarrow}(\mathcal{N})\leq P^{\leftrightarrow}(\mathcal{N})$ for all quantum channels $\mathcal{N}$.
	
	Since any one-way, or unassisted, communication protocol is a special case of a two-way assisted communication protocol, we immediately have the coherent information lower bound $Q^{\leftrightarrow}(\mathcal{N})\geq I_{\text{c}}(\mathcal{N})$. Another known lower bound is the reverse coherent information \cite{HHH00,DW05,DJKR06}, which is defined as
	\begin{equation}
		I_{\text{rc}}(\mathcal{N})\equiv \max_{\rho} I_{\text{rc}}(\rho,\mathcal{N}),\label{eq:reverse-coh-inf}
	\end{equation}
	where
	\begin{equation}
		I_{\text{rc}}(\rho,\mathcal{N})\equiv H(\rho)-H(\mathcal{N}^c(\rho)).
	\end{equation}

	The reverse coherent information as in \eqref{eq:reverse-coh-inf} was defined in \cite{HHH00} and shown in \cite{HHH00,DW05} to be a lower bound on the two-way assisted quantum capacity. It was proven to be additive in \cite{DJKR06}, and concavity in the input state $\rho$ was shown in \cite[Eq.~(8.48)]{MH06}.

\subsection{Bounds on the capacities of quantum channels}\label{sec-bounds_QCap}

	In this section, we recall several different techniques for placing upper bounds on the communication capacities of a quantum channel that we use throughout the rest of the paper.

\subsubsection{Data-processing upper bounds}
	
	Let $\N \circ \M$ denote the serial concatenation of two quantum channels $\N$ and $\M$. Upper bounds on the quantum capacity of the channel $\N\circ \M$ can be established as follows~\cite{WP07,SS08}:
	\begin{align}
		Q(\mathcal{N}\circ\mathcal{M})&\leq Q(\mathcal{M}),\label{eq-QCap_data_proc_1}\\
		Q(\mathcal{N}\circ\mathcal{M})&\leq Q(\mathcal{N}).\label{eq-QCap_data_proc_2}
	\end{align}
	The first inequality follows from definitions and the quantum data processing inequality. The second inequality is a consequence of the following argument: consider an arbitrary encoding and decoding scheme for quantum communication over the channel $\N\circ \M$. Then this encoding, followed by many uses of the channel $\M$, can be considered as an encoding for the channel $\N$. Since the quantum capacity of the channel $\N$ involves an optimization over all such encodings, the desired inequality follows.
	
	By similar reasoning as above, we can conclude analogous data-processing upper bounds for the private capacity and the classical capacity:
	\begin{align}
		P(\mathcal{N}\circ\mathcal{M})&\leq P(\mathcal{M}),\label{eq-PCap_data_proc_1}\\
		P(\mathcal{N}\circ\mathcal{M})&\leq P(\mathcal{N})\label{eq-PCap_data_proc_2},\\
		C(\mathcal{N}\circ\mathcal{M})&\leq C(\mathcal{M}),\\
		C(\mathcal{N}\circ\mathcal{M})&\leq C(\mathcal{N}).
	\end{align}

\subsubsection{Classical capacity upper bounds via approximate entanglement breakability and approximate covariance}\label{subsubsec-CCap_UB_EB_cov}

	Upper bounds on the classical capacity of any quantum channel have been obtained using the notions of approximate entanglement-breakability and approximate covariance of channels  \cite{LKDW18}. We now summarize these results. All of these results, as well as their proofs, can be found in \cite{LKDW18}.
	
	A quantum channel $\mathcal{N}$ is called $\varepsilon$-entanglement-breaking if there exists an entanglement-breaking channel $\mathcal{M}$ such that $\frac{1}{2}\norm{\mathcal{N}-\mathcal{M}}_{\diamond}\leq\varepsilon$. We let
	\begin{equation}\label{eq-approx_EB}
		\varepsilon_{\text{EB}}(\mathcal{N})\equiv\min_{\mathcal{M}}\left\{\frac{1}{2}\norm{\mathcal{N}-\mathcal{M}}_{\diamond}:\mathcal{M}\text{ entanglement breaking}\right\}
	\end{equation}
	denote the smallest $\varepsilon$ such that $\mathcal{N}$ is $\varepsilon$-entanglement-breaking. For qubit-to-qubit channels, the entanglement-breaking parameter $\varepsilon_{\text{EB}}(\mathcal{N})$ can be calculated by means of a semi-definite program \cite[Lemma~III.8]{LKDW18}. We suppress the channel dependence on $\varepsilon_{\text{EB}}$ if the channel is understood from the context.
	
	For any $\varepsilon$-entanglement-breaking channel $\mathcal{N}$, the following upper bound on the classical capacity $C(\mathcal{N})$ holds \cite[Corollary~III.7]{LKDW18}:
	\begin{equation}\label{eq-CCap_UB_approx_EB}
		C(\mathcal{N})\leq\chi(\mathcal{M})+2\varepsilon\log_2d_B+g(\varepsilon),
	\end{equation}
	where $\mathcal{M}$ is the entanglement-breaking channel such that $\varepsilon=\frac{1}{2}\norm{\mathcal{N}-\mathcal{M}}_{\diamond}$.
	
	We now define the notion of approximate covariance of a quantum channel $\mathcal{N}_{A\to B}$. Let $G$ be a finite group with a unitary representation $\{U_A(g)\}_{g\in G}$ on the input system $A$ and a unitary representation $\{V_B(g)\}_{g\in G}$ on the output system $B$. The so-called twirled channel $\mathcal{N}_{A\to B}^G$ is defined as
	\begin{equation}
		\mathcal{N}_{A\to B}^G(\cdot)\equiv \frac{1}{|G|}\sum_{g\in G} V_B(g)^\dagger\mathcal{N}_{A\to B}(U_A(g)(\cdot)U_A(g)^\dagger)V_B(g). \label{eq:twirled-channel}
	\end{equation}
	Note that the twirled channel $\mathcal{N}_{A\to B}^G$ can be realized by means of a generalized teleportation protocol \cite[Appendix~B]{KW17}.
	By construction, this channel is covariant with respect to the representations $\{U_A(g)\}_{g\in G}$ and $\{V_B(g)\}_{g\in G}$, meaning that
	\begin{equation}
		\mathcal{N}_{A\to B}^G(U_A(g)\rho_AU_A(g)^\dagger)=V_B(g)\mathcal{N}_{A\to B}^G(\rho_A)V_B(g)^\dagger
	\end{equation}
	for all states $\rho_A$ and all $g\in G$. We call $\mathcal{N}$ $\varepsilon$-covariant with respect to the representations $\{U_A(g)\}_{g\in G}$, $\{V_B(g)\}_{g\in G}$ if $\frac{1}{2}\norm{\mathcal{N}-\mathcal{N}^G}_{\diamond}\leq\varepsilon$. We let
	\begin{equation}
		\varepsilon_{\text{cov}}(\mathcal{N})\equiv\frac{1}{2}\norm{\mathcal{N}-\mathcal{N}^G}_{\diamond}
	\end{equation}
	denote the smallest $\varepsilon$ such that $\mathcal{N}$ is $\varepsilon$-covariant. The covariance parameter $\varepsilon_{\text{cov}}(\mathcal{N})$ can be computed by means of a semi-definite program, as observed in \cite{LKDW18}, due to the fact that the diamond norm can be computed by a semi-definite program \cite{W13}. We suppress the dependence of the covariance parameter on both the group and its representations for simplicity, and if it is clear from the context, we also suppress the dependence on the channel.
	
	Let $\mathcal{N}$ be a qubit-to-qubit channel, and let $G=\mathbb{Z}_2\times\mathbb{Z}_2$, with $\mathbb{Z}_2$ the group consisting of the set $\{0,1\}$ with addition modulo two. This group has the (projective) unitary representation consisting of the Pauli operators $\{\mathbbm{1},\sigma_x,\sigma_y,\sigma_z\}$. With this group and this representation, if $\mathcal{N}$ is $\varepsilon$-covariant, then \cite[Corollary III.5]{LKDW18}
	\begin{equation}\label{eq-CCap_UB_approx_cov}
		C(\mathcal{N})\leq\chi(\mathcal{N}^G)+2\varepsilon+g(\varepsilon).
	\end{equation}

\subsubsection{Quantum and private capacity upper bounds via approximate degradability and approximate anti-degradability}

	We now recall techniques to obtain upper bounds on the quantum and private capacities of a quantum channel using the concepts of approximate degradability and approximate anti-degradability. These concepts were developed in \cite{SSWR17}. All of the results stated in this subsection, as well as their proofs, can be found in \cite{SSWR17}.

	A channel $\mathcal{N}$ is called \textit{$\varepsilon$-degradable} if there exists a channel $\mathcal{D}$ such that $\frac{1}{2}\norm{\mathcal{N}^c-\mathcal{D}\circ\mathcal{N}}_{\diamond}\leq\varepsilon$. We let
	\begin{equation}\label{eq-approx_deg_parameter}
		\varepsilon_{\text{deg}}(\mathcal{N})\coloneqq\min_{\mathcal{D}}\left\{\frac{1}{2}\norm{\mathcal{N}^c-\mathcal{D}\circ\mathcal{N}}_{\diamond}:\mathcal{D}\text{ is a channel}\right\}
	\end{equation}
	denote the smallest $\varepsilon$ such that $\mathcal{N}$ is $\varepsilon$-degradable. We suppress the dependence of this quantity on the channel if it is clear from the context. Note that $\varepsilon_{\text{deg}}(\mathcal{N})$ can be calculated via a semi-definite program.

	For an $\varepsilon$-degradable channel $\mathcal{N}$ with corresponding (approximate) degrading channel $\mathcal{D}$, it holds that \cite[Theorem~7]{SSWR17}
	\begin{align}
		Q(\mathcal{N})&\leq U_{\mathcal{D}}(\mathcal{N})+2\varepsilon\log_2d_E+g(\varepsilon)\label{eq-approx_deg_UB},
	\end{align}
	where the quantity $U_{\mathcal{D}}(\mathcal{N})$ is defined as
	\begin{equation}
	U_{\mathcal{D}}(\mathcal{N})\equiv \max_{\rho}\{H(F|\tilde{E})_{\omega}:\omega_{\tilde{E}FE}=(W\otimes\mathbbm{1}_E)V\rho_A V^\dagger(W\otimes\mathbbm{1}_E)^\dagger\},
	\end{equation}
	with $V:\mathcal{H}_A\to\mathcal{H}_B\otimes\mathcal{H}_E$ and $W:\mathcal{H}_B\to\mathcal{H}_{\tilde{E}}\otimes\mathcal{H}_F$ being isometric extensions of channels $\mathcal{N}$ and $\mathcal{D}$, respectively. Moreover,  the following bound was established on the private capacity of an $\varepsilon$-degradable channel $\N$ in \cite[Theorem~13]{SWAT}:
	\begin{align}
	P(\mathcal{N})&\leq U_{\mathcal{D}}(\mathcal{N})+6\varepsilon\log_2d_E+3g(\varepsilon).\label{eq-approx_deg_UB_PCap}
	\end{align}

	Another upper bound on the quantum capacity of a quantum channel $\N$ can be established using the notion of $\varepsilon$-close degradability. A channel $\mathcal{N}$ is called \textit{$\varepsilon$-close-degradable} if there exists a degradable channel $\mathcal{M}$ such that $\frac{1}{2}\norm{\mathcal{N}-\mathcal{M}}_{\diamond}\leq\varepsilon$. If $\mathcal{N}$ is an $\varepsilon$-close-degradable channel, then the following bounds hold \cite[Proposition A2]{SSWR17}:
	\begin{align}
		Q(\mathcal{N})&\leq I_{\text{c}}(\mathcal{M})+2\varepsilon\log_2d_B+2g(\varepsilon),\label{eq-eps_approx_deg_UB}\\
		P(\mathcal{N})&\leq I_{\text{c}}(\mathcal{N})+4\varepsilon\log_2d_B+4g(\varepsilon).\label{eq-eps_approx_deg_UB_PCap}
	\end{align}

	A channel $\mathcal{N}$ is called an \textit{$\varepsilon$-anti-degradable} channel if there exists a channel $\mathcal{E}$ such that $\frac{1}{2}\norm{\mathcal{N}-\mathcal{E}\circ\mathcal{N}^c}_{\diamond}\leq\varepsilon$. We let
	\begin{equation}\label{eq-approx_adeg_parameter}
		\varepsilon_{\text{a-deg}}(\mathcal{N})\equiv\min_{\mathcal{E}}\left\{\frac{1}{2}\norm{\mathcal{N}-\mathcal{E}\circ\mathcal{N}^c}_{\diamond}:\mathcal{E}\text{ is a channel}\right\}
	\end{equation}
	denote the smallest $\varepsilon$ such that $\mathcal{N}$ is $\varepsilon$-anti-degradable. We suppress the dependence of this quantity on the channel if it is clear from the context. Note that $\varepsilon_{\text{a-deg}}(\mathcal{N})$ can be calculated via a semi-definite program.
	
	For any $\varepsilon$-anti-degradable channel $\mathcal{N}$, it holds that \cite[Theorem 11]{SSWR17}
	\begin{align}
		Q(\mathcal{N})&\leq P(\mathcal{N})\leq \varepsilon\log_2(d_B-1)+2\varepsilon\log_2d_B\nonumber\\
		&\qquad\qquad\qquad+h_2(\varepsilon)+g(\varepsilon).\label{eq-eps_anti_degrade_upper_bound}
	\end{align}

\subsubsection{Rains information upper bound on quantum capacity and relative entropy of entanglement upper bound on private capacity}

The Rains information of a quantum channel is an upper bound on its quantum capacity \cite{TWW17}, and a channel's relative entropy of entanglement is an upper bound on its private capacity \cite{WTB17}. Here we briefly recall these results.

The Rains relative entropy $R(A;B)_{\rho}$ \cite{R01,AMVW02} and the relative entropy of entanglement $E_R(A;B)_{\rho}$ \cite{VP98} of a bipartite state $\rho_{AB}$ are defined as
	\begin{align}
		R(A;B)_{\rho}&\equiv\min_{\sigma_{AB}\in\text{PPT}'(A:B)}D(\rho_{AB}\Vert\sigma_{AB}), \label{eq:Rains-state}\\
		E_R(A;B)_{\rho}&\equiv\min_{\sigma_{AB}\in\text{SEP}(A:B)}D(\rho_{AB}\Vert\sigma_{AB}), \label{eq:REE-state}
	\end{align}
	where $D(\rho_{AB}\Vert\sigma_{AB})$ is the quantum relative entropy of $\rho_{AB}$ and $\sigma_{AB}$ \cite{Ume62}. We have $D(\rho_{AB}\Vert\sigma_{AB})=\Tr[\rho(\log_2\rho-\log_2\sigma)]$ if $\text{supp}(\rho_{AB})\subset\text{supp}(\sigma_{AB})$, and $D(\rho_{AB}\Vert\sigma_{AB})=+\infty$ otherwise. Also, $\text{PPT}'(A\!:\!B)$ denotes the set $\{ \sigma_{AB} : \sigma_{AB} \geq 0, \Vert \sigma_{AB}^{\t_B} \Vert_1 \leq 1\}$ \cite{AMVW02}, and $\text{SEP}(A\!:\!B)$ denotes the set of separable states acting on $\mathcal{H}_A\otimes\mathcal{H}_B$ \cite{W89}. Note that one can efficiently calculate the Rains relative entropy by employing convex programming methods \cite{FF18,FWTD17,Wilde2018}, due to the fact that the constraints $\sigma_{AB} \geq 0$ and $ \Vert \sigma_{AB}^{\t_B} \Vert_1 \leq 1$ are semi-definite constraints.
	
	For any channel $\mathcal{N}_{A'\to B}$, we define its Rains relative entropy $R(\mathcal{N})$ and its relative entropy of entanglement $E_R(\mathcal{N})$ as follows:
\begin{align}
R(\N)& \equiv \max_{\phi_{AA'}} R(A;B)_{\rho}, \label{eq:Rains-channel}\\
E_R(\N)& \equiv \max_{\phi_{AA'}} E_R(A;B)_{\rho}, \label{eq:REE-channel}
\end{align}
where $\rho_{AB} \equiv \N_{A'\to B}(\phi_{AA'})$ and the optimization is with respect to all pure bipartite input states $\phi_{AA'}$, with the dimension of $A$ equal to the dimension of the input system $A'$ of the channel $\mathcal{N}$. As stated above, the information measures $R(\mathcal{N})$ and $E_R(\mathcal{N})$ are useful because they bound the quantum and private capacities, respectively, of the channel $\mathcal{N}$:
\begin{align}
Q(\N) & \leq R(\N), \label{eq:Rains-q-cap-bound}\\
P(\N) & \leq E_R(\N). \label{eq:REE-p-cap-bound}
\end{align}

By following an approach similar to that given in \cite[Proposition~2]{TWW17}, it follows that the maximizations in \eqref{eq:Rains-channel} and \eqref{eq:REE-channel} are concave in the reduced density operator $\Tr_A[\phi_{AA'}]$:

\begin{proposition}\label{prop:concavity-Rains}
Let $\mathcal{N}_{A^{\prime}\rightarrow B}$ be a quantum channel,
$\rho_{A^{\prime}}$ a state, $\phi_{AA^{\prime}}^{\rho}$ a purification of
$\rho_{A^{\prime}}$, and $\omega_{AB}\equiv\mathcal{N}_{A^{\prime}\rightarrow
B}(\phi_{AA^{\prime}}^{\rho})$. Then, the functions $\rho_{A^{\prime}%
}\mapsto R(A;B)_{\omega}$ and $\rho_{A^{\prime}}\mapsto E_{R}%
(A;B)_{\omega}$ are concave in the reduced state $\operatorname{Tr}_{A}%
[\phi_{AA^{\prime}}^{\rho}]=\rho_{A^{\prime}}$, regardless of which
purification $\phi_{AA^{\prime}}^{\rho}$ of $\rho_{A^{\prime}}$ is chosen.
\end{proposition}

We give a proof of Proposition~\ref{prop:concavity-Rains} in Appendix~\ref{proof-prop:concavity-Rains}. Proposition~\ref{prop:concavity-Rains}, combined with the results of \cite{FF18,FWTD17,Wilde2018}, implies that $R(\N)$ can be computed efficiently by convex programming techniques. One can effectively use convex programming techniques to calculate $E_R(\N)$, but it will not be efficient to do so in general since it is well known that optimizing over the set of separable states is difficult \cite{G03,L07,G10}.

For qubit-qubit systems $AB$, it is known that $R(A;B)_{\rho} = E_R(A;B)_{\rho}$ \cite{AS08}, which is related to the fact that the positive partial transposition criterion is necessary and sufficient for separability for such low-dimensional systems \cite{Peres96,HHH96}. (However, note that the analysis in \cite{AS08} goes well beyond this observation in order to establish the aforementioned equality.) This equality in turn implies that $R(\N)= E_R(\N)$ for qubit-to-qubit channels, which is useful for our purposes here since our focus is the qubit-to-qubit generalized amplitude damping channel.

\subsubsection{Upper bounds on two-way assisted quantum and private capacities}

	The squashed entanglement \cite{CW04} (see also \cite{RT99, RT02}) of a bipartite state $\rho_{AB}$ is defined as
	\begin{equation}
		E_{\textnormal{sq}}(A;B)_\rho=\frac{1}{2}\inf\{I(A;B|E)_\omega:\Tr_E[\omega_{ABE}]=\rho_{AB}\},
		\label{eq:squashed-E}
	\end{equation}
	where
	\begin{equation}\label{eq-QCMI}
		\begin{aligned}
		I(A;B|E)&\equiv H(A|E)+H(B|E)-H(AB|E)\\
		&=H(AE)+H(BE)-H(E)-H(ABE)
		\end{aligned}
	\end{equation}
	is the quantum conditional mutual information. Whether the infimum in \eqref{eq:squashed-E} can be replaced with a minimum is one of the outstanding challenges in quantum information theory.
	
	An alternative way of writing the squashed entanglement is to use the fact that for any extension $\omega_{ABE}$ of a state $\rho_{AB}$ there exists a channel $\mathcal{S}$ acting on a purification $\ket{\psi}_{ABE'}$ such that $\mathcal{S}_{E'\to E}(\ket{\psi}\bra{\psi}_{ABE'})=\omega_{ABE}$. This leads to the following alternative expression for $E_{\text{sq}}(A;B)_\rho$:
	\begin{multline}
		E_{\text{sq}}(A;B)_\rho\\=\frac{1}{2}\inf_{\mathcal{S}}\{I(A;B|E)_\omega:\omega_{ABE}=\mathcal{S}_{E'\to E}(\ket{\psi}\bra{\psi}_{ABE'})\},
	\end{multline}
	where $\ket{\psi}_{ABE'}$ is a purification of $\rho_{AB}$. The channels $\mathcal{S}$ over which we optimize are called \textit{squashing channels}.

	The squashed entanglement of a channel $\mathcal{N}$ \cite{TGW14a,TGW14b} is defined as
	\begin{equation}
		E_{\textnormal{sq}}(\mathcal{N})\equiv\max_{\phi_{AA'}}E_{\textnormal{sq}}(A;B)_{\rho},
	\end{equation}
	where $\rho_{AB}=\mathcal{N}_{A'\to B}(\phi_{AA'})$ and where the optimization is over all pure states $\phi_{AA'}$, with $A$ having the same dimension as the dimension of the input system $A'$ of the channel $\mathcal{N}$.
	
	For any channel $\mathcal{N}$, the following bounds hold \cite{TGW14a,TGW14b} (see also \cite{MMW16} for \eqref{eq:2-way-bound-private-cap}):
	\begin{align}
		Q^{\leftrightarrow}(\mathcal{N})&\leq E_{\text{sq}}(\mathcal{N}),\\
		P^{\leftrightarrow}(\mathcal{N})&\leq E_{\text{sq}}(\mathcal{N}). \label{eq:2-way-bound-private-cap}
	\end{align}
	
	By taking the identity squashing channel, and using the fact that $I(A;B|E)_\psi=I(A;B)_\rho$ for any pure state $\psi_{ABE}$, where $\rho_{AB}=\Tr_E[\ket{\psi}\bra{\psi}_{ABE}]$, we get that $E_{\text{sq}}(A;B)_\rho\leq\frac{1}{2}I(A;B)_\rho$ for all states $\rho_{AB}$. This implies that $E_{\text{sq}}(\mathcal{N})\leq\frac{1}{2}\max_{\phi_{AA'}}I(A;B)_{\rho}=\frac{1}{2}I(\mathcal{N})$, where $\rho_{AB}=\mathcal{N}_{A'\to B}(\phi_{AA'})$. In other words, the squashed entanglement of any channel is always bounded from above by half the mutual information of the channel. Therefore, we have
	\begin{equation}\label{eq-Q_cap_two_way_MI}
		Q^{\leftrightarrow}(\mathcal{N})\leq \frac{1}{2}I(\mathcal{N})
	\end{equation}
	for all channels $\mathcal{N}$ \cite{TGW14a, TGW14b, GEW16}.
	
	The max-Rains relative entropy of a bipartite state $\rho_{AB}$ is defined as \cite{WD16} (see also \cite{TWW17})
	\begin{equation}
		R_{\max}(A;B)_\rho\equiv\min_{\sigma_{AB}\in\text{PPT}'(A:B)}D_{\max}(\rho_{AB}\Vert\sigma_{AB}),
	\end{equation}
	where, as stated before, the set $\text{PPT}'(A\!:\!B)$ is defined as \cite{AMVW02}
	\begin{equation}
		\text{PPT}'(A\!:\!B)\equiv \{\sigma_{AB}:\sigma_{AB}\geq 0,~\norm{\sigma_{AB}^{\t_B}}_1\leq 1\},
	\end{equation}
	and the max-relative entropy $D_{\max}(\rho_{AB}\Vert\sigma_{AB})$ is defined as \cite{Datta08}
	\begin{equation}\label{eq-D_max}
		D_{\max}(\rho_{AB}\Vert\sigma_{AB})=\log_2\min_t\{t:\rho_{AB}\leq t\sigma_{AB}\}.
	\end{equation}
	The max-Rains information $R_{\max}(\mathcal{N})$ of a channel $\mathcal{N}$ is defined as \cite{WFD18} (see also \cite{TWW17})
	\begin{equation}\label{eq-max_Rains_channel}
		R_{\max}(\mathcal{N})\equiv \max_{\phi_{AA'}}R_{\max}(A;B)_{\rho},
	\end{equation}
	where $\rho_{AB}=\mathcal{N}_{A'\to B}(\phi_{AA'})$, and the optimization is over pure states $\phi_{AA'}$, with the dimension of $A$ the same as that of the input system $A'$ of the channel $\mathcal{N}$. It satisfies \cite{BW18}
	\begin{equation}\label{eq-max_Rains_bound}
		Q^{\leftrightarrow}(\mathcal{N})\leq R_{\max}(\mathcal{N}).
	\end{equation}
	Furthermore, it is a strong converse rate. As shown in \cite{WFD18}, it holds that
	\begin{equation}\label{eq-R_max}
		\begin{aligned}
		R_{\max}(\mathcal{N})&=\log_2\Delta(\mathcal{N}),\\
		\Delta(\mathcal{N})&=\left\{\begin{array}{l l} \text{min.} & \norm{\Tr_B[V_{AB}+Y_{AB}]}_\infty \\ \text{subject to} & Y_{AB}\geq 0, V_{AB}\geq 0,\\ & (V_{AB}-Y_{AB})^{\t_B}\geq \Gamma_{AB}^{\mathcal{N}}, \end{array}\right.
		\end{aligned}
	\end{equation}
	where $\Gamma_{AB}^{\mathcal{N}}$ is the Choi matrix of the channel $\mathcal{N}$, and $\norm{X}_{\infty}$ denotes the spectral norm of the matrix $X$, which is defined as the largest singular value of $X$. In particular, the quantity $\Delta(\mathcal{N})$ is given by an SDP.

	For the two-way assisted private capacity, we consider the following general strong converse upper bound~\cite{CMH17}:
	\begin{equation}
		P^{\leftrightarrow}(\mathcal{N})\leq E_{\text{max}}(\mathcal{N}),
	\end{equation}
	which holds for any channel $\mathcal{N}$. The quantity $E_{\max}(\mathcal{N})$ is the max-relative entropy of entanglement of $\mathcal{N}$, which is defined as \cite{CMH17}
	\begin{equation}
		E_{\max}(\mathcal{N})\equiv \max_{\phi_{AA'}}E_{\max}(A;B)_{\rho},
	\end{equation}
	where $\rho_{AB}=\mathcal{N}_{A'\to B}(\phi_{AA'})$, and the optimization is over pure states $\phi_{AA'}$, with the dimension of $A$ equal to the dimension of the input system $A'$ of the channel $\mathcal{N}$. The max-relative entropy of entanglement $E_{\max}(A;B)_{\rho}$ of any bipartite state $\rho_{AB}$ is defined as \cite{Datta08}
	\begin{equation}
		E_{\max}(A;B)_{\rho}\equiv \min_{\sigma_{AB}\in\text{SEP}(A:B)}D_{\max}(\rho_{AB}\Vert\sigma_{AB}),
	\end{equation}
	where $\text{SEP}(A\!:\!B)$ is the set of separable states acting on the space $\mathcal{H}_A\otimes\mathcal{H}_B$. It has been shown in \cite{BW18} that, for qubit-to-qubit channels, the quantity $E_{\max}(\mathcal{N})$ can be written as the solution to an SDP as follows:
	\begin{equation}\label{eq-E_max_SDP_primal}
		\begin{aligned}
		E_{\max}(\mathcal{N})&=\log_2 \Sigma(\mathcal{N}),\\
		\Sigma(\mathcal{N})&=\left\{\begin{array}{l l}\text{min}. & \norm{\Tr_B[Y_{AB}]}_{\infty} \\ 
		\text{subject to} & \Gamma_{AB}^{\mathcal{N}}\leq Y_{AB},\\[0.1cm]
		& Y_{AB}^{\t_B}\geq 0. \end{array}\right.
		\end{aligned}
	\end{equation}
	Using the fact that $\text{PPT}\subset\text{PPT}'$, we obtain $R_{\max}(A;B)_{\rho}\leq E_{\max}(A;B)_{\rho}$ for all states $\rho_{AB}$, which implies that
	\begin{equation}\label{eq-Rmax_Emax_ineq}
		R_{\max}(\mathcal{N})\leq E_{\max}(\mathcal{N})
	\end{equation}
	for any quantum channel $\mathcal{N}$.
	
	In \cite{KW17}, the following bounds on the two-way assisted capacities were established for a channel $\mathcal{N}$\ that is $\varepsilon$-approximately covariant (see Sec.~\ref{subsubsec-CCap_UB_EB_cov} for the definition):%
	\begin{align}
		Q^{\leftrightarrow}(\mathcal{N})  & \leq R(A;B)_{\rho}+2\varepsilon\log_{2}d_{B}+g(\varepsilon),\label{eq-Q_cap_two_way_cov}\\
		P^{\leftrightarrow}(\mathcal{N})  & \leq E_{R}(A;B)_{\rho}+2\varepsilon\log_{2}d_{B}+g(\varepsilon),\label{eq-P_cap_two_way_cov}
	\end{align}
	where $\rho_{AB}=\mathcal{N}_{A'\rightarrow B}^{G}(\Phi_{AA'}^+)$ and the twirled channel $\mathcal{N}_{A'\rightarrow B}^{G}$ is defined in \eqref{eq:twirled-channel}.

\subsection{The generalized amplitude damping channel}

	The generalized amplitude damping channel (GADC) $\mathcal{A}_{\gamma,N}$ is a qubit-to-qubit channel with the following four Kraus operators (in the standard basis) \cite{NC10}:
	\begin{align}
		A_1&=\sqrt{1-N}\left(\ket{0}\bra{0}+\sqrt{1-\gamma}\ket{1}\bra{1}\right),\\
		A_2&=\sqrt{\gamma(1-N)}\ket{0}\bra{1},\\
		A_3&=\sqrt{N}\left(\sqrt{1-\gamma}\ket{0}\bra{0}+\ket{1}\bra{1}\right),\\
		A_4&=\sqrt{\gamma N}\ket{1}\bra{0}.
	\end{align}
	It is completely positive and trace preserving for all $\gamma, N \in [0, 1]$. If we set $N=0$, then the GADC reduces to the ordinary amplitude damping channel $\mathcal{A}_\gamma$ with two Kraus operators. The GADC also has only two Kraus operators for $N=1$, in which case the channel behaves as an amplification process, driving the signal toward the state $\ket{1}\bra{1}$. 

	Let $\rho$ denote a single-qubit density operator:
	\begin{equation}
		\rho=\frac{1}{2}(\mathbbm{1}+r_x\sigma_x+r_y\sigma_y+r_z\sigma_z), 
	\end{equation}
	where $\vec{r}\equiv(r_x,r_y,r_z)\in\mathbb{R}^3$ is the Bloch vector, which satisfies $r_x^2+r_y^2+r_z^2\leq 1$. The action of the GADC $\A_{\gamma, N}$ on $\rho$ is given by the action of $\A_{\gamma, N}$ on the Pauli operators $\sigma_x,\sigma_y,\sigma_z$. We have that
	\begin{align}
		\mathcal{A}_{\gamma,N}(\sigma_x)&=\sqrt{1-\gamma}\sigma_x,\label{eq-GADC_Pauli_action_1}\\
		\mathcal{A}_{\gamma,N}(\sigma_y)&=\sqrt{1-\gamma}\sigma_y,\label{eq-GADC_Pauli_action_2}\\
		\mathcal{A}_{\gamma,N}(\sigma_z)&=(1-\gamma)\sigma_z,\label{eq-GADC_Pauli_action_3}\\
		\mathcal{A}_{\gamma,N}(\mathbbm{1})&=\mathbbm{1}+\gamma(1-2N)\sigma_z\label{eq-GADC_Pauli_action_4}
	\end{align}
	for all $\gamma,N\in[0,1]$. This implies that the vector $\vec{r}$ of the initial state $\rho$ gets transformed as
	\begin{equation}\label{eq-transformed_vec}
		\vec{r}\mapsto(r_x\sqrt{1-\gamma},r_y\sqrt{1-\gamma},r_z(1-\gamma)+\gamma(1-2N))\equiv \vec{R},\nonumber 
	\end{equation}
	where $\vec{R}\equiv (R_x, R_y, R_z)$.
	In particular, for any state $\rho$, we get
	\begin{multline}
		\left(\frac{R_x}{\sqrt{1-\gamma}}\right)^2+\left(\frac{R_y}{\sqrt{1-\gamma}}\right)^2+\left(\frac{R_z-\gamma(1-2N)}{1-\gamma}\right)^2 \\
		=r_x^2+r_y^2+r_z^2\leq 1,
	\end{multline}
	which implies that the initial Bloch sphere gets transformed to an ellipsoid centered at $(0,0,\gamma(1-2N))$ with $x$-, $y$- and $z$-axes $\sqrt{1-\gamma}$, $\sqrt{1-\gamma}$, $1-\gamma$, respectively. Note that all pure initial states, which satisfy $r_x^2+r_y^2+r_z^2=1$, get mapped to the surface of the ellipsoid.

	The relations \eqref{eq-GADC_Pauli_action_1}--\eqref{eq-GADC_Pauli_action_4} also imply that the GADC is covariant with respect to the Pauli-$z$ operator, i.e.,
	\begin{equation}\label{eq-GADC_Z_covariant}
		\mathcal{A}_{\gamma,N}(\sigma_z\rho\sigma_z)=\sigma_z\mathcal{A}_{\gamma,N}(\rho)\sigma_z
	\end{equation}
	for all states $\rho$ and all $\gamma,N\in[0,1]$. More generally, the GADC is covariant with respect to the operator $\e^{\mathrm{i}  a\hat{n}}$, where
	\begin{equation}\label{eq-num_operator}
		\hat{n}\equiv\ket{1}\bra{1}
	\end{equation}
	is the number operator, i.e.,
	\begin{equation}
		\mathcal{A}_{\gamma,N}(\e^{\mathrm{i} a\hat{n}}\rho\e^{-\mathrm{i} a\hat{n}})=\e^{\mathrm{i} a\hat{n}}\mathcal{A}_{\gamma,N}(\rho)\e^{-\mathrm{i} a\hat{n}}
	\end{equation}
	for all states $\rho$, all $a\in\mathbb{R}$, and all $\gamma,N\in[0,1]$.
	
	We also have that
	\begin{equation}\label{eq-GADC_N_symmetry}
		\mathcal{A}_{\gamma,N}(\rho)=\sigma_x\mathcal{A}_{\gamma,1-N}(\sigma_x\rho\sigma_x)\sigma_x
	\end{equation}
	for all states $\rho$ and all $\gamma,N\in[0,1]$. In other words, the GADC $\mathcal{A}_{\gamma,N}$ is related to the GADC $\mathcal{A}_{\gamma,1-N}$ via a simple pre- and post-processing by the unitary $\sigma_x$. The information-theoretic aspects of the GADC are thus invariant under the interchange $N\leftrightarrow 1-N$, which means that we can, without loss of generality, restrict the parameter $N$ to the interval $\left[0,1/2\right]$.

	We now recall the following well-known decomposition theorems for an arbitrary generalized amplitude damping channel $\A_{\gamma, N}$:
	\begin{enumerate}
		\item Let $\gamma\in[0,1]$ and $N\in[0,1]$. Then any generalized amplitude damping channel 	$\A_{\gamma, N}$ can be decomposed as a convex combination of $\A_{\gamma, 0}$ and $\A_{\gamma, 1}$, i.e., 	
			\begin{align}\label{eq-GADC_decomp_convex}
				\mathcal{A}_{\gamma,N}=(1-N)\mathcal{A}_{\gamma,0}+N\mathcal{A}_{\gamma,1}.
			\end{align}
			
		\item Let $\gamma_1,\gamma_2\in[0,1]$ and $N_1,N_2\in[0,1]$. Then, any generalized amplitude damping channel $\A_{\gamma, N}$ can be decomposed as the concatenation of two generalized amplitude damping channels $\A_{\gamma_1, N_1}$ and $\A_{\gamma_2, N_2}$ \cite{LG15}:
			\begin{equation}\label{eq-GADC_decomp_gen}
				\mathcal{A}_{\gamma,N}=\mathcal{A}_{\gamma_2,N_2}\circ\mathcal{A}_{\gamma_1,N_1}
			\end{equation}
			where $\gamma = \gamma_1+\gamma_2-\gamma_1\gamma_2$ and $N= \frac{\gamma_1(1-\gamma_2)N_1+\gamma_2N_2}{\gamma_1+\gamma_2-\gamma_1\gamma_2}$. 
		
	\end{enumerate}

	A consequence of \eqref{eq-GADC_decomp_gen} is that, for all $\gamma,N\in[0,1]$,
	\begin{align}
		\mathcal{A}_{\gamma,N}&=\mathcal{A}_{\gamma N,1}\circ\mathcal{A}_{\frac{\gamma(1-N)}{1-\gamma N},0},\label{eq-GADC_decomp_spec_1} \\
		\mathcal{A}_{\gamma,N}&=\mathcal{A}_{\gamma(1-N),0}\circ\mathcal{A}_{\frac{\gamma N}{1-\gamma(1-N)},1}.\label{eq-GADC_decomp_spec_2}
	\end{align}

	We define
	\begin{equation}\label{eq-GADC_comp}
		\mathcal{A}_{\gamma,N}^c(\rho_A)\equiv\Tr_B[V_{A\to BE}^{\gamma,N}\rho_A(V_{A\to BE}^{\gamma,N})^\dagger]
	\end{equation}
	to be a channel complementary to $\mathcal{A}_{\gamma,N}$, where $V_{A\to BE}^{\gamma,N}$ is an isometric extension of $\mathcal{A}_{\gamma,N}$, which we take to be
	\begin{equation}\label{eq-GADC_iso_ext}
		V_{A\to BE}^{\gamma,N}\equiv A_1\otimes\ket{0}_E+A_2\otimes\ket{1}_E+A_3\otimes\ket{2}_E+A_4\otimes\ket{3}_E.
	\end{equation}

\subsection{The qubit thermal channel}\label{subsec-qubit_thermal_chan}

	The GADC is presented in a different form in \cite{RMG18} and is called the ``qubit thermal attenuator channel''. In this section, we show explicitly that the qubit thermal attenuator channel is equal to the GADC up to a reparameterization.
	
	A qubit thermal attenuator channel, which we refer to here as a ``qubit thermal channel'', is defined by analogy with the bosonic thermal channel \cite{AS17} as the interaction of two qubit systems $A$ and $E$ via a unitary channel, given by the unitary $U^\eta$, followed by discarding the system $E$ \cite{GF05}. See Fig.~\ref{fig-qubit_thermal_noise} for an illustration. The unitary $U^\eta$ is defined as
	\begin{align}\label{eq-BS_unitary}
		U^\eta=\begin{pmatrix} 1&0&0&0\\0&\sqrt{\eta}&\sqrt{1-\eta}&0\\0&-\sqrt{1-\eta}&\sqrt{\eta}&0\\0&0&0&1\end{pmatrix}.
	\end{align}
	This unitary is analogous to the unitary transformation induced by an optical beamsplitter with transmissivity $\eta\in[0,1]$. Such an optical beamsplitter is defined such that if one of the input arms contains no light, then the fraction $\eta$ of the light is transmitted unaltered, while the remaining fraction is reflected into the other output arm. The unitary transformation for the optical beamsplitter can be written as $\e^{\I \theta H_{\text{BS}}}$, where $H_{\text{BS}}=\I(\hat{a}^\dagger \hat{b}-\hat{b}^\dagger \hat{a})$ and $\theta=\arccos(\sqrt{\eta})$ (see, e.g., \cite{KMN+07}). Here, $\hat{a}$ and $\hat{b}$ are the bosonic annihilation operators corresponding to the two input arms of the beamsplitter. The unitary $U^\eta$ for the qubit thermal channel can be written in the same form $\e^{\I\theta H_{\text{BS}}}$ by replacing the bosonic annihilation operator $\hat{a}$ in $H_{\text{BS}}$ with $\sigma_-\otimes\mathbbm{1}$ and the operator $\hat{b}$ with $\mathbbm{1}\otimes\sigma_-$, where $\sigma_-\equiv \ket{0}\bra{1}$ can be thought of as the qubit analogue of the annihilation operator.
	
	\begin{figure}
		\centering
		\includegraphics[scale=1]{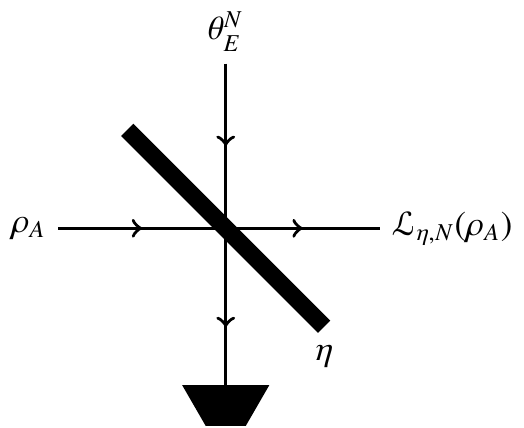}
		\caption{The qubit thermal  channel is defined by analogy with the bosonic thermal channel as the interaction of a system $A$ in the state $\rho_A$ with an environment in the state $\theta_E^N$ (see \eqref{eq-thermal_state}) at a ``beamsplitter'' of transmissivity $\eta$, which is a unitary channel defined by the unitary $U^\eta$ in \eqref{eq-BS_unitary}. The state of the environment is then discarded to obtain the output $\mathcal{L}_{\eta,N}(\rho_A)$.}\label{fig-qubit_thermal_noise}
	\end{figure}
	
	Let $\rho_A$ denote the state of the input system $A$, and let the initial state of the system $E$ be
	\begin{align}\label{eq-thermal_state}
		\theta^N_E \equiv (1-N)\ket{0}\bra{0}_E+N\ket{1}\bra{1}_E.
	\end{align}
	Then, the qubit thermal  channel $\L_{\eta, N}$ is defined as
	\begin{align}
		\mathcal{L}_{\eta,N}(\rho_A)&\equiv\Tr_E[U_{AE\to BE}^\eta(\rho_A\otimes\theta_E^N)(U_{AB\to AE}^\eta)^\dagger]\label{eq-GADC_unitary_rep}\\
		&=\Tr_{EE'}[(U_{AE\to BE}^\eta\otimes\mathbbm{1}_{E'})(\rho_A\otimes\ket{\theta^N}\bra{\theta^N}_{EE'})\nonumber\\
		&\qquad\qquad\qquad\times (U_{AE\to BE}^\eta\otimes\mathbbm{1}_{E'})^\dagger]\label{eq-GADC_unitary_rep_2},
	\end{align}
	where
	\begin{equation}
		\ket{\theta^N}_{EE'}\equiv \sqrt{1-N}\ket{0,0}_{EE'}+\sqrt{N}\ket{1,1}_{EE'}.
	\end{equation}
	When $N=0$, we call the qubit thermal  channel $\mathcal{L}_{\eta,0}$ the qubit pure-loss channel.
	
	The qubit thermal channel as defined in \eqref{eq-GADC_unitary_rep} has exactly the same form as the bosonic thermal channel, the latter having the unitary $U^{\eta}$ defined in \eqref{eq-BS_unitary} replaced by $\e^{\I\theta H_{\text{BS}}}$. In particular, the initial state $\theta_E^N$ of the system $E$ can be thought of as the qubit analogue of the bosonic thermal state $\e^{-\beta\hat{a}^\dagger\hat{a}}/{\Tr[\e^{-\beta\hat{a}^\dagger\hat{a}}]}$ \cite{AS17}, and the parameter $N\in[0,1]$ can be thought of as the mean number of photons. Indeed, if we replace $\hat{a}$ with $\sigma_-$ in the definition of the bosonic thermal state, observe using the definition of the number operator $\hat{n}$ in \eqref{eq-num_operator} that $\sigma_-^\dagger\sigma_-=\hat{n}$, and let $\beta=\ln\left(\frac{1-N}{N}\right)$, then we obtain
	\begin{align}
		\frac{\e^{-\beta\sigma_-^\dagger\sigma_-}}{\Tr[\e^{-\beta\sigma_-^\dagger\sigma_-}]}&=\frac{1}{1+\e^{-\beta}}\ket{0}\bra{0}+\frac{\e^{-\beta}}{1+\e^{-\beta}}\ket{1}\bra{1}\\
		&=(1-N)\ket{0}\bra{0}+N\ket{1}\bra{1}\\
		&=\theta^N.
	\end{align} 

	There is a simple connection between the qubit thermal  channel and the generalized amplitude damping channel that is straightforward to prove: for all $\gamma\in[0,1]$ and $N \in [0, 1]$,
	\begin{equation}\label{eq-GADC_to_thermalnoise}
		\mathcal{A}_{\gamma,N}=\mathcal{L}_{1-\gamma,N}.
	\end{equation}
	Using this, along with \eqref{eq-GADC_decomp_spec_1} and \eqref{eq-GADC_decomp_spec_2}, we obtain the following serial decompositions of the qubit thermal channel:
	\begin{align}
		\mathcal{L}_{\eta,N}&=\mathcal{L}_{1-(1-\eta)N,1}\circ\mathcal{L}_{\frac{\eta}{1-(1-\eta)N},0},\label{eq-qubitThermal_decomp_spec_1}\\
		\mathcal{L}_{\eta,N}&=\mathcal{L}_{\eta+(1-\eta)N,0}\circ\mathcal{L}_{\frac{\eta}{\eta+(1-\eta)N},1}.\label{eq-qubitThermal_decomp_spec_2}
	\end{align}
	These decompositions are depicted in Fig. \ref{fig-qubitThermal_decomp}.
	
	\begin{figure}
		\centering
		\includegraphics[width=\columnwidth]{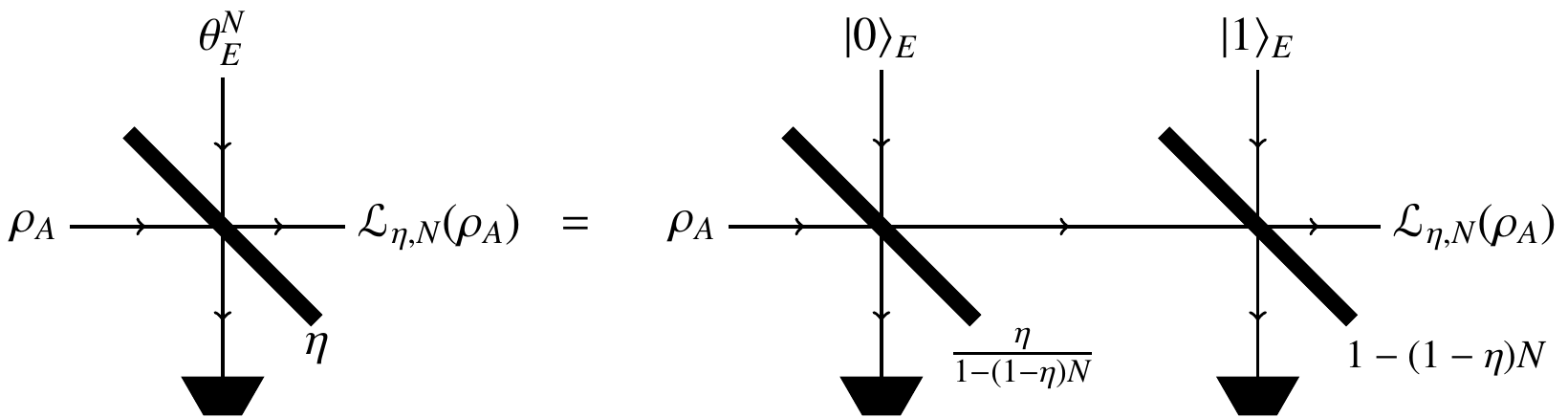}\\[0.5cm]
		\includegraphics[width=\columnwidth]{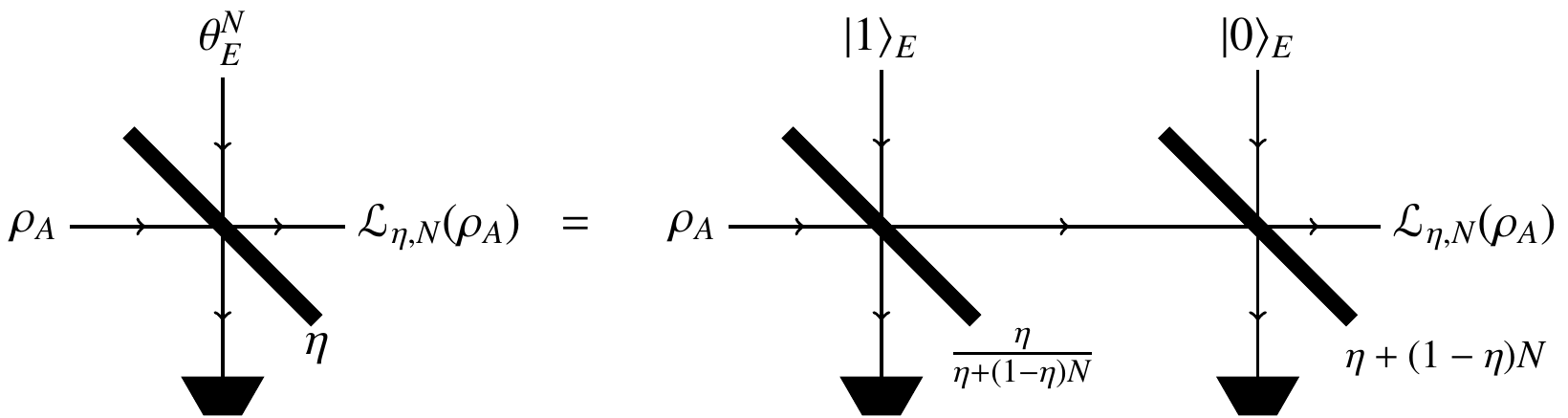}
		\caption{Serial decompositions of the qubit thermal channel, as given in \eqref{eq-qubitThermal_decomp_spec_1} and \eqref{eq-qubitThermal_decomp_spec_2}.}\label{fig-qubitThermal_decomp}
	\end{figure}
	
	We take a channel complementary to the qubit thermal  channel to be
	\begin{equation}
		\begin{aligned}
		\mathcal{L}_{\eta,N}^c(\rho_A)&\equiv \Tr_B[(U_{AE\to BE}^{\eta}\otimes\mathbbm{1}_{E'})(\rho_A\otimes\ket{\theta^N}\bra{\theta^N}_{EE'})\\
		&\qquad\qquad\qquad\times(U_{AE\to BE}^{\eta}\otimes\mathbbm{1}_{E'})^\dagger],
		\end{aligned}
	\end{equation}
	and we define a \textit{weakly complementary channel} \cite{CG06} to be
	\begin{equation}\label{eq-QTN_weak_complement}
		\widetilde{\mathcal{L}}_{\eta,N}^c(\rho_A)\equiv\Tr_B[U_{AE\to BE}^{\eta}(\rho_A\otimes\theta^N)(U_{AE\to BE}^\eta)^\dagger].
	\end{equation}

\section{Entanglement breakability of the GADC}\label{sec:ent-break}

	Having defined the GADC, we now proceed to examine its properties. We start by determining when the channel is entanglement breaking. Necessary and sufficient conditions for entanglement-breakability of the GADC have been previously determined in \cite{FRTZ12,LG15}. For completeness, we provide the derivation here, following the same approach given in \cite{FRTZ12,LG15}.

	For any two-qubit quantum state $\rho_{AB}$, the condition
	\begin{equation}\label{eq-two_qubit_sep}
		\det(\rho_{AB}^{\t_B})\geq 0
	\end{equation}
	is necessary and sufficient for the separability of $\rho_{AB}$ \cite{ADH08}.
	
	\begin{figure}
		\centering
		\includegraphics[scale=1]{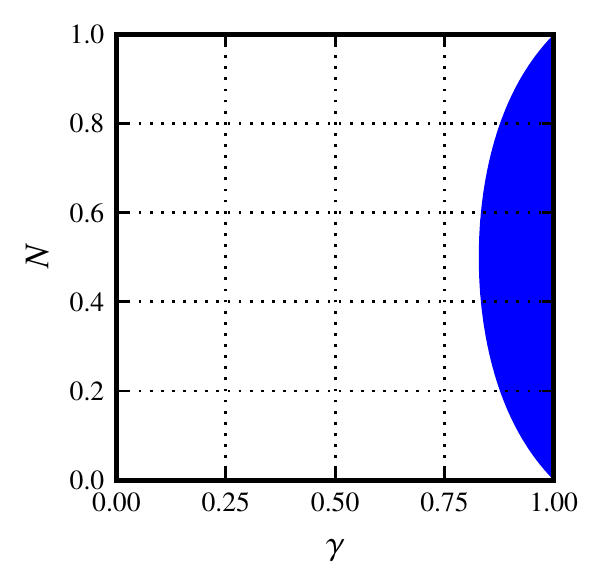}
		\caption{Region of parameters, indicated in blue as per \eqref{eq-GADC_ent_break}, for which the GADC is entanglement breaking. See also \cite{FRTZ12,LG15}.}\label{eq-ent_break_region}
	\end{figure}

	Since a channel is entanglement breaking if and only if its Choi state is separable \cite{HSR03}, to determine when the GADC $\mathcal{A}_{\gamma,N}$ is entanglement breaking, we can apply the condition in \eqref{eq-two_qubit_sep} to its Choi state $\rho_{AB}^{\gamma,N}\equiv\rho_{AB}^{\mathcal{A}_{\gamma,N}}$ as defined by \eqref{eq:choi-state}. We have
	\begin{align}
		\rho_{AB}^{\gamma,N}&=\frac{1}{2}\bigg((1-\gamma N)\ket{0,0}\bra{0,0}_{AB}+\sqrt{1-\gamma}\ket{0,0}\bra{1,1}_{AB}\nonumber\\
		&\qquad+\gamma N\ket{0,1}\bra{0,1}_{AB} +\gamma(1-N)\ket{1,0}\bra{1,0}_{AB}\nonumber\\
		&\qquad+\sqrt{1-\gamma}\ket{1,1}\bra{0,0}_{AB}\nonumber\\
		&\qquad+(1-\gamma(1-N))\ket{1,1}\bra{1,1}_{AB}\bigg)\label{eq-GADC_Choi_state}\\
		& = 		\frac{1}{2}\begin{pmatrix}
			1-\gamma N & 0 & 0 & \sqrt{1-\gamma}\\
			0 & \gamma N & 0 & 0\\
			0 & 0 & \gamma\left(  1-N\right)   & 0\\
			\sqrt{1-\gamma} & 0 & 0 & 1-\gamma\left(  1-N\right)
		\end{pmatrix}.
	\end{align}
	Then,
	\begin{equation}
	\det\left(\left(\rho_{AB}^{\gamma,N}\right)^{\t_B}\right)=\frac{-1+2\gamma-\gamma^2+\gamma^4(1-N)^2N^2}{16},
	\end{equation}
	so that $\det\left(\left(\rho_{AB}^{\gamma,N}\right)^{\t_B}\right)\geq 0$ leads to the following necessary and sufficient condition for the GADC to be entanglement breaking (see also \cite{FRTZ12,LG15}):
	\begin{equation}\label{eq-GADC_ent_break}
		\begin{aligned}
		2(\!\!\sqrt{2}-1)&\leq\gamma\leq 1,\\
\frac{1}{2}\left(1-\sqrt{\frac{\gamma^2+4\gamma-4}{\gamma^2}}\right)&\leq N\leq\frac{1}{2}\left(1+\sqrt{\frac{\gamma^2+4\gamma-4}{\gamma^2}}\right).
		\end{aligned}
	\end{equation}
	Note that $2(\!\!\sqrt{2}-1)\approx 0.8284$. See Fig.~\ref{eq-ent_break_region} for a plot of this region of parameters. It is worth remarking that while the GADC has many parallels with the bosonic thermal channel, as outlined in Section \ref{subsec-qubit_thermal_chan}, the entanglement-breakability condition obtained here is starkly different from the corresponding condition in the bosonic case. In particular, entanglement breakability of the bosonic thermal channel is given by the relatively simple condition $\eta\leq \frac{N}{N+1}$ \cite{Hol08}.

\section{Degradability and anti-degradability of the GADC}\label{sec:anti-deg}
	

\subsection{Degradability of the GADC}

	It is known that the GADC is degradable for all $\gamma\in [0,1/2]$ when $N=0$ or $N=1$ \cite{GF05}. For $N\in(0,1)$ and $\gamma\in(0,1]$, it follows from \cite[Theorem 4]{CRS08} that the GADC is not degradable.
	
	In the case $N=0$, it can be shown that \cite{GF05}
	\begin{equation}
		\mathcal{A}_{\gamma,0}^c=\mathcal{A}_{1-\gamma,0}.
	\end{equation}
	Then, using \eqref{eq-GADC_decomp_gen}, it follows from the condition $\mathcal{D}_{\gamma,0}\circ\mathcal{A}_{\gamma,0}=\mathcal{A}_{\gamma,0}^c=\mathcal{A}_{1-\gamma,0}$ that a degrading channel $\mathcal{D}_{\gamma,0}$ is simply
	\begin{equation}
		\mathcal{D}_{\gamma,0}=\mathcal{A}_{\frac{1-2\gamma}{1-\gamma},0}.
	\end{equation}
	In other words,
	\begin{equation}\label{eq-GADC_deg}
		\mathcal{A}_{\frac{1-2\gamma}{1-\gamma},0}\circ\mathcal{A}_{\gamma,0}=\mathcal{A}_{\gamma,0}^c
	\end{equation}
	for all $\gamma\in[0,1/2)$. In terms of the qubit thermal  channel, we use the correspondence in \eqref{eq-GADC_to_thermalnoise} to write the condition \eqref{eq-GADC_deg} as
	\begin{equation}
		\mathcal{L}_{\frac{1-\eta}{\eta},0}\circ\mathcal{L}_{\eta,0}=\mathcal{L}_{\eta,0}^c.
	\end{equation}
	for all $\eta\in(1/2,1]$.
	
	Although the qubit thermal channel is not degradable for $N>0$, it is \textit{weakly degradable}, meaning that there exists a channel $\widetilde{\mathcal{D}}_{\eta,N}$ such that
	\begin{equation}
		\widetilde{\mathcal{D}}_{\eta,N}\circ\mathcal{L}_{\eta,N}=\widetilde{\mathcal{L}}_{\eta,N}^c.
	\end{equation}
	In particular, one possible weakly degrading channel $\widetilde{\mathcal{D}}_{\eta,N}$ is \cite{RMG18}
	\begin{equation}\label{eq-QTN_weak_deg_channel}
		\widetilde{\mathcal{D}}_{\eta,N}=\mathcal{P}_{1-2N}\circ\mathcal{L}_{\frac{1-\eta}{\eta},N},
	\end{equation}
	where $\mathcal{P}_{\mu}$ denotes the phase damping channel, which is defined via its Kraus operators
	\begin{equation}
		\begin{pmatrix} 1 & 0 \\ 0 & \sqrt{\mu} \end{pmatrix}\quad\text{and}\quad \begin{pmatrix} 0 & 0 \\ 0 & \sqrt{1-\mu}\end{pmatrix}.
	\end{equation}

\subsection{Anti-degradability of the GADC}
	
	To determine the anti-degradability of the GADC, we use the fact that a channel is anti-degradable if and only if its Choi state is two-extendable \cite{MyhrThesis}. 
	
	\begin{proposition}[Anti-degradability of the GADC]
		For all $N\in[0,1]$, the condition
		\begin{equation}
			\gamma\geq\frac{1}{2}
		\end{equation}
		is necessary and sufficient for the anti-degradability of the GADC $\mathcal{A}_{\gamma,N}$.
	\end{proposition}
	
	\begin{proof}
		Since the GADC is a qubit-to-qubit channel, its Choi state $\rho_{AB}^{\gamma,N}$ is a two-qubit state. For any two-qubit state $\rho_{AB}$, the inequality 
		\begin{equation}\label{eq-2ext_cond}
			\Tr[\rho_{AB}^2]-\Tr[\rho_B^2]\leq 4\sqrt{\det(\rho_{AB})}
		\end{equation}
		is necessary and sufficient for $\rho_{AB}$ to be two-extendable \cite{ML09,CJK+14}. For the Choi state $\rho_{AB}^{\gamma,N}$, we find that
		\begin{align}
			\Tr\left[\left(\rho_{AB}^{\gamma,N}\right)^2\right]&=\gamma^2N^2-\gamma^2N+\frac{1}{2}\gamma^2-\gamma+1,\\
			\Tr\left[\left(\rho_B^{\gamma,N}\right)^2\right]&=2\gamma^2N^2-2\gamma^2N+\frac{1}{2}\gamma^2+\frac{1}{2},\\
			\det\left(\rho_{AB}^{\gamma,N}\right)&=\frac{\gamma^4N^2(1-N)^2}{16}.
		\end{align}
		Substituting these quantities into the inequality in \eqref{eq-2ext_cond} and simplifying leads to $\gamma\geq\frac{1}{2}$ as the necessary and sufficient condition for two-extendability of the Choi state of the GADC, and hence for anti-degradability of the GADC.
	\end{proof}

	It is interesting to note that the condition for anti-degradability of the GADC has no dependence on $N$, even though, intuitively, the noise of the channel increases with $N$. This is another way in which the GADC is in contrast with the bosonic thermal channel, since for the bosonic thermal channel the anti-degradability condition depends on $N$ and is given by $\eta\leq\frac{N+1/2}{N+1}$ \cite[Eq.~(4.6)]{dp2006}.
	
	When the GADC is anti-degradable, there exists a simple anti-degrading channel $\mathcal{E}$ satisfying \eqref{eq-anti_degrade}, the form of which follows immediately from the following lemma.
	
	\begin{lemma}\label{lem-anti_degrad_chan_bd}
		Define the channel $\mathcal{E}_N^*$ by the Kraus operators
		\begin{align}
			E_0&=\ket{0}_B\bra{0}_E+\ket{1}_B\bra{1}_E,\\
			E_1&=\ket{0}_B\bra{3}_E+\ket{1}_B\bra{2}_E,
		\end{align}
		which acts on the four-dimensional output space of the complementary channel $\mathcal{A}_{\gamma,N}^c$ defined in \eqref{eq-GADC_comp}. Then,
		\begin{equation}\label{eq-anti_degrad_chan_bd}
			\mathcal{E}_N^*\circ\mathcal{A}_{\gamma,N}^c=\mathcal{A}_{1-\gamma,N}
		\end{equation}
		for all $N\in[0,1]$ and all $\gamma\in[0,1]$.
	\end{lemma}
	
	\begin{proof}
		See Appendix \ref{app-GADC_anti_degrade}.
	\end{proof}
	
	It follows that the channel $\mathcal{E}_N^*$ defined in Lemma \ref{lem-anti_degrad_chan_bd} is an anti-degrading channel at the boundary $\gamma=\frac{1}{2}$ for all $N\in[0,1]$. To find an anti-degrading channel for $\gamma>\frac{1}{2}$, we use \eqref{eq-GADC_decomp_gen} to obtain the following.
	
	\begin{proposition}
		For all $N\in[0,1]$ and all $\gamma\geq\frac{1}{2}$, the channel
		\begin{equation}
			\mathcal{E}_{\gamma,N}\equiv\mathcal{A}_{\frac{2\gamma-1}{\gamma},N}\circ\mathcal{E}_N^*
		\end{equation}
		is an anti-degrading channel for the GADC, meaning that $\mathcal{E}_{\gamma,N}\circ\mathcal{A}_{\gamma,N}^c=\mathcal{A}_{\gamma,N}$.
	\end{proposition}
	
	\begin{proof}
		The decomposition in \eqref{eq-GADC_decomp_gen} implies that
		\begin{equation}
			\mathcal{A}_{\frac{2\gamma-1}{\gamma},N}\circ\mathcal{A}_{1-\gamma,N}=\mathcal{A}_{\gamma,N}.
		\end{equation}
		Combining this with \eqref{eq-anti_degrad_chan_bd}, we find that 
		\begin{equation}
			\mathcal{A}_{\frac{2\gamma-1}{\gamma},N}\circ\mathcal{E}_N^*\circ\mathcal{A}_{\gamma,N}^c=\mathcal{A}_{\gamma,N}
		\end{equation}
		for all $N\in[0,1]$ and all $\gamma\geq\frac{1}{2}$. The result then follows.
	\end{proof}

\section{Bounds on the classical capacity of the GADC}\label{sec:c-cap-bounds}

	We now consider the communication capacities of the GADC, starting with the classical capacity. In general, the Holevo information recalled in \eqref{eq:Holevo-info-channel} is a lower bound on the classical capacity of any channel. Then, as implied by the formula in \eqref{eq-classical_capacity}, determining the classical capacity of a quantum channel essentially reduces to determining the additivity of the Holevo information, as the capacity of additive channels can be calculated without any regularization. Remarkably, even in the case $N=0$, in which case the GADC reduces to the amplitude damping channel, determining the additivity of the Holevo information remains an important open problem. In the case $N=\frac{1}{2}$, however, we observe from \eqref{eq-GADC_Pauli_action_4} that the GADC is unital, i.e., $\mathcal{A}_{\gamma,\frac{1}{2}}(\mathbbm{1})=\mathbbm{1}$ for all $\gamma\in[0,1]$. The Holevo information is additive for unital qubit channels \cite{King02}, i.e.,
	\begin{equation}
		\chi(\mathcal{N}\otimes\mathcal{M})=\chi(\mathcal{N})+\chi(\mathcal{M})
	\end{equation}
	for any unital qubit channel $\mathcal{N}$ and for any channel $\mathcal{M}$. This implies that the classical capacity of any unital qubit channel is equal to its Holevo information. In particular, for the GADC, we obtain
	\begin{equation}\label{eq-GADC_CCap_N05}
		C(\mathcal{A}_{\gamma,\frac{1}{2}})=\chi(\mathcal{A}_{\gamma,\frac{1}{2}}).
	\end{equation}
	Furthermore, the Holevo information for unital qubit channels is directly related to its minimum output entropy \cite{KR01,Cort02,Cort04} (see also \cite[Example~8.10]{H13book}), such that for the GADC with $N=\frac{1}{2}$ we obtain
	\begin{equation}
		\chi(\mathcal{A}_{\gamma,\frac{1}{2}})=1-h_2\left(\frac{1-\sqrt{1-\gamma}}{2}\right).
	\end{equation}	
	
	The Holevo information is also known to be additive for entanglement breaking channels \cite{Shor02}. Therefore, using the result in \eqref{eq-GADC_ent_break}, we obtain
	\begin{equation}
				C(\mathcal{A}_{\gamma,N})=\chi(\mathcal{A}_{\gamma,N}),
	\end{equation}
	for all $\gamma$ and $N$ satisfying
	\begin{equation}
		\begin{aligned}
			2(\sqrt{2}-1)&\leq\gamma\leq 1,\\
			\frac{1}{2}\left(1-\sqrt{\frac{\gamma^2+4\gamma-4}{\gamma^2}}\right)&\leq N\leq \frac{1}{2}\left(1+\sqrt{\frac{\gamma^2+4\gamma-4}{\gamma^2}}\right).
		\end{aligned}
	\end{equation}
	
	Using the techniques from \cite{Cort02,Ber05}, it has been shown in \cite{LM07b} that the Holevo information of the GADC for its entire parameter range is given by
	\begin{equation}\label{eq-GADC_Hol_inf}
		\chi(\mathcal{A}_{\gamma,N})=\frac{1}{2}(f(r^*)-\log_2(1-q^2)-qf'(q)),
	\end{equation}
	where
	\begin{align}
		f(x)&\equiv(1+x)\log_2(1+x)+(1-x)\log_2(1-x),\\
		f'(x)&=\frac{\text{d}}{\text{d}x}f(x)=\log_2\left(\frac{1+x}{1-x}\right),\\
		r^*&\equiv\sqrt{1-\gamma-\frac{(q-\gamma(1-2N))^2}{1-\gamma}+q^2},
	\end{align}
	and $q$ is determined as the solution to the equation
	\begin{multline}
		(\gamma q-\gamma^2(1-2N)-\gamma(1-\gamma)(1-2N))f'(r^*)\\=-r^*(1-\gamma)f'(q).
	\end{multline}
	
	Let us now compare the Holevo information lower bound with two upper bounds based on the concepts of $\varepsilon$-entanglement-breakability and $\varepsilon$-covariance.
	
	\begin{proposition}[Classical capacity upper bounds via $\varepsilon$-entanglement-breakability and $\varepsilon$-covariance]\label{prop-GADC_CCap_UB_EB_cov}
		For all $\gamma,N\in(0,1)$ it holds that
		\begin{align}
			C(\mathcal{A}_{\gamma,N})&\leq \chi(\mathcal{M}_{\gamma,N})+2\varepsilon_1+g(\varepsilon_1)\equiv C_{\operatorname{EB}}^{\operatorname{UB}}(\gamma,N),\label{eq-CCap_UB1_EB}\\
			C(\mathcal{A}_{\gamma,N})&\leq \chi(\mathcal{A}_{\gamma,\frac{1}{2}})+2\varepsilon_2+g(\varepsilon_2)\equiv C_{\operatorname{cov}}^{\operatorname{UB}}(\gamma,N),\label{eq-CCap_UB2_cov}
		\end{align}
		where $\varepsilon_1=\varepsilon_{\operatorname{EB}}(\mathcal{A}_{\gamma,N})=\frac{1}{2}\norm{\mathcal{A}_{\gamma,N}-\mathcal{M}_{\gamma,N}}_{\diamond}$ and $\varepsilon_2=\varepsilon_{\operatorname{cov}}(\mathcal{A}_{\gamma,N})=\gamma\left|N-\frac{1}{2}\right|$.
	\end{proposition}
	
	\begin{proof}
		To obtain \eqref{eq-CCap_UB1_EB}, we use \eqref{eq-CCap_UB_approx_EB} and the fact that $d_B=2$ for the GADC. Furthermore, we note here again that since the GADC is a qubit-to-qubit channel, the entanglement-breaking parameter $\varepsilon_{\text{EB}}(\mathcal{A}_{\gamma,N})$ defined in \eqref{eq-approx_EB} can be calculated via an SDP \cite[Lemma~III.8]{LKDW18} due to the fact that, for two-qubit states, the set of separable states is equal to the set of states with positive partial transpose \cite{Peres96,HHH96}.
		
	\begin{figure}
		\centering
		\includegraphics[width=\columnwidth]{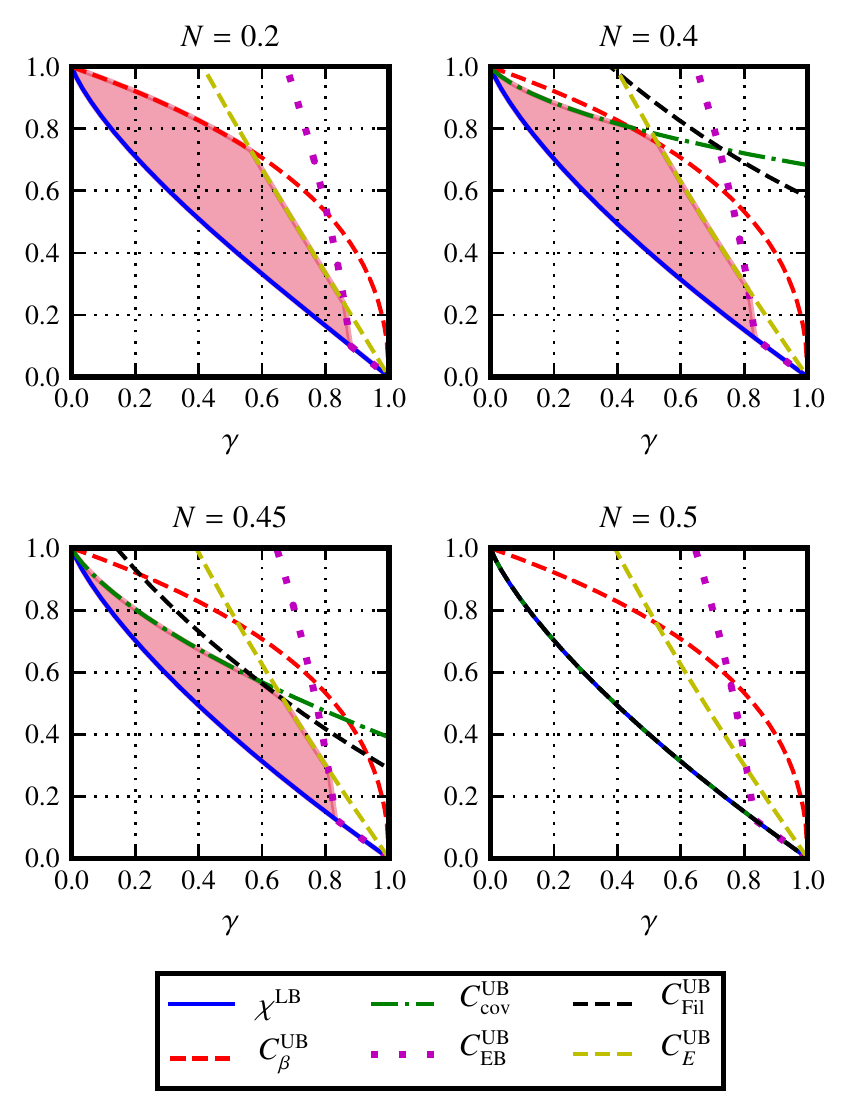}
		\caption{Bounds on the classical capacity of the GADC. Shown is the Holevo information lower bound given by \eqref{eq-GADC_Hol_inf}, as well as the $C_{\beta}$ upper bound given by \eqref{eq-C_beta_GADC}. We also plot the upper bound in \eqref{eq-CCap_UB1_EB} based on approximate entanglement breakability, the upper bound in \eqref{eq-CCap_UB2_cov} based on approximate covariance, the upper bound in \eqref{eq-GADC_CCap_UB_Fil} from \cite{F18}, and the entanglement-assisted classical capacity $C_E$ given by \eqref{eq-GADC_mut_inf}. The classical capacity lies within the shaded region. For $N=\frac{1}{2}$, the classical capacity is equal to the Holevo information, and this coincides with the approximate covariance upper bound and the upper bound from \cite{F18}.}\label{fig-CCap_bounds}
	\end{figure}
		
		For the bound in \eqref{eq-CCap_UB2_cov}, we make use of \eqref{eq-CCap_UB_approx_cov}. Let us first show that the channel $\mathcal{A}_{\gamma,N}^G$ obtained by twirling with the Pauli operators $\{\mathbbm{1},\sigma_x,\sigma_y,\sigma_z\}$ is equal to $\mathcal{A}_{\gamma,\frac{1}{2}}$. We start by recalling the convex decomposition of the GADC as stated in \eqref{eq-GADC_decomp_convex}:
		\begin{equation}\label{eq-GADC_decomp_convex_2}
			\mathcal{A}_{\gamma,N}=(1-N)\mathcal{A}_{\gamma,0}+N\mathcal{A}_{\gamma,1}.
		\end{equation}
		Thus, by linearity of the twirling channel, we have that $\mathcal{A}_{\gamma,N}^G=(1-N)\mathcal{A}_{\gamma,0}^{G}+N\mathcal{A}_{\gamma,1}^G$. Next, we recall \eqref{eq-GADC_Z_covariant} and \eqref{eq-GADC_N_symmetry}, respectively:
		\begin{align}
			\mathcal{A}_{\gamma,0}(\cdot)&=\sigma_z\mathcal{A}_{\gamma,0}(\sigma_z(\cdot)\sigma_z)\sigma_z,\\
			\mathcal{A}_{\gamma,1}(\cdot)&=\sigma_x\mathcal{A}_{\gamma,0}(\sigma_x(\cdot)\sigma_x)\sigma_x.
		\end{align}
		Using these relations, and the fact that $\sigma_y=\I\sigma_x\sigma_z$, we obtain
		\begin{align}
			\mathcal{A}_{\gamma,0}^G&=\frac{1}{2}\mathcal{A}_{\gamma,0}+\frac{1}{2}\mathcal{A}_{\gamma,1}=\mathcal{A}_{\gamma,\frac{1}{2}},\\
			\mathcal{A}_{\gamma,1}^G&=\frac{1}{2}\mathcal{A}_{\gamma,0}+\frac{1}{2}\mathcal{A}_{\gamma,1}=\mathcal{A}_{\gamma,\frac{1}{2}},
		\end{align}
		where to obtain the last equality in both equations we used \eqref{eq-GADC_decomp_convex_2}. Therefore,
		\begin{equation}\label{eq-GADC_twirl}
			\mathcal{A}_{\gamma,N}^G=\mathcal{A}_{\gamma,\frac{1}{2}}
		\end{equation}
		for all $\gamma,N\in[0,1]$. The final step is to show that $\varepsilon_{\text{cov}}(\mathcal{A}_{\gamma,N})=\frac{1}{2}\norm{\mathcal{A}_{\gamma,N}-\mathcal{A}_{\gamma,N}^G}_{\diamond}=\frac{1}{2}\norm{\mathcal{A}_{\gamma,N}-\mathcal{A}_{\gamma,\frac{1}{2}}}_{\diamond}=\gamma\left|N-\frac{1}{2}\right|$, which we do in Appendix \ref{app-GADC_cov_parameter}.
	\end{proof}

	We now compare the upper bounds obtained above with two strong converse upper bounds on the classical capacity that hold for any quantum channel $\mathcal{N}$ \cite{WXD18}. The first upper bound is
	\begin{equation}\label{eq-C_beta}
		C(\mathcal{N})\leq C_\beta(\mathcal{N})\equiv \log_2\beta(\mathcal{N}),
	\end{equation}
	where
	\begin{equation}\label{eq-C_beta_primal}
		\beta(\mathcal{N})\equiv\left\{\begin{array}{l l}\text{min.} & \Tr[S_{B}]\\
			\text{subject to} & -R_{AB}\leq \left(\Gamma_{AB}^{\mathcal{N}}\right)^{\t_B}\leq R_{AB},\\
			& -\mathbbm{1}_A\otimes S_B\leq R_{AB}^{\t_B}\leq\mathbbm{1}_A\otimes S_B.
			\end{array}\right.
	\end{equation}
	Note that the optimization is with respect to the operators $S_B$ and $R_{AB}$. We also observe that the optimization problem is a semi-definite program (SDP). 
	
	The second upper bound from \cite{WXD18}, which is also given by an SDP, is the following:
	\begin{equation}\label{eq-C_zeta}
		C(\mathcal{N})\leq C_{\zeta}(\mathcal{N})\equiv \log_2\zeta(\mathcal{N}),
	\end{equation}
	where
	\begin{equation}\label{eq-C_zeta_primal}
		\zeta(\mathcal{N})=\left\{\begin{array}{l l} \text{min.} & \Tr[S_B] \\ 
			\text{subject to} & V_{AB}\geq\Gamma_{AB}^{\mathcal{N}},\\
			& -\mathbbm{1}_A\otimes S_B\leq V_{AB}^{\t_B}\leq \mathbbm{1}_A\otimes S_B. \end{array}\right.
	\end{equation}
	
	By considering the dual of the SDPs in \eqref{eq-C_beta_primal} and \eqref{eq-C_zeta_primal}, we obtain analytic expressions for $C_\beta(\mathcal{A}_{\gamma,N})$ and $C_{\zeta}(\mathcal{A}_{\gamma,N})$ for all values of $\gamma$ and $N$, and we find that $C_{\zeta}(\mathcal{A}_{\gamma,N})=C_{\beta}(\mathcal{A}_{\gamma,N})$ for all values of $\gamma$ and $N$.
	
	\begin{proposition}\label{prop-C_beta}
		For all $\gamma,N\in[0,1]$,
		\begin{equation}\label{eq-C_beta_GADC}
			C_\beta(\mathcal{A}_{\gamma,N})=C_{\zeta}(\mathcal{A}_{\gamma,N})=\log_2(1+\sqrt{1-\gamma}).
		\end{equation}
	\end{proposition}
	
	\begin{proof}
		See Appendix \ref{app-C_beta}.
	\end{proof}
	
	Let us now compare the Holevo information lower bound and the upper bounds in Proposition \ref{prop-GADC_CCap_UB_EB_cov}, \eqref{eq-C_beta}, and \eqref{eq-C_zeta} to the upper bound obtained  in \cite{F18}. This bound is obtained using a technique developed in \cite{FFK18}, which is based on a decomposition of the channel of interest in terms of a unital channel (for which we know the classical capacity, as mentioned above). When applied to the GADC, the technique leads to the following upper bound \cite[Eq.~(35)]{F18}:
	\begin{align}
		C(\mathcal{A}_{\gamma,N})&\leq C_{\text{Fil}}^{\text{UB}}(\gamma,N)\nonumber\\
		&\equiv  1-h_2\left(\frac{1}{2}\left(1-\frac{\sqrt{1-\gamma}}{f(\gamma,N)}\right)\right)+\log_2f(\gamma,N)\nonumber\\
		&\qquad +\frac{1}{2}\log_2\frac{N}{1-N},\label{eq-GADC_CCap_UB_Fil}
	\end{align}
	where
	\begin{multline}
		f(\gamma,N)\equiv \gamma\sqrt{N(1-N)}\\
		+\sqrt{N+(1-N)(1-\gamma)}\sqrt{1-N+N(1-\gamma)}.
	\end{multline}
	
	Finally, we consider the entanglement-assisted classical capacity as another upper bound on the classical capacity of the GADC. The entanglement-assisted classical capacity of a quantum channel $\mathcal{N}$, denoted by $C_E(\mathcal{N})$, is defined as the maximum rate at which classical information can be sent over the channel in the asymptotic limit, with the assistance of entanglement between the sender and the receiver. It is known \cite{BSST99,Hol01a,ieee2002bennett} that $C_E(\mathcal{N})$ is given simply by the mutual information $I(\mathcal{N})$ of the channel, i.e.,
	\begin{equation}
		C_E(\mathcal{N})=I(\mathcal{N})\equiv\max_{\phi_{AA'}}I(A;B)_{\rho},
	\end{equation}
	where $\rho_{AB}=\mathcal{N}_{A'\to B}(\phi_{AA'})$ and the dimension of $A$ is equal to the dimension of the input system $A'$ of the channel $\mathcal{N}$. For the GADC, by using its Pauli-$z$ covariance, as well as the concavity of the function $\rho_{A'}\mapsto I(A;B)_{\omega}$, where $\omega_{AA'}=\mathcal{N}_{A'\to B}(\phi_{AA'}^{\rho})$ and $\phi_{AA'}^{\rho}$ is any purification of $\rho_{A'}$, it has been shown \cite{LM07} that
	\begin{equation}\label{eq-GADC_mut_inf}
		I(\mathcal{A}_{\gamma,N})=\max_{z\in[-1,1]}F(\gamma,N,z)
	\end{equation}
	for all $\gamma,N\in(0,1)$, where
	\begin{multline}
		F(\gamma,N,z)\\\equiv -\sum_{i=1}^2\lambda_i\log_2\lambda_i-\sum_{i=1}^2\lambda_i'\log_2\lambda_i'+\sum_{i=1}^4\lambda_i''\log_2\lambda_i''
	\end{multline}
	and
	\begin{align}
		\lambda_1&=\frac{1}{2}(1+z),\\
		\lambda_2&=\frac{1}{2}(1-z),\\
		\lambda_1'&=\frac{1}{2}\left(1+((2N-1)\gamma-(1-\gamma)z)\right),\\
		\lambda_2'&=\frac{1}{2}\left(1-((2N-1)\gamma-(1-\gamma)z)\right),\\
		\lambda_1''&=\frac{1}{2}(1-N)\gamma(1-z),\\
		\lambda_2''&=\frac{1}{2}N\gamma(1+z),\\
		\lambda_3''&=\frac{1}{4}\left(2-(1+(2N-1)z)\gamma\right.\nonumber\\
		&\quad\left.+\sqrt{4-4(1+z(2N-1))\gamma+(2N-1+z)^2\gamma^2}\right),\\
		\lambda_4''&=\frac{1}{4}\left(2-(1+(2N-1)z)\gamma\right.\nonumber\\
		&\quad-\left.\sqrt{4-4(1+z(2N-1))\gamma+(1-2N+z)^2\gamma^2}\right).
	\end{align}
	
	In Fig~\ref{fig-CCap_bounds}, we plot the Holevo information lower bound as well as the $C_{\beta}$ upper bound, the upper bound $C_{\text{EB}}^{\text{UB}}$ based on approximate entanglement breakability, the upper bound $C_{\text{cov}}^{\text{UB}}$ based on approximate covariance, the bound $C_{\text{Fil}}^{\text{UB}}$ defined in \eqref{eq-GADC_CCap_UB_Fil}, and the entanglement-assisted classical capacity $C_E$. We find that the $C_{\beta}$ upper bound is close to the Holevo information lower bound for low values of $\gamma$ and $N$. For higher values of $\gamma$, the entanglement-assisted classical capacity provides a tighter upper bound than $C_{\beta}$. For values of $N$ close to $\frac{1}{2}$, as one might expect, the approximate covariance upper bound $C_{\text{cov}}^{\text{UB}}$ is tighter than both $C_{\beta}$ and $C_E$, at least for low to intermediate values of $\gamma$. In this same regime for $N$, the bound $C_{\text{Fil}}^{\text{UB}}$ is the tightest for small intervals of $\gamma$ close to $\gamma=0.6$. For $N=\frac{1}{2}$, we know from \eqref{eq-GADC_CCap_N05} that the classical capacity of the GADC is given by the Holevo information. Accordingly, the Holevo information and the upper bounds $C_{\text{cov}}^{\text{UB}}$ and $C_{\text{Fil}}^{\text{UB}}$ coincide. Also, as expected, the approximate entanglement-breaking bound $C_{\text{EB}}^{\text{UB}}$ is tight, matching the lower bound, whenever the GADC is entanglement breaking. For values of $\gamma$ and $N$ close to the entanglement breaking region, this upper bound is also the tightest among all of the other upper bounds. 
	

\section{Bounds on the quantum and private capacities of the GADC}\label{sec:q-cap-bounds}

	We now consider the quantum and private capacities of the GADC and provide upper bounds using the data-processing bounds, the approximate degradability and approximate anti-degradability bounds, and the Rains information and relative entropy of entanglement bounds defined in Sec.~\ref{sec-bounds_QCap}.
	
	We start with the decompositions of the GADC in \eqref{eq-GADC_decomp_spec_1} and~\eqref{eq-GADC_decomp_spec_2}:
	\begin{align}
		\mathcal{A}_{\gamma,N}&=\mathcal{A}_{\gamma N,1}\circ\mathcal{A}_{\frac{\gamma(1-N)}{1-\gamma N},0},\label{eq-GADC_decomp_spec_1a}\\
		\mathcal{A}_{\gamma,N}&=\mathcal{A}_{\gamma(1-N),0}\circ\mathcal{A}_{\frac{\gamma N}{1-\gamma(1-N)},1}.\label{eq-GADC_decomp_spec_2a}
	\end{align}
	These decompositions of the GADC involve the amplitude damping channels $\mathcal{A}_{\frac{\gamma(1-N)}{1-\gamma N},0}$ and $\mathcal{A}_{\gamma(1-N),0}$. Moreover, these decompositions are similar in spirit to the ones used in  \cite{dp2006,dp2012,SWAT,NAJ18, RMG18} in the context of bosonic Gaussian thermal channels.
	
	Unlike the classical capacity, the quantum capacity of the amplitude damping channel has a known closed-form expression and is given by \cite{GF05}
	\begin{equation}\label{eq-AD_QCap}
		Q(\mathcal{A}_{\gamma,0})=\max_{p\in[0,1]}\bigg[h_2((1-\gamma)p)-h_2(\gamma p)\bigg].
	\end{equation}
	for $\gamma \in [0,1/2)$, and $Q(\mathcal{A}_{\gamma,0}) = 0$ for $\gamma \in [1/2,1]$. The quantum capacity can be determined easily in this case since the amplitude damping channel is degradable for all $\gamma\in[0,1/2]$, which implies that the coherent information of the channel is additive. The relation between $\mathcal{A}_{\gamma,0}$ and $\mathcal{A}_{\gamma,1}$ given by \eqref{eq-GADC_N_symmetry} implies that the quantum capacity of the channel $\mathcal{A}_{\gamma,1}$ is equal to the quantum capacity of the amplitude damping channel, i.e., $Q(\mathcal{A}_{\gamma,1})=Q(\mathcal{A}_{\gamma,0})$. Furthermore, since the private and quantum capacities are equal to each other for degradable channels, we have that $P(\mathcal{A}_{\gamma,0})=Q(\mathcal{A}_{\gamma,0})$.

	\begin{figure*}
		\centering
		\includegraphics[scale=1]{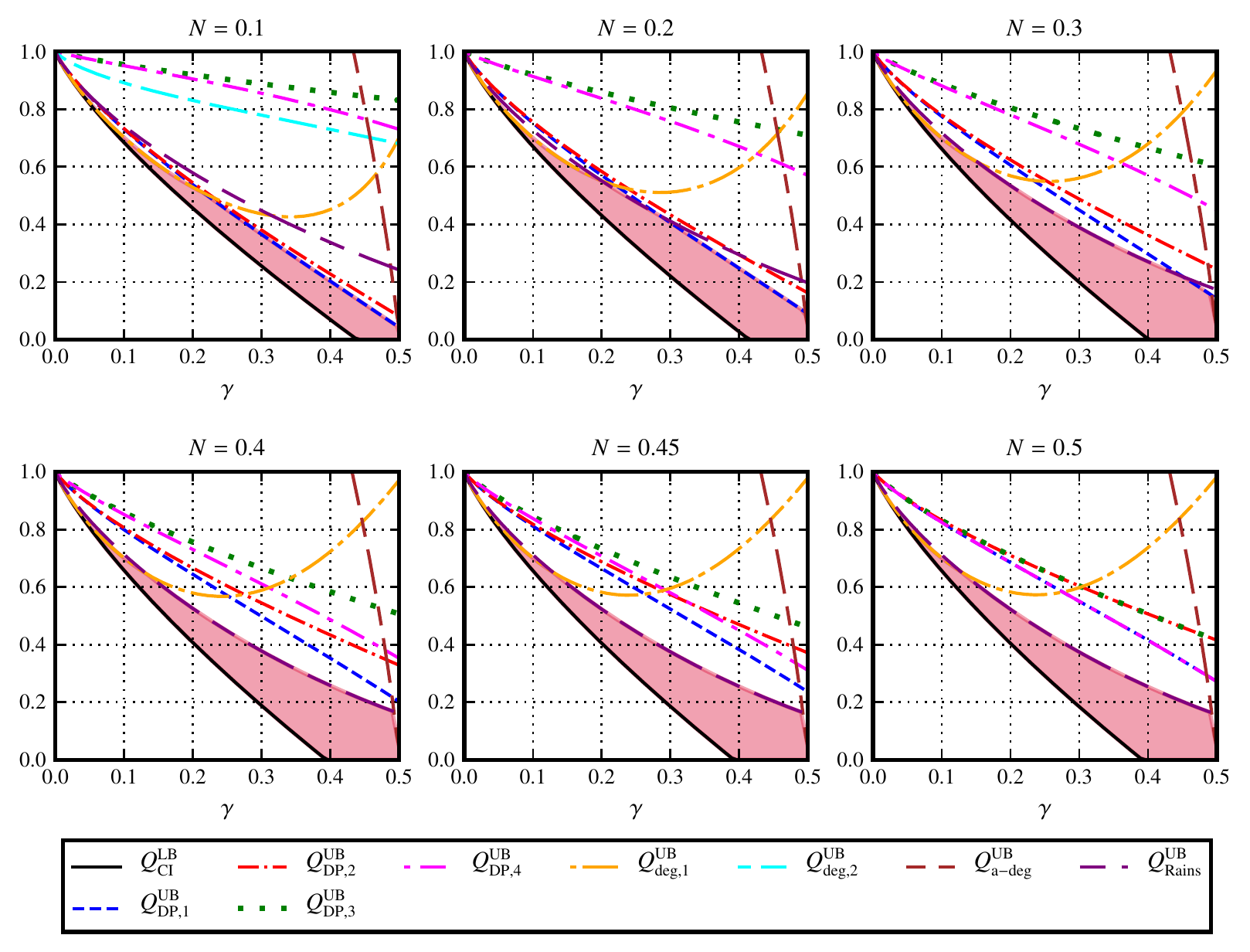}
		\caption{Bounds on the quantum capacity of the GADC. Shown is the coherent information lower bound $Q_{\text{CI}}^{\text{LB}}$ given in \eqref{eq-GADC_LB1}, the data-processing upper bounds $Q_{\text{DP},1}^{\text{UB}}$ to $Q_{\text{DP},4}^{\text{UB}}$ from Proposition \ref{prop-QCap_data_processing}, the upper bounds $Q_{\text{deg},1}^{\text{UB}}$, $Q_{\text{deg},2}^{\text{UB}}$ and $Q_{\text{a-deg}}^{\text{UB}}$ from Proposition \ref{prop-approx_deg_adeg}, and the upper bound $Q_{\text{Rains}}^{\text{UB}}$ defined in \eqref{eq-GADC_QCap_UB8}. The quantum capacity lies within the shaded region.}\label{fig-GADC_QCap_bounds}
	\end{figure*}
	
	\begingroup
	\allowdisplaybreaks[0]
	\begin{proposition}[Data-processing upper bounds]\label{prop-QCap_data_processing}
		For all $\gamma,N\in(0,1)$, it holds that
		\begin{align}
			Q(\mathcal{A}_{\gamma,N})&\leq P(\mathcal{A}_{\gamma,N})\leq Q\left(\mathcal{A}_{\frac{\gamma(1-N)}{1-\gamma N},0}\right)\equiv Q_{\textnormal{DP},1}^{\textnormal{UB}}(\gamma,N),\label{eq-GACD_Qcap_UB_DP_1}\\
			Q(\mathcal{A}_{\gamma,N})&\leq P(\mathcal{A}_{\gamma,N})\leq Q(\mathcal{A}_{\gamma(1-N),0})\equiv Q_{\textnormal{DP},2}^{\textnormal{UB}}(\gamma,N),\label{eq-GACD_Qcap_UB_DP_2}\\
			Q(\mathcal{A}_{\gamma,N})&\leq P(\mathcal{A}_{\gamma,N})\leq Q(\mathcal{A}_{\gamma N,1})\equiv Q_{\textnormal{DP},3}^{\textnormal{UB}}(\gamma,N), \label{eq-GACD_Qcap_UB_DP_3}\\
			Q(\mathcal{A}_{\gamma,N})&\leq P(\mathcal{A}_{\gamma,N})\leq Q\left(\mathcal{A}_{\frac{\gamma N}{1-\gamma(1-N)},1}\right)\equiv Q_{\textnormal{DP},4}^{\textnormal{UB}}(\gamma,N).\label{eq-GACD_Qcap_UB_DP_4}
		\end{align}
	\end{proposition}
	\endgroup
	
	\begin{proof}
		All of these inequalities follow from the relation between the quantum and private capacities in \eqref{eq-quant_priv_cap_ineq}, the decompositions of the GADC in \eqref{eq-GADC_decomp_spec_1a} and \eqref{eq-GADC_decomp_spec_2a}, and the general data processing upper bounds given in \eqref{eq-QCap_data_proc_1} and \eqref{eq-QCap_data_proc_2} for the quantum capacity and \eqref{eq-PCap_data_proc_1} and \eqref{eq-PCap_data_proc_2} for the private capacity. In particular, for the bounds on the private capacity, we make use of the fact that the amplitude damping channel is degradable, which means that its private capacity is equal to its quantum capacity, as given in \eqref{eq-AD_QCap}.
	\end{proof}
	
	We obtain more upper bounds using the concepts of $\varepsilon$-degradability, $\varepsilon$-close-degradability, and $\varepsilon$-anti-degradability. 
	
	\begingroup
	\allowdisplaybreaks[0]
	\begin{proposition}[Approximate degradability and anti-degradability upper bounds]\label{prop-approx_deg_adeg}
		For all $\gamma\in(0,1/2)$ and all $N\in(0,1)$, we have the following $\varepsilon$-degradable upper bounds:
		\begin{align}
			Q(\mathcal{A}_{\gamma,N})&\leq Q_{\textnormal{deg},1}^{\operatorname{UB}}(\gamma,N)\equiv U_{\mathcal{D}}(\mathcal{A}_{\gamma,N})+4\varepsilon_1+g(\varepsilon_1),\label{eq-GACD_Qcap_UB_DP_5}\\
			P(\mathcal{A}_{\gamma,N})&\leq U_{\mathcal{D}}(\mathcal{A}_{\gamma,N})+12\varepsilon_1+3g(\varepsilon_1)\label{eq-GACD_Pcap_UB_DP_5},
		\end{align}
		where $\varepsilon_1=\varepsilon_{\operatorname{deg}}(\mathcal{A}_{\gamma,N})$. The $\varepsilon$-close-degradable upper bounds are
		\begin{align}
			Q(\mathcal{A}_{\gamma,N})&\leq Q_{\textnormal{deg},2}^{\operatorname{UB}}(\gamma,N) \equiv Q(\mathcal{A}_{\gamma,0})+2\varepsilon_2+2g(\varepsilon_2),\label{eq-GACD_Qcap_UB_DP_6}\\
			P(\mathcal{A}_{\gamma,N})&\leq Q(\mathcal{A}_{\gamma,0})+4\varepsilon_2+4g(\varepsilon_2)\label{eq-GACD_Pcap_UB_DP_6},
		\end{align}
		where $\varepsilon_2=\frac{1}{2}\norm{\mathcal{A}_{\gamma,N}-\mathcal{A}_{\gamma,0}}_{\diamond}$. Finally, the $\varepsilon$-anti-degradable upper bounds are
		\begin{align}
			Q(\mathcal{A}_{\gamma,N})&\leq P(\mathcal{A}_{\gamma,N})\\
			&\leq  Q_{\textnormal{a-deg}}^{\operatorname{UB}}(\gamma,N)\equiv 2\varepsilon_3+h_2(\varepsilon_3)+g(\varepsilon_3) ,\label{eq-GADC_QCap_UB7_Adeg}
		\end{align}
		where $\varepsilon_3=\varepsilon_{\textnormal{a-deg}}(\mathcal{A}_{\gamma,N})$.
	\end{proposition}
	\endgroup
	
	\begin{proof}
		We start with the bounds in \eqref{eq-approx_deg_UB} and \eqref{eq-approx_deg_UB_PCap}. For the GADC, we have $d_E=4$, since the channel has four Kraus operators (assuming $N\neq 0$ and $N\neq 1$). Therefore, by determining the approximate-degradability parameter $\varepsilon_{\text{deg}}(\mathcal{A}_{\gamma,N})$, we immediately obtain the bounds in \eqref{eq-GACD_Qcap_UB_DP_5} and \eqref{eq-GACD_Pcap_UB_DP_5}.
	
		Similarly, we obtain the bounds in \eqref{eq-GACD_Qcap_UB_DP_6} and \eqref{eq-GACD_Pcap_UB_DP_6} using \eqref{eq-eps_approx_deg_UB} and \eqref{eq-eps_approx_deg_UB_PCap}, respectively, as follows. Since the channel $\mathcal{A}_{\gamma,0}$ is degradable for all $\gamma\in[0,1/2]$, we can take that to be our $\varepsilon$-close-degradable channel to $\mathcal{A}_{\gamma,N}$. Then, since $I_{\text{c}}(\mathcal{A}_{\gamma,0})$ is simply the quantum capacity of $\mathcal{A}_{\gamma,0}$ (as given by \eqref{eq-AD_QCap}), we obtain \eqref{eq-GACD_Qcap_UB_DP_6}.
	
		Finally, we use the bounds in \eqref{eq-eps_anti_degrade_upper_bound} arising from $\varepsilon$-anti-degradability. Since $d_B=2$, after calculating the anti-degradability parameter $\varepsilon_{\text{a-deg}}(\mathcal{A}_{\gamma,N})$, we obtain \eqref{eq-GADC_QCap_UB7_Adeg}.
	\end{proof}
	
	We obtain another upper bound on the private and quantum capacities of the GADC by employing the Rains information of the GADC, as given in \eqref{eq:Rains-channel}, \eqref{eq:Rains-q-cap-bound}, and \eqref{eq:REE-p-cap-bound}:
	\begin{equation}\label{eq-GADC_QCap_UB8}
	Q(\mathcal{A}_{\gamma,N}) \leq P(\mathcal{A}_{\gamma,N}) \leq R(\mathcal{A}_{\gamma,N}) \equiv Q_{\text{Rains}}^{\text{UB}}(\gamma,N),
	\end{equation}
	which follows from the fact that, as stated previously, the Rains information $R(\mathcal{A}_{\gamma,N})$ is equal to the channel's relative entropy of entanglement $E_R(\mathcal{A}_{\gamma,N})$ for qubit-to-qubit channels, due to \cite{AS08}. To compute the latter, we can perform the minimization over PPT states, due to \cite{Peres96,HHH96}. Furthermore, due to the $\sigma_z$ covariance of the GADC, we can make several simplifications to the task of computing the Rains information $R(\mathcal{A}_{\gamma,N})$, which speed it up significantly. First, due to the $\sigma_z$ covariance and concavity of Rains information in the input state, as presented in Proposition~\ref{prop:concavity-Rains}, it suffices to perform the maximization over input states with respect to the one-parameter family of states $\ket{\theta^p}_{AA'} = \sqrt{1-p} \ket{0,0}_{AA'} + \sqrt{p} \ket{1,1}_{AA'}$ (see Appendix~\ref{app-GADC_Esq_UB_pf} for details on how to show this). Second, the minimization in the Rains relative entropy in the definition in \eqref{eq:Rains-channel} can be performed over PPT states having the following form:
	\begin{equation}
	\sigma_{AB}  = \frac{1}{2}
	\begin{pmatrix}
	\alpha & 0 & 0 & \xi \e^{i \phi} \\
	0 & \beta & 0 & 0 \\
	0 & 0 & \gamma & 0 \\
	\xi \e^{-i \phi} & 0 & 0 & \delta
	\end{pmatrix},
	\end{equation}
	where $\alpha, \beta, \gamma, \delta \geq 0$, $\alpha+ \beta+ \gamma+ \delta =2$, $0 \leq \xi \leq \min\{\sqrt{\alpha \delta},\sqrt{\beta \gamma}\}$, $\phi\in[0,2\pi)$. This latter simplification follows from the same argument given in \cite[Appendix~B]{RKB+18}.
	
	See Fig.~\ref{fig-GADC_QCap_bounds} for a plot of the upper bounds $Q_{\text{DP},1}^{\text{UB}}$ to $Q_{\text{Rains}}^{\text{UB}}$. 
	To get a sense for how good these upper bounds are, it is worth comparing them to a lower bound. The coherent information $I_{\text{c}}(\mathcal{A}_{\gamma,N})$ provides a lower bound on the quantum capacity of the GADC. It can be shown that \cite{GPLS09}
	\begin{equation}\label{eq-GADC_LB1}
		I_{\text{c}}(\mathcal{A}_{\gamma,N})=\max_{p\in[0,1]}I_{\text{c}}\left(\begin{pmatrix} 1-p&0\\0&p\end{pmatrix},\mathcal{A}_{\gamma,N}\right)\equiv Q^{\text{LB}}_{\text{CI}}(\gamma,N).
	\end{equation}
	By plotting in Fig.~\ref{fig-GADC_QCap_bounds} the coherent information lower bound alongside the upper bounds $Q_{\text{DP},1}^{\text{UB}}$ to $Q_{\text{Rains}}^{\text{UB}}$, we find that the gap between the upper bounds and the lower bound is smallest when both $\gamma$ and $N$ are small. We also find that, as expected, the upper bound $Q_{\text{deg},1}^{\text{UB}}$ based on $\varepsilon$-degradability is a tighter bound for $\gamma$ close to zero, since $\gamma=0$ is the point at which the GADC is close to an identity channel. We note here that the generic behavior of the $\varepsilon$-degradable bound being tangent to the lower bound for low noise quantum channels was studied in detail in \cite{LLS18}. On the other hand, the upper bound $Q_{\text{deg},2}^{\text{UB}}$ based on $\varepsilon$-close-degradability is relatively poor for large values of $N$. Similarly, we observe that the upper bound $Q_{\text{a-deg}}^{\text{UB}}$ based on $\varepsilon$-anti-degradability is relatively poor except for values of $\gamma$ close to $\gamma=\frac{1}{2}$, where, as expected, the bound is tighter, since $\gamma=\frac{1}{2}$ is the point beyond which the GADC is anti-degradable. From Fig.~\ref{fig-GADC_QCap_bounds}, it is also evident that the upper bound $Q^{\text{UB}}_{\text{DP},1}$ is tighter than all other data-processing upper bounds for all values of $\gamma$ and for $N<0.5$. Moreover, for $N=0.5$, the upper bounds $Q^{\text{UB}}_{\text{DP},1}$ and $Q^{\text{UB}}_{\text{DP}, 2}$ coincide with the upper bounds $Q^{\text{UB}}_{\text{DP},4}$ and $Q^{\text{UB}}_{\text{DP},3}$, respectively. 
	Furthermore, the upper bound $Q_{\text{a-deg}}^{\text{UB}}$ is tighter than all other upper bounds for both $\gamma$ and $N$ close to $\frac{1}{2}$. While the Rains information upper bound $Q_{\text{Rains}}^{\text{UB}}$ is worse than two of the data-processing upper bounds for all values of $\gamma$ when $N$ is close to zero, it is tighter than all four data-processing upper bounds for all values of $\gamma$ when $N$ is close to $\frac{1}{2}$. In this region of $N$ close to $\frac{1}{2}$, it is also tighter than the bounds $Q_{\text{deg},1}^{\text{UB}}$ and $Q_{\text{a-deg}}^{\text{UB}}$ for values of $\gamma$ roughly between $0.15$ and $0.49$.

\subsection{Comparison with prior work}

	Let us now compare the bounds obtained here with those from prior work.
	
	In \cite{RMG18}, in order to obtain an upper bound on the quantum capacity of the qubit thermal  channel, the authors consider the ``extended'' channel 
	\begin{equation}
		\begin{aligned}
		\widehat{\mathcal{L}}_{\eta,N}(\rho_A)&\equiv \Tr_E[(U_{AE\to BE}^\eta\otimes\mathbbm{1}_{E'})(\rho_A\otimes\ket{\theta^N}\bra{\theta^N}_{EE'})\\
		&\qquad\qquad\qquad\times (U_{AE\to BE}^{\eta}\otimes\mathbbm{1}_{E'})^\dagger].
		\end{aligned}
	\end{equation}
	Note that
	\begin{equation}\label{eq-QTN_ext_2}
		\mathcal{L}_{\eta,N}=\Tr_{E'}\circ\widehat{\mathcal{L}}_{\eta,N},
	\end{equation}
	which implies, via \eqref{eq-QCap_data_proc_1}, that
	\begin{equation}\label{eq-QTN_QCap_UB_RMG}
		Q(\mathcal{L}_{\eta,N})\leq Q(\widehat{\mathcal{L}}_{\eta,N})\equiv Q_{\text{RMG}}^{\text{UB}}(\eta,N).
	\end{equation}
	
	\begin{figure}
		\centering
		\includegraphics[width=\columnwidth]{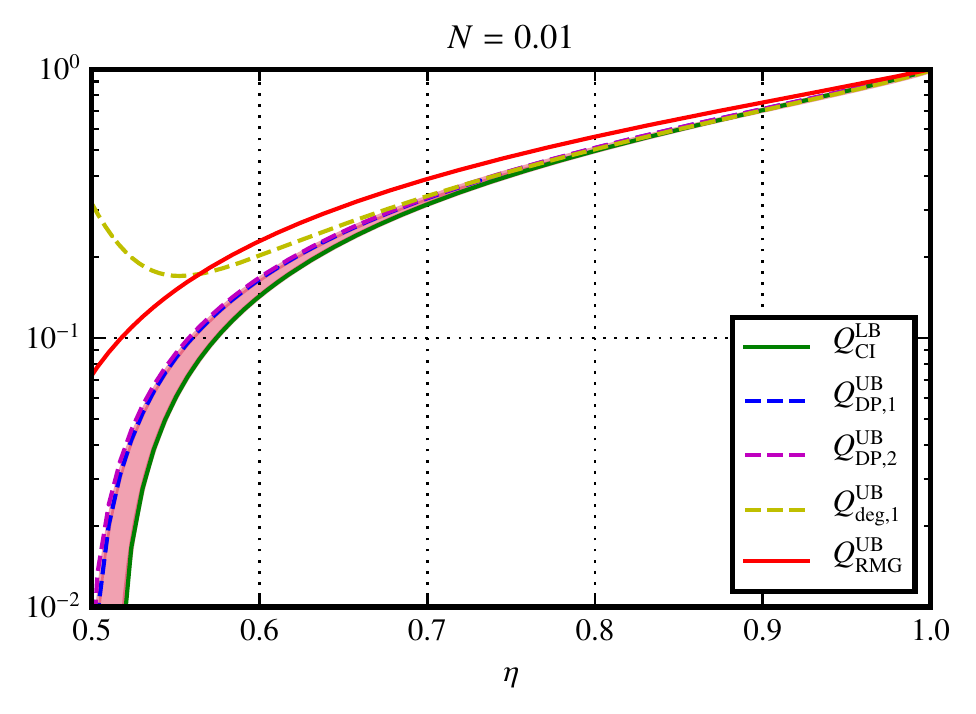}\\
		\includegraphics[width=\columnwidth]{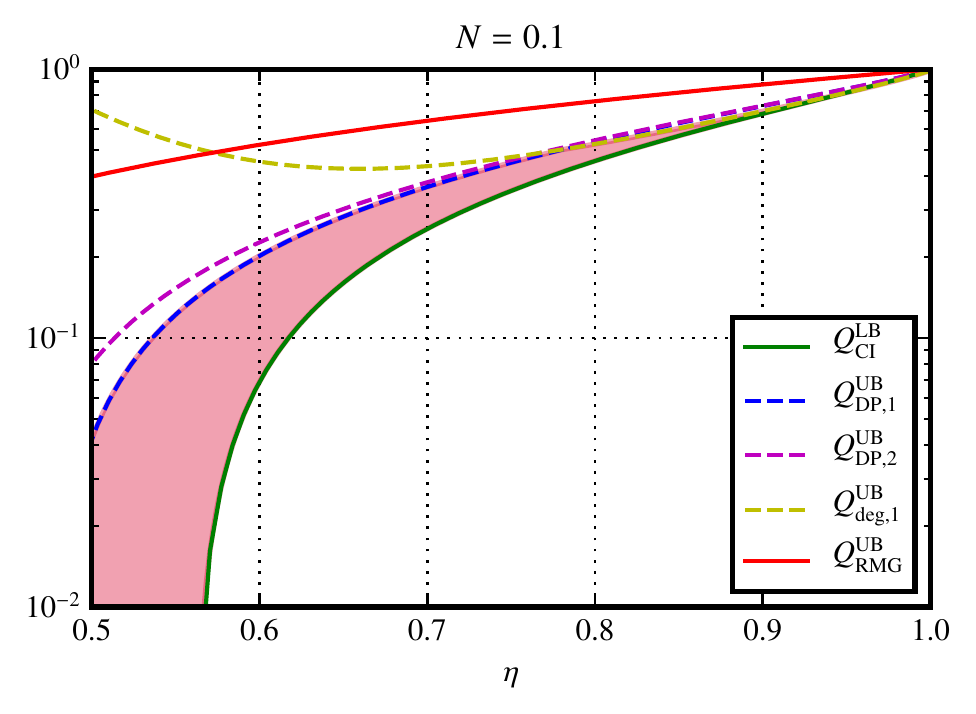}
		\caption{Comparison between the data-processing upper bounds $Q_{\text{DP},1}^{\text{UB}}$ and $Q_{\text{DP},2}^{\text{UB}}$ defined in \eqref{eq-GACD_Qcap_UB_DP_1} and \eqref{eq-GACD_Qcap_UB_DP_2}, respectively, the $\varepsilon$-degradable upper bound $Q_{\text{deg},1}^{\text{UB}}$ defined in \eqref{eq-GACD_Qcap_UB_DP_5}, and the upper bound $Q_{\text{RMG}}^{\text{UB}}$ obtained in \cite{RMG18} and defined in \eqref{eq-QTN_QCap_UB_RMG}. Also shown is the coherent information lower bound $Q_{\text{CI}}^{\text{LB}}$ defined in \eqref{eq-GADC_LB1}. The quantum capacity lies within the shaded region.}\label{fig-compRMG}
	\end{figure}
	
	As explained in \cite{RMG18}, to compute the upper bound $Q(\widehat{\mathcal{L}}_{\eta,N})$, we observe that by defining a channel complementary to $\widehat{\mathcal{L}}_{\eta,N}$ as
	\begin{equation}
		\begin{aligned}
		\widehat{\mathcal{L}}_{\eta,N}^{c}(\rho_A)&\equiv \Tr_{BE'}[(U_{AE\to BE}^\eta\otimes\mathbbm{1}_{E'})(\rho_A\otimes\ket{\theta^N}\bra{\theta^N}_{EE'})\\
		&\qquad\qquad\qquad\times (U_{AE\to BE}^{\eta}\otimes\mathbbm{1}_{E'})^\dagger],
		\end{aligned}
	\end{equation}
	we get
	\begin{equation}
		\widehat{\mathcal{L}}_{\eta,N}^{c}=\widetilde{\mathcal{L}}_{\eta,N}^c
	\end{equation}
	for all $\eta,N\in[0,1]$, where $\widetilde{\mathcal{L}}_{\eta,N}$ is the channel weakly complementary  to $\mathcal{L}_{\eta,N}$ defined in \eqref{eq-QTN_weak_complement}. This implies that whenever the qubit thermal  channel is weakly degradable, the extended channel is degradable. Indeed, for all $N>0$ and all $\eta\in[0,1]$, the channel $\widehat{\mathcal{D}}_{\eta,N}\equiv \mathcal{P}_{1-2N}\circ\mathcal{L}_{\frac{1-\eta}{\eta},N}\circ\Tr_{E'}$ satisfies
	\begin{align}
		\widehat{\mathcal{D}}_{\eta,N}\circ\widehat{\mathcal{L}}_{\eta,N}&=\mathcal{P}_{1-2N}\circ\mathcal{L}_{\frac{1-\eta}{\eta},N}\circ\Tr_{E'}\circ\widehat{\mathcal{L}}_{\eta,N}\\
		&=\mathcal{P}_{1-2N}\circ\mathcal{L}_{\frac{1-\eta}{\eta},N}\circ\mathcal{L}_{\eta,N}\\
		&=\widetilde{\mathcal{L}}_{\eta,N}^c\\
		&=\widehat{\mathcal{L}}_{\eta,N}^c,
	\end{align}
	where to obtain the second equality we used \eqref{eq-QTN_ext_2} and to obtain the third equality we used \eqref{eq-QTN_weak_deg_channel}. The quantum capacity of the extended channel is therefore given by its coherent information. In other words,
	\begin{align}
		Q(\widehat{\mathcal{L}}_{\eta,N})&=\max_{\rho}\left(H(\widehat{\mathcal{L}}_{\eta,N}(\rho))-H(\widehat{\mathcal{L}}_{\eta,N}^c(\rho))\right)\\
		&=\max_{p\in[0,1]} I_{\text{c}}\left(\begin{pmatrix} 1-p & 0 \\ 0 & p\end{pmatrix},\widehat{\mathcal{L}}_{\eta,N}\right)
	\end{align}
	for all $N>0$ and $\eta\in[0,1]$, where the last equality holds due to the fact $\widehat{\mathcal{L}}_{\eta,N}(\sigma_z\rho_A\sigma_z)=(\sigma_z\otimes\mathbbm{1}_{E'})\widehat{\mathcal{L}}_{\eta,N}(\rho)(\sigma_z\otimes\mathbbm{1}_{E'})$ and the fact that the coherent information is concave in the input state of the channel whenever the channel is degradable \cite{YHD08}.
	
	\begin{figure}
		\centering
		\includegraphics[width=\columnwidth]{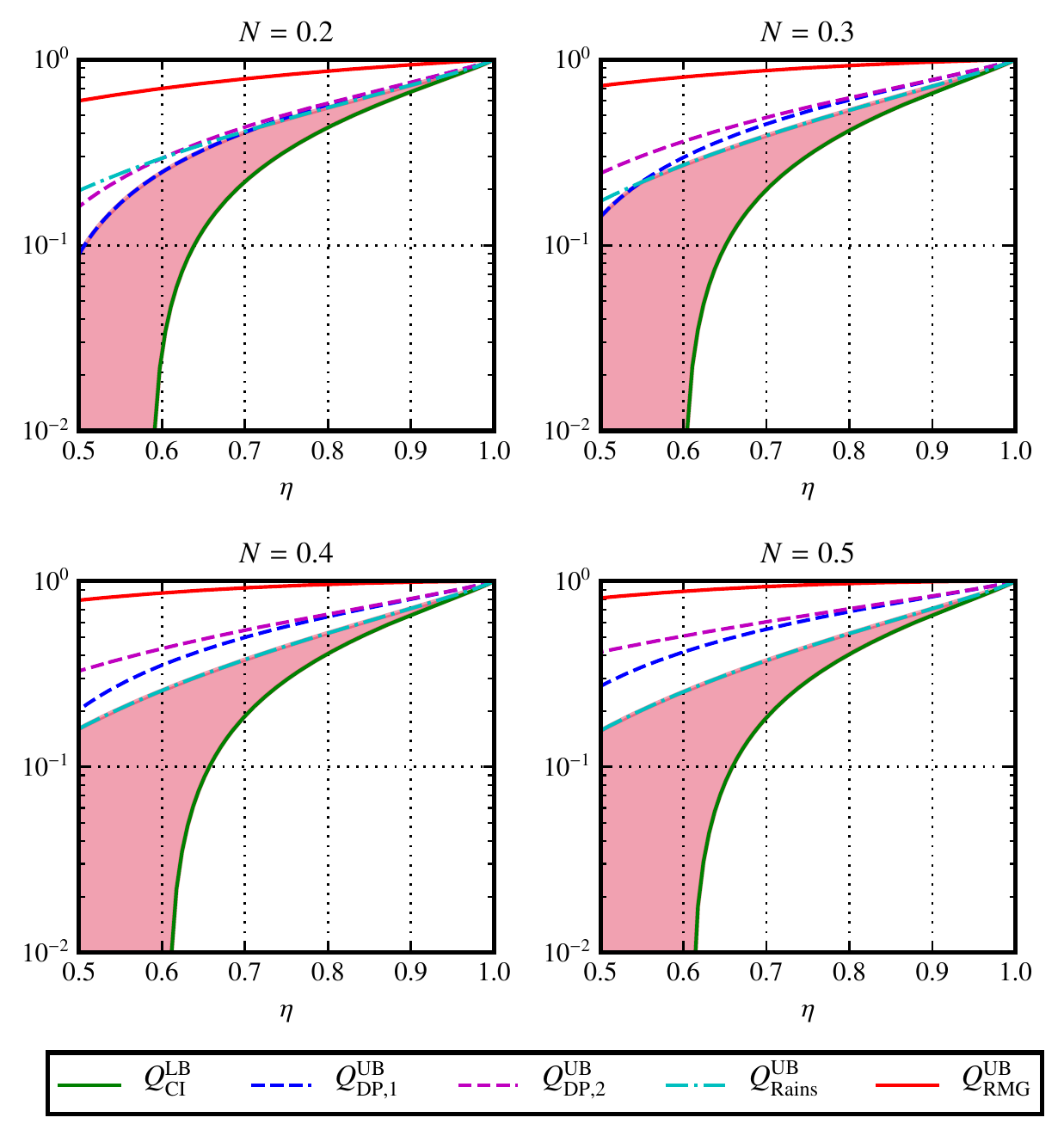}
		\caption{Comparison between the data-processing upper bounds $Q_{\text{DP},1}^{\text{UB}}$ and $Q_{\text{DP},2}^{\text{UB}}$ defined in \eqref{eq-GACD_Qcap_UB_DP_1} and \eqref{eq-GACD_Qcap_UB_DP_2}, respectively, the Rains information upper bound $Q_{\text{Rains}}^{\text{UB}}$ in \eqref{eq-GADC_QCap_UB8}, and the upper bound $Q_{\text{RMG}}^{\text{UB}}$ obtained in \cite{RMG18}. Also shown is the coherent information lower bound $Q_{\text{CI}}^{\text{LB}}$ defined in \eqref{eq-GADC_LB1}. The quantum capacity lies within the shaded region.}\label{fig-compRMG_all}
	\end{figure}
	
	See Fig.~\ref{fig-compRMG} for a comparison of the upper bounds obtained in this paper and the upper bound obtained in \cite{RMG18} for $N=0.01$ and $N=0.1$. We find that the upper bound $Q_5^{\text{UB}}$ based on approximate degradability is tighter than $Q_{\text{RMG}}^{\text{UB}}$ beyond roughly $\eta=0.56$ for both $N=0.01$ and $N=0.1$, while the data-processing upper bounds $Q_{\text{DP},1}^{\text{UB}}$ and $Q_{\text{DP},2}^{\text{UB}}$ are tighter than $Q_{\text{RMG}}^{\text{UB}}$ for all values of $\eta$. In fact, as shown in Fig.~\ref{fig-compRMG_all}, these data-processing bounds are tighter for all values of $N$. The data-processing upper bounds are thus tighter than the bound in \cite{RMG18} for the entire parameter range of the qubit thermal channel/GADC. For values of $N$ close to $\frac{1}{2}$, the Rains information upper bound $Q_{\text{Rains}}^{\text{UB}}$ is tighter than both data-processing upper bounds for all values of $\eta$.

\section{Bounds on the two-way-assisted quantum and private capacities}\label{sec:two-way-q-cap-bounds}

	In this section, we consider the two-way assisted quantum and private capacities $Q^{\leftrightarrow}(\mathcal{A}_{\gamma,N})$ and $P^{\leftrightarrow}(\mathcal{A}_{\gamma,N})$, respectively, of the GADC.

\subsection{Squashed entanglement upper bounds}

	Recalling from \eqref{eq-Q_cap_two_way_MI} that one-half of the mutual information of a channel is an upper bound on its two-way assisted quantum capacity, and using the expression for the mutual information of the GADC in \eqref{eq-GADC_mut_inf}, we get
	\begin{equation}\label{eq-GADC_Q2cap_UB1}
		Q^{\leftrightarrow}(\mathcal{A}_{\gamma,N})\leq \frac{1}{2}\max_{z\in[-1,1]}F(\gamma,N,z)\equiv Q_{\text{MI}}^{\leftrightarrow,\text{UB}}(\gamma,N)
	\end{equation}
	for all $\gamma,N\in(0,1)$.
	
	A potentially better upper bound on the two-way quantum capacity of the GADC than the one in \eqref{eq-GADC_Q2cap_UB1} can be obtained by a different choice of squashing channel. In particular, we make use of the decompositions in \eqref{eq-GADC_decomp_spec_1} and \eqref{eq-GADC_decomp_spec_2} to obtain the following result. Our approach is related to the constructions in \cite{GEW16,DSW18}.
	
	\begin{proposition}[Squashed entanglement upper bounds]\label{prop-sq_UB}
		For all $\gamma,N\in(0,1)$, it holds that
		\begin{align}
			Q^{\leftrightarrow}(\mathcal{A}_{\gamma,N})&\leq P^{\leftrightarrow}(\mathcal{A}_{\gamma,N})\nonumber\\
			&\leq \frac{1}{2}\max_{p\in[0,1]}I(A;B|E_1E_2)_{\tau^p}\equiv Q_{\textnormal{sq},1}^{\leftrightarrow,\textnormal{UB}}(\gamma,N),\label{eq-GADC_Esq_UB1}
		\end{align}
		where the state $\tau^p$ on which we evaluate the conditional mutual information is
		\begin{equation}\label{eq-GADC_Esq_UB_pf3}
			\tau_{ABE_1E_2}^p=(\id_{AB}\otimes\mathcal{A}_{\frac{1}{2},0}\otimes\mathcal{A}_{\frac{1}{2},0})(\ket{\psi_p}\bra{\psi_p}_{ABE_1'E_2'}),
		\end{equation}
		with $\ket{\psi_p}_{ABE_1'E_2'}=V_{B'\to BE_2'}^{\gamma N,1}V_{A'\to B'E_1'}^{\frac{\gamma(1-N)}{1-\gamma N},0}\ket{\theta^p}_{AA'}$ and $\ket{\theta^p}_{AA'}=\sqrt{1-p}\ket{0,0}_{AA'}+\sqrt{p}\ket{1,1}_{AA'}$.

		Also,
		\begin{align}
			Q^{\leftrightarrow}(\mathcal{A}_{\gamma,N})&\leq P^{\leftrightarrow}(\mathcal{A}_{\gamma,N})\nonumber\\
			&\leq \frac{1}{2}\max_{p\in[0,1]}I(A;B|E_1E_2)_{\tilde{\tau}^p}\equiv Q_{\textnormal{sq},2}^{\leftrightarrow,\textnormal{UB}}(\gamma,N),\label{eq-GADC_Esq_UB2}
		\end{align}
		where the state $\tilde{\tau}^p$ on which we evaluate the conditional mutual information is
		\begin{equation}
			\tilde{\tau}_{ABE_1E_2}^p=(\id_{AB}\otimes\mathcal{A}_{\frac{1}{2},0}\otimes\mathcal{A}_{\frac{1}{2},0})(\ket{\tilde{\psi}_p}\bra{\tilde{\psi}_p}_{ABE_1'E_2'}),
		\end{equation}
		with $\ket{\tilde{\psi}_p}_{ABE_1'E_2'}=V_{B'\to BE_2'}^{\gamma(1-N),0}V_{A'\to B'E_1'}^{\frac{\gamma N}{1-\gamma(1-N)},1}\ket{\theta^p}_{AA'}$.
	\end{proposition}
	
	\begin{proof}
		We use the fact that $Q^{\leftrightarrow}(\mathcal{A}_{\gamma,N})\leq E_{\text{sq}}(\mathcal{A}_{\gamma,N})$, where
		\begin{equation}\label{eq-GADC_Esq_UB_pf1}
			E_{\text{sq}}(\mathcal{A}_{\gamma,N})=\frac{1}{2}\max_{\phi_{AA'}}\inf_{\mathcal{S}_{E'\to E}}I(A;B|E)_{\omega},
		\end{equation}
		where $\omega_{ABE}=\mathcal{S}_{E'\to E}(\ket{\psi}\bra{\psi}_{ABE'})$ and $\ket{\psi}_{ABE'}$ is a purification of the state $(\id_A\otimes\mathcal{A}_{\gamma,N})(\ket{\phi}\bra{\phi}_{AA'})$.
		
		To obtain the first upper bound in \eqref{eq-GADC_Esq_UB1}, we use the fact that $\mathcal{A}_{\gamma,N}$ can be decomposed as $\mathcal{A}_{\gamma,N}=\mathcal{A}_{\gamma N,1}\circ\mathcal{A}_{\frac{\gamma(1-N)}{1-\gamma N},0}$. This means that, for any pure state $\ket{\phi}_{AA'}$, a purification of the state $\rho_{AB}\equiv(\id_A\otimes\mathcal{A}_{\gamma,N})(\ket{\phi}\bra{\phi}_{AA'})$ can be written as
		\begin{equation}\label{eq-Esq_rho_AB_purif}
			\ket{\psi}_{ABE_1'E_2'}\equiv V_{B'\to BE_2'}^{\gamma N,1}V_{A'\to B'E_1'}^{\frac{\gamma(1-N)}{1-\gamma N},0}\ket{\phi}_{AA'}
		\end{equation}
		As the squashing channels, which act on $E_1'$ and $E_2'$, we take the channels $\mathcal{A}_{\gamma_1,N_1}$ and $\mathcal{A}_{\gamma_2,N_2}$, respectively. The state $\omega_{ABE_1E_2}$ on which the quantum conditional mutual information in \eqref{eq-GADC_Esq_UB_pf1} is evaluated is then
		\begin{multline}\label{eq-GADC_Esq_ext}
			\omega_{ABE_1E_2}(\gamma_1,N_1,\gamma_2,N_2)\\\equiv (\id_{AB}\otimes\mathcal{A}_{\gamma_1,N_1}\otimes\mathcal{A}_{\gamma_2,N_2})(\ket{\psi}\bra{\psi}_{ABE_1'E_2'}).
		\end{multline}
		We can optimize over the open parameters $\gamma_1,N_1,\gamma_2,N_2\in[0,1]$ such that the squashed entanglement of $\rho_{AB}$ can be bounded from above as
		\begin{equation}\label{eq-Esq_UB1_opt}
			E_{\text{sq}}(A;B)_\rho\leq \frac{1}{2}\min_{\gamma_1,\gamma_2,N_1,N_1}I(A;B|E_1E_2)_\omega,
		\end{equation}
		where the state $\omega_{ABE_1E_2}$ is given in \eqref{eq-GADC_Esq_ext}. This means that
		\begin{equation}
			E_{\text{sq}}(\mathcal{A}_{\gamma,N})\leq\frac{1}{2}\max_{\phi_{AA'}}\min_{\gamma_1,\gamma_2,N_1,N_2}I(A;B|E_1E_2)_{\omega}.
		\end{equation}
		Now, numerical evidence suggests that $\gamma_1=\frac{1}{2}=\gamma_2$ and $N_1=0=N_2$ is optimal. The corresponding squashing channel can be viewed as qubit pure-loss channels with beamsplitters of transmissivity $\frac{1}{2}$, analogous to the construction in \cite{GEW16,DSW18}; see Fig. \ref{fig-squashing_channel_1}. So we have
		\begin{equation}
			E_{\text{sq}}(\mathcal{A}_{\gamma,N})\leq\frac{1}{2}\max_{\phi_{AA'}}I(A;B|E_1E_2)_{\tau},
		\end{equation}
		where $\tau_{ABE_1E_2}=\omega_{ABE_1E_2}(\frac{1}{2},0,\frac{1}{2},0)$. Finally, due to the covariance of the GADC with respect to the Pauli-$z$ operator, it suffices to optimize over pure states $\ket{\phi}_{AA'}=\ket{\theta^p}_{AA'}=\sqrt{1-p}\ket{0,0}_{AA'}+\sqrt{p}\ket{1,1}_{AA'}$, where $p\in[0,1]$. In other words, the following equality holds:
		\begin{equation}\label{eq-GADC_Esq_UB_pf2}
			\frac{1}{2}\max_{\phi_{AA'}}I(A;B|E_1E_2)_{\tau}=\frac{1}{2}\max_{p\in[0,1]}I(A;B|E_1E_2)_{\tau^p},
		\end{equation}
		where $\tau_{ABE_1E_2}^p$ is defined in \eqref{eq-GADC_Esq_UB_pf3}. See Appendix \ref{app-GADC_Esq_UB_pf} for a proof. We thus obtain the bound in \eqref{eq-GADC_Esq_UB1}.
		
		\begin{figure}
		\centering
		\includegraphics[width=\columnwidth]{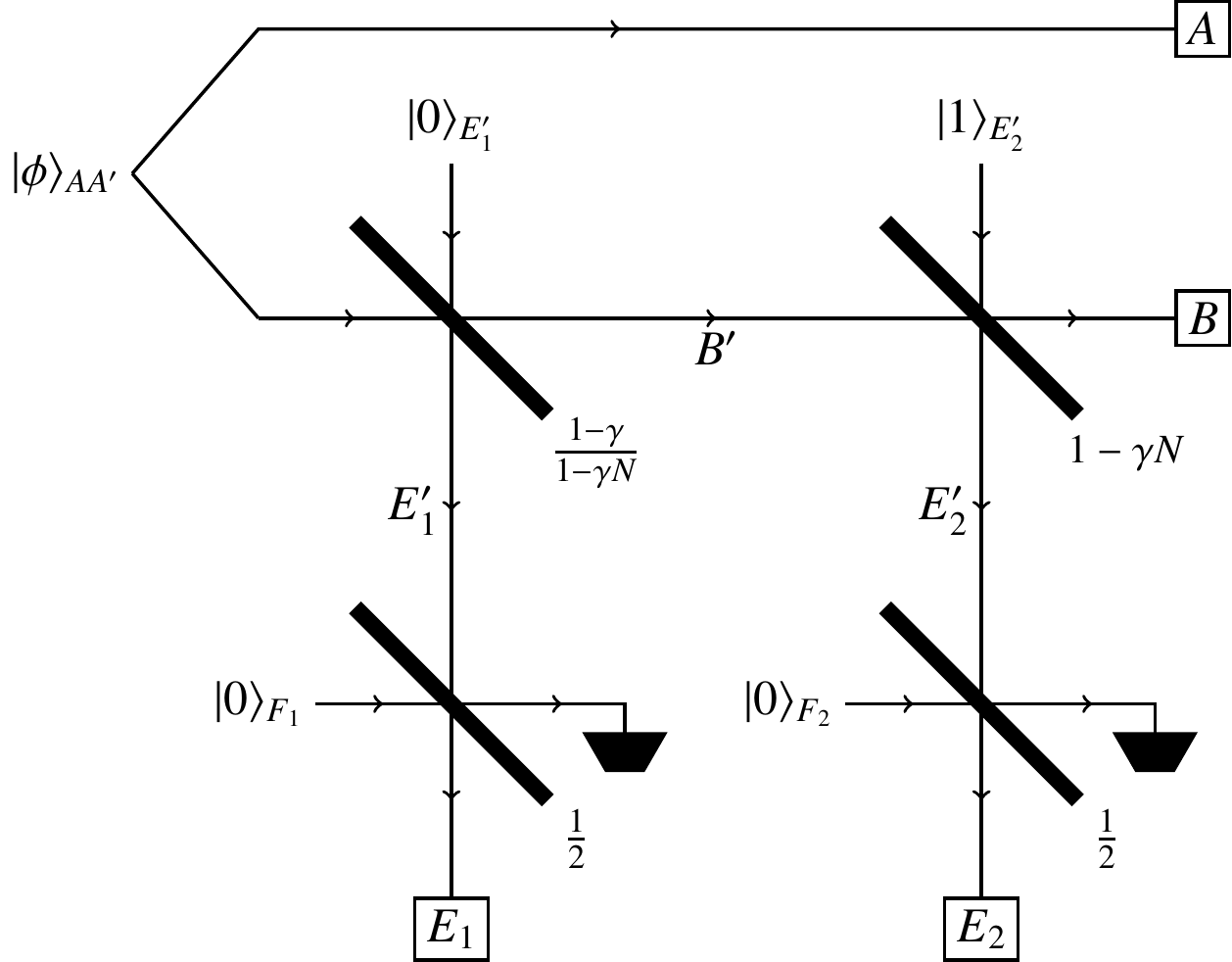}
		\caption{Strategy for the first part of the proof of Proposition \ref{prop-sq_UB}, in which we decompose the GADC as $\mathcal{A}_{\gamma,N}=\mathcal{A}_{\gamma N,1}\circ\mathcal{A}_{\frac{\gamma(1-N)}{1-\gamma N},0}$, as per \eqref{eq-GADC_decomp_spec_1}. Using \eqref{eq-GADC_to_thermalnoise}, we can write this decomposition using the qubit thermal channel as $\mathcal{A}_{\gamma,N}=\mathcal{L}_{1-\gamma N,1}\circ\mathcal{L}_{\frac{1-\gamma}{1-\gamma N},0}$. To place an upper bound on the squashed entanglement of the GADC, we use a squashing channel consisting of a 50/50 ``qubit beamsplitter'' (i.e., the unitary transformation $U^{\eta}$ defined in \eqref{eq-BS_unitary} with $\eta=\frac{1}{2}$) acting on the environment of each of the two qubit thermal channels in the decomposition of the GADC.}\label{fig-squashing_channel_1}
	\end{figure}
	
		We obtain the second upper bound in \eqref{eq-GADC_Esq_UB2} using the decomposition $\mathcal{A}_{\gamma,N}=\mathcal{A}_{\gamma(1-N),0}\circ\mathcal{A}_{\frac{\gamma N}{1-\gamma(1-N)},1}$. In this case, we take a purification of the state $\rho_{AB}=(\id_A\otimes\mathcal{A}_{\gamma,N})(\ket{\phi}\bra{\phi}_{AA'})$ to be
		\begin{equation}
			\ket{\tilde{\psi}}_{ABE_1'E_2'}\equiv V_{B'\to BE_2}^{\gamma(1-N),0}V_{A'\to B'E_1}^{\frac{\gamma N}{1-\gamma(1-N)},1}\ket{\phi}_{AA'}.
		\end{equation}
		Then, letting
		\begin{multline}
			\tilde{\omega}_{ABE_1E_2}(\gamma_1,N_1,\gamma_2,N_2)\\\equiv (\id_{AB}\otimes\mathcal{A}_{\gamma_1,N_1}\otimes\mathcal{A}_{\gamma_2,N_2})(\ket{\tilde{\psi}}\bra{\tilde{\psi}}_{ABE_1'E_2'})
		\end{multline}
		and performing the optimization $\min_{\gamma_1,\gamma_2,N_1,N_2}I(A;B|E_1E_2)_{\tilde{\omega}}$ analogous to the one in \eqref{eq-Esq_UB1_opt}, we find numerically that $\gamma_1=\frac{1}{2}=\gamma_2$ and $N_1=0=N_2$ gives the optimal value. Therefore, we get
		\begin{equation}
			E_{\text{sq}}(\mathcal{A}_{\gamma,N})\leq\frac{1}{2}\max_{p\in[0,1]} I(A;B|E_1E_2)_{\tilde{\tau}^p},
		\end{equation}
		as required. As with the first upper bound, it suffices to optimize over pure states $\ket{\theta^p}_{AA'}$ due to the covariance of the GADC with respect to the Pauli-$z$ operator, and the proof is analogous to the one presented in Appendix \ref{app-GADC_Esq_UB_pf} for the first upper bound.
	\end{proof}

	\begin{figure*}
		\centering
		\includegraphics[scale=1]{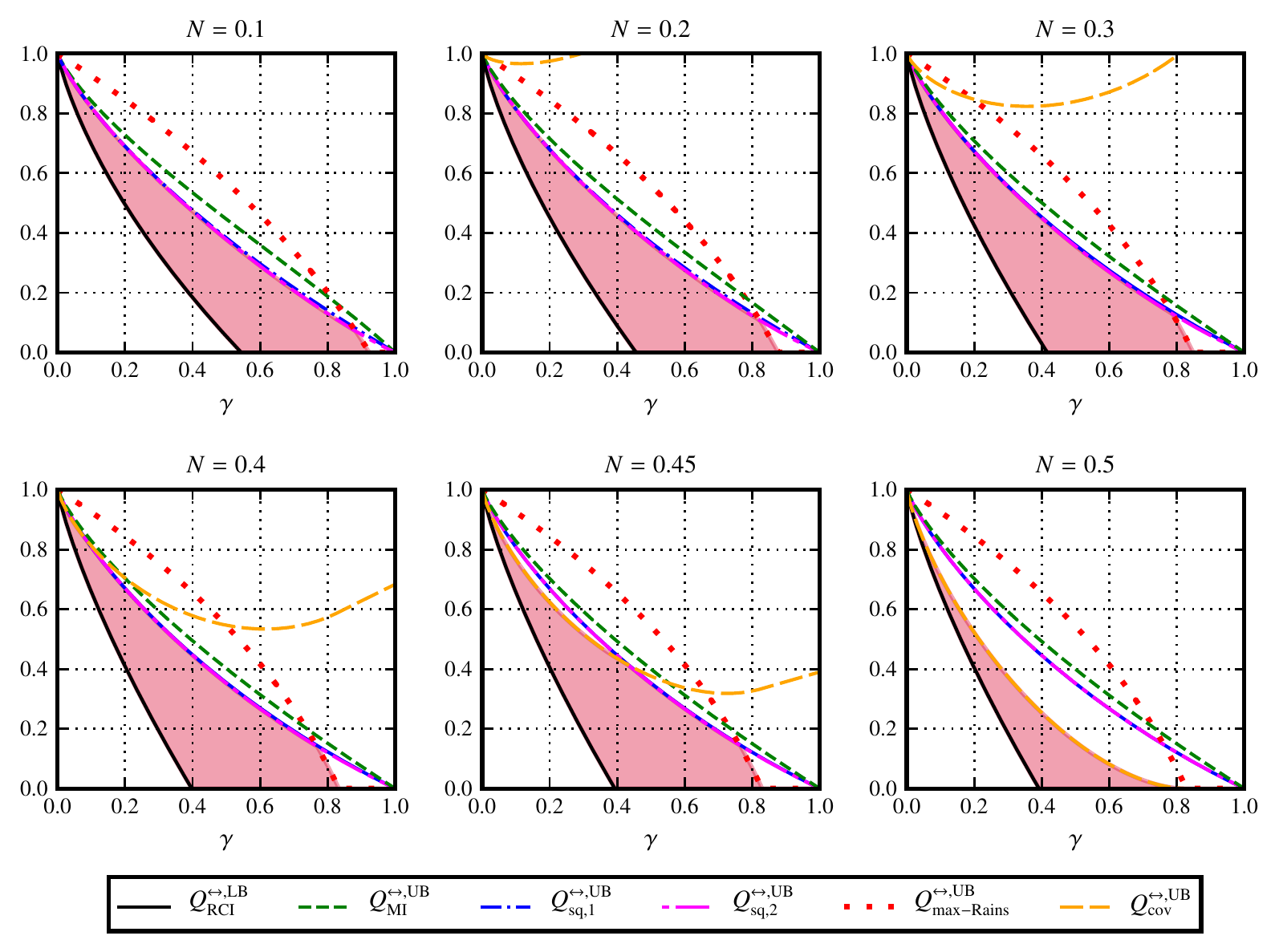}
		\caption{Bounds on the two-way assisted quantum capacity of the GADC. Shown is the reverse coherent information $Q_{\text{RCI}}^{\leftrightarrow,\text{LB}}$ lower bound, given by the expression in \eqref{eq-GADC_RCI}. We also plot the mutual information upper bound $Q_{\text{MI}}^{\leftrightarrow,\text{UB}}$ defined in \eqref{eq-GADC_Q2cap_UB1} and obtained by employing the identity squashing channel in the definition of the squashed entanglement, along with the squashed entanglement upper bounds $Q_{\text{sq},1}^{\leftrightarrow,\text{UB}}$ and $Q_{\text{sq},2}^{\leftrightarrow,\text{UB}}$ defined in \eqref{eq-GADC_Esq_UB1} and \eqref{eq-GADC_Esq_UB2}, respectively. The bounds $Q_{\text{sq},1}^{\leftrightarrow,\text{UB}}$ and $Q_{\text{sq},2}^{\leftrightarrow,\text{UB}}$ are obtained by employing the squashing channel as shown in \eqref{eq-GADC_Esq_UB_pf3}. The max-Rains upper bound $Q_{\text{max-Rains}}^{\leftrightarrow,\text{UB}}$ is given by the SDP in \eqref{eq-R_max}. (See also the analytic expression in \eqref{eq-GADC_Emax_Rmax}.) The upper bound $Q_{\text{cov}}^{\leftrightarrow,\text{UB}}$ is given in \eqref{eq:approx-tele-sim-bnd} and is based on the notion of approximate covariance. The two-way assisted quantum capacity lies within the shaded region.}\label{fig-2Way_QCap_bounds_all}
	\end{figure*}
	
	See Fig.~\ref{fig-2Way_QCap_bounds_all} for a plot of the squashed entanglement upper bounds in \eqref{eq-GADC_Esq_UB1} and \eqref{eq-GADC_Esq_UB2} along with the mutual information upper bound $E_{\text{sq}}(\mathcal{A}_{\gamma,N})\leq\frac{1}{2}I(\mathcal{A}_{\gamma,N})$, with $I(\mathcal{A}_{\gamma,N})$ given in \eqref{eq-GADC_mut_inf}. We also plot the reverse coherent information $I_{\text{rc}}(\mathcal{A}_{\gamma,N})$ lower bound. Due to Pauli-$z$ covariance and concavity of the reverse coherent information, $I_{\text{rc}}(\mathcal{A}_{\gamma, N})$ can be obtained by optimizing over diagonal input states, i.e.,
	\begin{equation}\label{eq-GADC_RCI}
		I_{\text{rc}}(\mathcal{A}_{\gamma,N})=\max_{p\in[0,1]}I_{\text{rc}}\left(\begin{pmatrix}1-p&0\\0&p\end{pmatrix},\mathcal{A}_{\gamma,N}\right)\equiv Q_{\text{RCI}}^{\leftrightarrow,\text{LB}}(\gamma,N).
	\end{equation}
	We note that the coherent information lower bound is not plotted in Fig.~\ref{fig-2Way_QCap_bounds_all} because it is smaller than the RCI lower bound for all values of $\gamma$ and $N$.

\subsection{Max-Rains and max-relative entropy of entanglement upper bounds}
	
	For the amplitude damping channel $\mathcal{A}_{\gamma,0}$, it has been shown in \cite[Proposition 2]{RKB+18} that
	\begin{equation}
		E_{\max}(\mathcal{A}_{\gamma,0})=\log_2(2-\gamma). \label{eq:E_max-GADC}
	\end{equation}
	We now generalize this formula to all values of $\gamma,N$ for the GADC. We also prove that the inequality opposite to the one in \eqref{eq-Rmax_Emax_ineq} holds for the GADC. As stated, this result generalizes the equality in \eqref{eq:E_max-GADC}, and the proof that we give is arguably simpler than that given for \cite[Proposition 2]{RKB+18}.
	
	\begingroup
	\allowdisplaybreaks[0] 
	\begin{proposition}\label{prop-GADC_Emax}
		For all $\gamma,N$ such that the GADC $\mathcal{A}_{\gamma,N}$ is not entanglement breaking, it holds that
		\begin{multline}\label{eq-GADC_Emax_Rmax}
			E_{\max}(\mathcal{A}_{\gamma,N})=R_{\max}(\mathcal{A}_{\gamma,N})\\=\log_2\left(1-\frac{\gamma}{2}+\frac{1}{2}\sqrt{(\gamma(2N-1))^2+4(1-\gamma)}\right).
		\end{multline}
		If the GADC is entanglement breaking, as given by \eqref{eq-GADC_ent_break}, then $E_{\max}(\mathcal{A}_{\gamma,N})=R_{\max}(\mathcal{A}_{\gamma,N})=0$.
	\end{proposition}
	\endgroup
	
	\begin{proof}
		See Appendix \ref{app-GADC_Emax}.
	\end{proof}
	
	By \eqref{eq-max_Rains_bound}, and using Proposition~\ref{prop-GADC_Emax}, we have that
	\begin{multline}\label{eq-GADC_Emax_Rmax_UB}
		Q^{\leftrightarrow}(\mathcal{A}_{\gamma,N}),P^{\leftrightarrow}(\mathcal{A}_{\gamma,N})\leq Q_{\text{max-Rains}}^{\leftrightarrow,\text{UB}}\\
		\equiv \log_2\left(1-\frac{\gamma}{2}+\frac{1}{2}\sqrt{(\gamma(2N-1))^2+4(1-\gamma)}\right).
	\end{multline}		
	for all $\gamma,N\in[0,1]$. In Fig.~\ref{fig-2Way_QCap_bounds_all}, we compare this max-Rains upper bound with the squashed entanglement upper bounds from the previous subsection. We observe that the max-Rains upper bound is tight when the channel is entanglement breaking. This is due to the fact that the state $\rho_{AB}$ for which $R_{\text{max}}(A;B)_{\rho}$ is evaluated in \eqref{eq-max_Rains_channel} is separable whenever the channel is entanglement breaking, and the fact that any separable state is in the set $\text{PPT}'$.

\subsection{Approximate covariance upper bounds}

	Applying the bounds in Eq.~\eqref{eq-Q_cap_two_way_cov} and Eq.~\eqref{eq-P_cap_two_way_cov} to the GADC, recalling from \eqref{eq-GADC_twirl} that $\mathcal{A}_{\gamma,N}^G=\mathcal{A}_{\gamma,\frac{1}{2}}$, and using the fact that the quantity $R(A;B)_{\rho}$ coincides with $E_R(A;B)_{\rho}$ for qubit-qubit states $\rho_{AB}$ \cite[Section III]{AS08}, these bounds reduce to the following:%
	\begin{align}
		Q^{\leftrightarrow}(\mathcal{A}_{\gamma,N})&, P^{\leftrightarrow}(\mathcal{A}_{\gamma,N})\leq Q_{\text{cov}}^{\leftrightarrow,\text{UB}}(\gamma,N)\nonumber\\
		&\equiv E_{R}(A;B)_{\rho}+2\varepsilon_{\text{cov}}+g(\varepsilon_{\text{cov}})\label{eq:approx-tele-sim-bnd},
	\end{align}
	where $\varepsilon_{\text{cov}}\equiv\varepsilon_{\text{cov}}(\mathcal{A}_{\gamma,N})=\gamma\left|N-\frac{1}{2}\right|$ and%
	\begin{align}
		\rho_{AB}^{\gamma}&\equiv\mathcal{A}_{\gamma,\frac{1}{2}}(\Phi_{AA'}^+)\\
		& =\frac{1}{2}%
		\begin{pmatrix}
			1-\frac{\gamma}{2} & 0 & 0 & \sqrt{1-\gamma}\\
			0 & \frac{\gamma}{2} & 0 & 0\\
			0 & 0 & \frac{\gamma}{2} & 0\\
			\sqrt{1-\gamma} & 0 & 0 & 1-\frac{\gamma}{2}%
		\end{pmatrix}.
	\end{align}
	Note that, due to \eqref{eq-GADC_ent_break}, $\rho_{AB}^{\gamma}$ is entangled only when $0\leq\gamma<2(\!\sqrt{2}-1)$. In this case, it is a Bell-diagonal state of the form:%
	\begin{equation}
		\rho_{AB}^{\gamma}=\sum_{i,j=0}^1 r_{i,j}\ket{\Phi_{i,j}}\bra{\Phi
_{i,j}}_{AB},
	\end{equation}
	with $\ket{\Phi_{i,j}}_{AB}\equiv\left(\mathbbm{1}_{A}\otimes \sigma_x^{i}\sigma_z^{j}\right)\ket{\Phi^+}_{AB}$ and%
	\begin{align}
		r_{0,0}  & =\frac{1}{4}\left(  2+2\sqrt{1-\gamma}-\gamma\right)  \\
		r_{0,1}  & =\frac{1}{4}\left(  2-2\sqrt{1-\gamma}-\gamma\right)  ,\\
		r_{1,0}  & =r_{11}=\frac{\gamma}{4}.
	\end{align}
	The closest separable state for such a Bell-diagonal state with $r_{0,0}\geq\frac{1}{2}$ is well known to have the form \cite{VPRK97} (see also \cite{AS08})
	\begin{align}
		\sigma_{AB}& =\frac{1}{2}\ket{\Phi_{0,0}}\bra{\Phi_{0,0}}_{AB}\nonumber\\
		&\qquad +\frac
{1}{2(1-r_{0,0})}\sum_{i,j\neq\left(0,0\right)}r_{i,j}%
\ket{\Phi_{i,j}}\bra{\Phi_{i,j}}_{AB}\\
		& =%
		\begin{pmatrix}
			\frac{1}{2}-x & 0 & 0 & x\\
			0 & x & 0 & 0\\
			0 & 0 & x & 0\\
			x & 0 & 0 & \frac{1}{2}-x
		\end{pmatrix},
	\end{align}
	where%
	\begin{equation}
		x=\frac{\gamma}{2\left(  2-2\sqrt{1-\gamma}+\gamma\right)  }.
	\end{equation}
	We then find that%
	\begin{multline}\label{eq-GADC_REE}
		E_{R}(A;B)_{\rho}\\
		=\sum_{i,j=0}^1 r_{i,j}\log_{2}r_{i,j}+1-\frac{\gamma}{2}\log_{2}\left(\frac{\gamma}{2-2\sqrt{1-\gamma}+\gamma}\right)  \\
		+\frac{\gamma-2+2\sqrt{1-\gamma}}{4}\log_{2}\left(  \frac{4-\gamma-4\sqrt{1-\gamma}}{8+\gamma}\right)  ,
	\end{multline}
	which completes the analytic form of the bound in \eqref{eq:approx-tele-sim-bnd}. Note that this formula for $E_R(A;B)_{\rho}$ holds only for $\gamma\in [0, 2(\!\sqrt{2}-1))$; otherwise, $\rho_{AB}^{\gamma}$ is separable, which means that $E_R(A;B)_{\rho}=0$. We also note that for $N=\frac{1}{2}$, which is when the GADC is covariant with respect to the Pauli group and thus $\varepsilon_{\text{cov}}=0$, the bound in \eqref{eq:approx-tele-sim-bnd} reduces to $Q^{\leftrightarrow}(\mathcal{A}_{\gamma,N}),P^{\leftrightarrow}(\mathcal{A}_{\gamma,N})\leq E_R(A;B)_{\rho}$, which is precisely the bound determined in \cite[Theorem~5]{PLOB17} and in \cite[Theorem~12]{WTB17} for the class of teleportation-simulable channels. (Any channel that is covariant with respect to the Pauli group is teleportation-simulable; see, e.g., \cite[Appendix~A]{WTB17}.)
	
	In Fig.~\ref{fig-2Way_QCap_bounds_all}, we plot the bound $Q_{\text{cov}}^{\leftrightarrow,\text{UB}}$ in \eqref{eq:approx-tele-sim-bnd}. While the bound is relatively poor for small values of $N$, for values of $N$ close to $\frac{1}{2}$ we find that it is tighter than the other upper bounds for some values of $\gamma$. Notably, at $N=\frac{1}{2}$, this upper bound is the tightest among the other upper bounds, and by a significant margin as well.

\section{Conclusion}\label{sec:conclusion}

	In this work, we provided an information-theoretic study of the generalized amplitude damping channel (GADC), which is a generalized form of the well-known amplitude damping channel and can be thought of as the qubit analogue of the bosonic thermal channel. We first determined the range of parameters for which the channel is entanglement breaking, as well as the range of parameters for which it is anti-degradable. 
	
	
	We then established several upper bounds on the classical capacity of the GADC. We used the concepts of approximate covariance and approximate entanglement-breakability \cite{LKDW18} to obtain upper bounds. We compared these upper bounds with known SDP-based upper bounds \cite{WXD18}, for which we proved an analytical formula for the GADC, as well as the known entanglement-assisted classical capacity upper bound \cite{LM07}. 
	
	We also provided several upper bounds on the quantum and private capacities of the GADC. We exploited the two decompositions of the GADC in \eqref{eq-GADC_decomp_spec_1} and \eqref{eq-GADC_decomp_spec_2} in terms of amplitude damping channels in order to obtain data-processing upper bounds, and we used the concepts of approximate degradability and approximate anti-degradability \cite{SSWR17} to obtain further upper bounds. We found that one of the data-processing upper bounds is tighter than the recently obtained upper bound from \cite{RMG18} for all parameter values of the GADC, and that the Rains information upper bound is tighter than the upper bound from \cite{RMG18} for certain parameter regimes.
	
	We also considered the two-way assisted quantum and private capacities of the GADC. We determined upper bounds on these capacities using the squashed entanglement \cite{TGW14a,TGW14b}, the max-Rains information \cite{BW18}, and the max-relative entropy of entanglement \cite{CMH17}. The squashed entanglement upper bounds exploited the decompositions of the GADC in \eqref{eq-GADC_decomp_spec_1} and \eqref{eq-GADC_decomp_spec_2}, as well as a particular choice of squashing channel. This allowed us to obtain upper bounds that are better than the mutual information bound that can be obtained via the identity squashing channel. We also obtained upper bounds using the concept of approximate covariance. Along the way, we also determined an analytic form for both the max-Rains information $R_{\max}$ and the max-relative entropy of entanglement $E_{\max}$ of the GADC, and we found that for the GADC both quantities are equal to each other. In light of the latter result, it is worth exploring whether the equality $R_{\max}(\mathcal{N})=E_{\max}(\mathcal{N})$ holds for all qubit-to-qubit channels $\mathcal{N}$.
	
	Obtaining the communication capacities of the GADC for its entire parameter range remains a challenging open problem. This work has applied many state-of-the-art techniques to obtain upper bounds, and it is clear that obtaining tighter upper bounds, or even to obtain an exact expression for the capacity, will require new techniques. To this end, some directions for future work include: employing a different squashing channel than the one used here to obtain a better upper bound on the two-way assisted quantum and private capacities of the GADC. Another method to reduce the gap between lower and upper bounds for any communication scenario is to look at improving current lower bounds rather than upper bounds, via potential superadditivity effects.

\begin{acknowledgments}

	All authors acknowledge support from the National Science Foundation. Also, SK acknowledges the NSERC PGS-D, and MMW the Office of Naval Research.

\end{acknowledgments}

\appendix

\section{Proof of Proposition~\ref{prop:concavity-Rains}}

\label{proof-prop:concavity-Rains}

The proof is similar in spirit to \cite[Proposition~2]{TWW17}, and in fact implies it for
the relative entropy. Let $\psi_{AA^{\prime}}^{0}$ and $\psi_{AA^{\prime}}%
^{1}$ be pure states and define%
\begin{equation}
\psi_{A^{\prime}}^{\lambda}\equiv\left(  1-\lambda\right)  \psi_{A^{\prime}%
}^{0}+\lambda\psi_{A^{\prime}}^{1},
\end{equation}
for $\lambda\in\left[  0,1\right]  $. A purification of $\psi_{A^{\prime}%
}^{\lambda}$ is given by%
\begin{equation}
\ket{\psi}_{PAA'}^{\lambda}\equiv\sqrt{1-\lambda}|0\rangle_{P}%
|\psi^{0}\rangle_{AA^{\prime}}+\sqrt{\lambda}|1\rangle_{P}|\psi^{1}%
\rangle_{AA^{\prime}}.
\end{equation}
This purification is related to another purification $\phi_{AA^{\prime}%
}^{\lambda}$ by an isometric channel $\mathcal{U}_{A\rightarrow PA}$:
$\psi_{PAA^{\prime}}^{\lambda}=\mathcal{U}_{A\rightarrow PA}(\phi_{AA^{\prime
}}^{\lambda})$. Let $\sigma_{AB}^{\lambda}\in\text{PPT}'(A\!:\!B)$ be the operator such that $R(\mathcal{N}_{A'\to B}(\phi_{AA'}^{\lambda}))\equiv R(A;B)_{\rho^{\lambda}}=D(\mathcal{N}_{A'\to B}(\phi_{AA'}^{\lambda})\Vert\sigma_{AB}^{\lambda})$, where $\rho_{AB}^{\lambda}=\mathcal{N}_{A'\to B}(\phi_{AA'}^{\lambda})$, and define $\xi_{PAB}^{\lambda}=\mathcal{U}_{A\rightarrow PA}(\sigma_{AB}^{\lambda})$. Observe that $\xi_{PAB}^{\lambda}\in\text{PPT}'(PA\!:\!B)$. Let
\begin{equation}
\overline{\Delta}_{P}(\xi_{PAB}^{\lambda})=q|0\rangle\langle0|_{P}\otimes\tau_{AB}%
^{0}+\left(  1-q\right)  |1\rangle\langle1|_{P}\otimes\tau_{AB}^{1},
\end{equation}
where $\overline{\Delta}_{P}$ is a completely dephasing channel, defined as%
\begin{align}
\overline{\Delta}_{P}(\cdot)  & \equiv|0\rangle\langle0|_{P}(\cdot)|0\rangle\langle
0|_{P}+|1\rangle\langle1|_{P}(\cdot)|1\rangle\langle1|_{P},\\
q  & \equiv\operatorname{Tr}[\left(  |0\rangle\langle0|_{P}\otimes
\mathbbm{1}_{AB}\right)  \xi_{PAB}^{\lambda}],\\
\tau_{AB}^{0}  & \equiv\frac{1}{q}\operatorname{Tr}_{P}[\left(  |0\rangle
\langle0|_{P}\otimes \mathbbm{1}_{AB}\right)  \xi_{PAB}^{\lambda}],\\
\tau_{AB}^{1}  & \equiv\frac{1}{1-q}\operatorname{Tr}_{P}[\left(
|1\rangle\langle1|_{P}\otimes \mathbbm{1}_{AB}\right)  \xi_{PAB}^{\lambda}].
\end{align}
Note that the states $\tau_{AB}^{0}$ and $\tau_{AB}^{1}$ are in the set
PPT$^{\prime}(A\!:\!B)$ since $\xi_{PAB}^{\lambda}$ is in $\text{PPT}'(PA\!:\!B)$. Then we
have that%
\begin{align}
& R(\mathcal{N}_{A^{\prime}\rightarrow B}(\phi_{AA^{\prime}}^{\lambda
}))\nonumber\\
& =D(\mathcal{N}_{A^{\prime}\rightarrow B}(\phi_{AA^{\prime}}^{\lambda}%
)\Vert\sigma_{AB}^{\lambda})\\
& =D(\mathcal{N}_{A^{\prime}\rightarrow B}(\psi_{PAA^{\prime}}^{\lambda}%
)\Vert\xi_{PAB}^{\lambda})\\
& \geq D(\overline{\Delta}_{P}(\mathcal{N}_{A^{\prime}\rightarrow B}(\psi_{PAA^{\prime}%
}^{\lambda}))\Vert\overline{\Delta}_P(\xi_{PAB}^{\lambda}))\\
& =D(\mathcal{N}_{A^{\prime}\rightarrow B}(\overline{\Delta}_{P}(\psi_{PAA^{\prime}%
}^{\lambda}))\Vert\overline{\Delta}_{P}(\xi_{PAB}^{\lambda}))\\
& =\left(  1-\lambda\right)  D(\mathcal{N}_{A^{\prime}\rightarrow B}%
(\psi_{AA^{\prime}}^{0})\Vert\tau_{AB}^{0}) +\lambda D(\mathcal{N}_{A^{\prime}\rightarrow B}(\psi_{AA^{\prime}}%
^{1})\Vert\tau_{AB}^{1})\notag \\
& \qquad +D(\left\{  1-\lambda,\lambda\right\}  \Vert\left\{
q,1-q\right\}  )\\
& \geq\left(  1-\lambda\right)  D(\mathcal{N}_{A^{\prime}\rightarrow B}%
(\psi_{AA^{\prime}}^{0})\Vert\tau_{AB}^{0})\nonumber\\
& \qquad +\lambda D(\mathcal{N}_{A^{\prime}\rightarrow B}(\psi_{AA^{\prime}}%
^{1})\Vert\tau_{AB}^{1})\\
& \geq\left(  1-\lambda\right)  R(\mathcal{N}_{A^{\prime}\rightarrow B}%
(\psi_{AA^{\prime}}^{0}))+\lambda R(\mathcal{N}_{A^{\prime}\rightarrow B}%
(\psi_{AA^{\prime}}^{1})).
\end{align}
The second equality follows from the isometric invariance of the relative
entropy. The first inequality follows from the data processing property of
relative entropy. The fourth equality follows from the identity
\cite[Exercise~11.8.8]{W17}%
\begin{equation}
D(\rho_{XB}\Vert\sigma_{XB})=\sum_{x}p(x)D(\rho_{B}^{x}\Vert\sigma_{B}%
^{x})+D(p\Vert r),
\end{equation}
holding for classical-quantum states%
\begin{align}
\rho_{XB}  & =\sum_{x}p(x)|x\rangle\langle x|_{X}\otimes\rho_{B}^{x},\\
\sigma_{XB}  & =\sum_{x}r(x)|x\rangle\langle x|_{X}\otimes\sigma_{B}^{x}.
\end{align}
Note that $D(p\Vert r)$ denotes the classical relative entropy of the probability distributions $p$ and $r$. For binary probability distributions such that $p(0)=1-\lambda$, $p(1)=\lambda$, $r(0)=1-q$, $r(1)=q$, we let $D(\left\{1-\lambda,\lambda\right\}\Vert\left\{1-q,q\right\})\equiv D(p\Vert r)$. The second inequality follows from the non-negativity of the relative entropy. The final inequality follows because the Rains relative entropy involves a
minimization over all states in PPT$^{\prime}(A:B)$.

A proof for the concavity statement for the relative entropy of entanglement
$E_{R}(A;B)_{\omega}$ is identical, except replacing PPT$^{\prime}(A\!:\!B)$ with
SEP$(A\!:\!B)$.

\section{Proof of Lemma~\ref{lem-anti_degrad_chan_bd}}\label{app-GADC_anti_degrade}

Let $\E^*$, $E_0$, and $E_1$ be as defined in the statement of Lemma~\ref{lem-anti_degrad_chan_bd}.
	Let $V_{A\to BE}^{\gamma,N}$ be the isometric extension of the GADC defined in \eqref{eq-GADC_iso_ext}, and define the pure state
	\begin{align}
		\ket{\psi}_{ABE}^{\gamma,N}&\equiv (\mathbbm{1}_A\otimes V_{A'\to BE}^{\gamma,N})\ket{\Phi^+}_{AA'}\\
		&=\frac{1}{\sqrt{2}}\left(\sqrt{1-N}\ket{0,0,0}_{ABE}+\sqrt{N(1-\gamma)}\ket{0,0,2}_{ABE}\right.\nonumber\\
		&\quad \left.+\sqrt{N\gamma}\ket{0,1,3}_{ABE}+\sqrt{(1-\gamma)(1-N)}\ket{1,1,0}_{ABE}\right.\nonumber\\
		&\quad \left.+\sqrt{N}\ket{1,1,2}_{ABE}+\sqrt{\gamma(1-N)}\ket{1,0,1}_{ABE}\right)
	\end{align}
	Then, $\rho_{AB}^{\gamma,N}\equiv\Tr_E[\ket{\psi}\bra{\psi}_{ABE}^{\gamma,N}]$ is the Choi state of the GADC $\mathcal{A}_{\gamma,N}$, while $\rho_{AE}^{\gamma,N}\equiv \Tr_B[\ket{\psi}\bra{\psi}_{ABE}^{\gamma,N}]$ is the Choi state of the complementary channel $\mathcal{A}_{\gamma,N}^c$ as defined in \eqref{eq-GADC_comp}. In order to prove that $\mathcal{E}_N^*\circ\mathcal{A}_{\gamma,N}^c=\mathcal{A}_{1-\gamma,N}$, it suffices to show that $(\mathcal{E}_N^*)_{E\to B'}(\rho_{AE}^{\gamma,N})=\rho_{AB}^{1-\gamma,N}$. In other words, it suffices to show that the Choi state of the complementary channel $\mathcal{A}_{\gamma,N}^c$ is mapped to the Choi state of the channel $\mathcal{A}_{1-\gamma,N}$.
	
	We have
	\begin{equation}
		\begin{aligned}
		\rho_{AB}^{\gamma,N}&=\frac{1}{2}\left((1-\gamma N)\ket{0,0}\bra{0,0}_{AB}+\sqrt{1-\gamma}\ket{0,0}\bra{1,1}_{AB}\right.\\
		&\qquad\left.+\sqrt{1-\gamma}\ket{1,1}\bra{0,0}_{AB}+\gamma N\ket{0,1}\bra{0,1}_{AB}\right.\\
		&\qquad\left.+\gamma(1-N)\ket{1,0}\bra{1,0}_{AB}\right.\\
		&\qquad\left.+(1-\gamma(1-N))\ket{1,1}\bra{1,1}_{AB}\right).
		\end{aligned}
	\end{equation}
	Let an isometric extension of the channel $\mathcal{E}_N^*$ be
	\begin{equation}
		W^{\mathcal{E}_N^*}_{E\to B'E'}=E_0\otimes\ket{0}_{E'}+E_1\otimes\ket{1}_{E'}.
	\end{equation}
	Then,
	\begin{equation}
		\begin{aligned}
		\ket{\phi}_{ABB'E'}^{\gamma,N}&\equiv W^{\mathcal{E}_N^*}_{E\to B'E'}\ket{\psi}_{ABE}^{\gamma,N}\\
		&=\frac{1}{\sqrt{2}}\left(\sqrt{1-N}\ket{0,0,0,0}_{ABB'E'}\right.\\
		&\qquad\left.+\sqrt{N(1-\gamma)}\ket{0,0,1,1}_{ABB'E}\right.\\
		&\qquad\left.+\sqrt{N\gamma}\ket{0,1,0,1}_{ABB'E'}\right.\\
		&\qquad\left.+\sqrt{(1-\gamma)(1-N)}\ket{1,1,0,0}_{ABB'E'}\right.\\
		&\qquad\left.+\sqrt{N}\ket{1,1,1,1}_{ABB'E'}\right.\\
		&\qquad\left.+\sqrt{\gamma(1-N)}\ket{1,0,1,0}_{ABB'E'}\right).
		\end{aligned}
	\end{equation}
	Then,
	\begin{equation}
		\begin{aligned}
		&\Tr_{BE'}[\ket{\phi}\bra{\phi}_{ABB'E'}^{\gamma,N}]=(\mathcal{E}_N^*)_{E\to B'}(\rho_{AE}^{\gamma,N})\\
		&=\frac{1}{2}\left((1-(1-\gamma)N)\ket{0,0}\bra{0,0}_{AB'}+\sqrt{\gamma}\ket{0,0}\bra{1,1}_{AB'}\right.\\
		&\qquad\left.+\sqrt{\gamma}\ket{1,1}\bra{0,0}_{AB'}+N(1-\gamma)\ket{0,1}\bra{0,1}_{AB'}\right.\\
		&\qquad\left.+(1-\gamma)(1-N)\ket{1,0}\bra{1,0}_{AB'}\right.\\
		&\qquad\left.+(N+\gamma(1-N))\ket{1,1}\bra{1,1}_{AB'}\right)\\
		&=\rho_{AB'}^{1-\gamma,N},
		\end{aligned}
	\end{equation}
	as required.

\section{Proof of Proposition \ref{prop-C_beta}}\label{app-C_beta}

	We start by recalling the convex decomposition of the GADC as stated in \eqref{eq-GADC_decomp_convex}:
	\begin{equation}
		\mathcal{A}_{\gamma,N}=(1-N)\mathcal{A}_{\gamma,0}+N\mathcal{A}_{\gamma,1}
	\end{equation}
	for all $\gamma,N\in[0,1]$. We also recall from \eqref{eq-GADC_N_symmetry} that
	\begin{equation}\label{eq-GADC_N_symmetry_2}
		\mathcal{A}_{\gamma,1}(\rho)=\sigma_x\mathcal{A}_{\gamma,0}(\sigma_x\rho\sigma_x)\sigma_x
	\end{equation}
	for all $\gamma\in[0,1]$. Next, note that it follows from \eqref{eq-C_beta_primal} that the quantity $\beta(\mathcal{N})$ in the definition of $C_\beta(\mathcal{N})$ is convex in the channel $\N$: for any two channels $\mathcal{N}_1$ and $\mathcal{N}_2$ and any $\lambda\in[0,1]$,
	\begin{equation}
		\beta(\lambda\mathcal{N}_1+(1-\lambda)\mathcal{N}_2)\leq \lambda\beta(\mathcal{N}_1)+(1-\lambda)\beta(\mathcal{N}_2).
	\end{equation}
	Furthermore, $\beta(\mathcal{N})$ is invariant under pre- and post-processing of the channel $\mathcal{N}$ by unitaries. Therefore,
	\begin{align}
		C_\beta(\mathcal{A}_{\gamma,N})&=C_{\beta}((1-N)\mathcal{A}_{\gamma,0}+N\mathcal{A}_{\gamma,1})\\
		&\leq (1-N)C_{\beta}(\mathcal{A}_{\gamma,0})+NC_{\beta}(\mathcal{A}_{\gamma,1})\\
		&=C_{\beta}(\mathcal{A}_{\gamma,0}),
	\end{align}
	where to obtain the last line we used \eqref{eq-GADC_N_symmetry_2} and the invariance of $C_{\beta}$ under pre- and post-processing of the given channel by unitaries to find that $C_{\beta}(\mathcal{A}_{\gamma,1}) = C_{\beta}(\mathcal{A}_{\gamma,0})$.
	
	Given the facts above, our proof strategy is as follows. First, we provide an upper bound of $1+\sqrt{1-\gamma}$ for the SDP in \eqref{eq-C_beta_primal} in the case $N=0$, i.e., for the amplitude damping channel, which establishes that $C_{\beta}(\mathcal{A}_{\gamma,N})\leq\log_2(1+\sqrt{1-\gamma})$. Next, we consider the SDP dual to the one in \eqref{eq-C_beta_primal} and prove that $1+\sqrt{1-\gamma}$ is a lower bound on it. By strong duality, it follows that $C_{\beta}(\mathcal{A}_{\gamma,N})=\log_2(1+\sqrt{1-\gamma})$ for all $\gamma,N\in\left[0,1\right]$.
	
	We first recall from \eqref{eq-C_beta_primal} that
	\begin{equation}\label{eq:c_beta_quantity}
		\beta(\mathcal{N})=\left\{\begin{array}{l l}\text{min.} & \Tr[S_B]\\
			\text{subject to} & -R_{AB}\leq \left(\Gamma_{AB}^{\mathcal{N}}\right)^{\t_B}\leq R_{AB},\\
			& -\mathbbm{1}_A\otimes S_B\leq R_{AB}^{\t_B}\leq\mathbbm{1}_A\otimes S_B,
		\end{array}\right.
	\end{equation}
	where the optimization is with respect to the Hermitian operators $S_{B}$ and $R_{AB}$. Note that it follows from the above constraints that $S_{B},R_{AB}\geq 0$.

	As a matrix in the standard basis, the Choi matrix for the amplitude damping channel is (see \eqref{eq-GADC_Choi_state}) 
	\begin{equation}
		\Gamma_{AB}^{\gamma,0}=2\rho_{AB}^{\gamma,0}=%
		\begin{pmatrix}
			1 & 0 & 0 & \sqrt{1-\gamma}\\
			0 & 0 & 0 & 0\\
			0 & 0 & \gamma & 0\\
			\sqrt{1-\gamma} & 0 & 0 & 1-\gamma
		\end{pmatrix},
	\end{equation}
	so that the partial transpose is given by%
	\begin{equation}
		(\Gamma_{AB}^{\gamma,0})^{\t_B}=%
		\begin{pmatrix}
			1 & 0 & 0 & 0\\
			0 & 0 & \sqrt{1-\gamma} & 0\\
			0 & \sqrt{1-\gamma} & \gamma & 0\\
			0 & 0 & 0 & 1-\gamma
		\end{pmatrix}.
	\end{equation}
	Let us choose the operators $R_{AB}$ and $S_B$ to be
	\begin{align}
		R_{AB} &  =%
		\begin{pmatrix}
			1 & 0 & 0 & 0\\
			0 & 1-\gamma+a & a & 0\\
			0 & a & 1+a & 0\\
			0 & 0 & 0 & 1-\gamma
		\end{pmatrix},\\
		S_{B} &  =%
		\begin{pmatrix}
			1+a & 0\\
			0 & 1-\gamma+a
		\end{pmatrix},
	\end{align}
	where $a=\frac{1}{2}(\sqrt{1-\gamma}-\left(  1-\gamma\right))$. We first check that the constraint $-R_{AB}\leq \left(\Gamma_{AB}^{\gamma,0}\right)^{\t_B}\leq R_{AB}$ is satisfied. Consider that%
	\begin{equation}
		R_{AB}-\left(\Gamma_{AB}^{\gamma,0}\right)^{\t_B}=%
		\begin{pmatrix}
			0 & 0 & 0 & 0\\
			0 & b & -b & 0\\
			0 & -b & b & 0\\
			0 & 0 & 0 & 0
		\end{pmatrix},
	\end{equation}
	where%
	\begin{equation}
		b=\frac{1}{2}\left(\sqrt{1-\gamma}+\left(  1-\gamma\right)\right).
	\end{equation}
	Due to the inequality $b\geq 0$ for all $\gamma\in\left[0,1\right]  $ and the fact that $\begin{pmatrix}1 & -1\\-1 & 1 \end{pmatrix} \geq 0$, it follows that $R_{AB}-\left(\Gamma_{AB}^{\gamma,0}\right)^{\t_B}\geq 0$. We also have that%
	\begin{equation}\label{eq-C_beta_pf}
		R_{AB}+\left(\Gamma_{AB}^{\gamma,0}\right)^{\t_B}=%
		\begin{pmatrix}
			2 & 0 & 0 & 0\\
			0 & b & \sqrt{1-\gamma}+a & 0\\
			0 & \sqrt{1-\gamma}+a & 1+\gamma+a & 0\\
			0 & 0 & 0 & 2\left(  1-\gamma\right)
		\end{pmatrix}.
	\end{equation}
	To determine whether $R_{AB}+\left(\Gamma_{AB}^{\gamma,0}\right)^{\t_B}\geq0$, it is clear that we can focus on the inner $2\times 2$ matrix. For the cases $\gamma=0$ or $\gamma=1$, one can directly confirm the condition $R_{AB}+\left(\Gamma_{AB}^{\gamma,0}\right)^{\t_B}\geq 0$. A general $2\times 2$ matrix is positive definite if and only its trace and determinant are strictly positive. The trace of the inner $2\times 2$ matrix in \eqref{eq-C_beta_pf} is%
	\begin{equation}
		\sqrt{1-\gamma}+1+\gamma>0
	\end{equation}
	for all $\gamma\in(0,1)$, and its determinant is%
	\begin{equation}
		\left(  2-\gamma\right)  \left(  \sqrt{1-\gamma}-(1-\gamma)\right)>0
	\end{equation}
	for all $\gamma\in(0,1)$. It thus follows that $R_{AB}+\left(\Gamma_{AB}^{\gamma,0}\right)^{\t_B}>0$ for all $\gamma \in\left(  0,1\right)  $.

	We now check the conditions $-\mathbbm{1}_{A}\otimes S_{B}\leq R_{AB}^{\t_B}\leq \mathbbm{1}_{A}\otimes S_{B}$. Consider that%
	\begin{align}
		\mathbbm{1}_{A}\otimes S_{B} &  =%
		\begin{pmatrix}
			1+a & 0 & 0 & 0\\
			0 & 1-\gamma+a & 0 & 0\\
			0 & 0 & 1+a & 0\\
			0 & 0 & 0 & 1-\gamma+a
		\end{pmatrix},\\
		R_{AB}^{\t_B} &  =%
		\begin{pmatrix}
			1 & 0 & 0 & a\\
			0 & 1-\gamma+a & 0 & 0\\
			0 & 0 & 1+a & 0\\
			a & 0 & 0 & 1-\gamma
		\end{pmatrix}.
	\end{align}
	Then%
	\begin{equation}
		\mathbbm{1}_{R}\otimes S_{B}-R_{AB}^{\t_B}=%
		\begin{pmatrix}
			a & 0 & 0 & -a\\
			0 & 0 & 0 & 0\\
			0 & 0 & 0 & 0\\
			-a & 0 & 0 & a
		\end{pmatrix}.
	\end{equation}
	Due to the fact that $a\geq 0$ for all $\gamma\in\left[0,1\right]  $, it follows that $\mathbbm{1}_{A}\otimes S_{B}-R_{AB}^{\t_B}\geq0$. We also need to consider%
	\begin{multline}
		\mathbbm{1}_{R}\otimes S_{B}+R_{AB}^{\t_B}\\=%
		\begin{pmatrix}
			a+2 & 0 & 0 & a\\
			0 & 2a+2\left(  1-\gamma\right)   & 0 & 0\\
			0 & 0 & 2a+2 & 0\\
			a & 0 & 0 & a+2\left(  1-\gamma\right)
		\end{pmatrix}.
	\end{multline}
	We have that $2a+2\left(  1-\gamma\right)  \geq0$ and $2a+2\geq0$ for all $\gamma\in\left[  0,1\right]  $. Thus, to determine whether $\mathbbm{1}_{A}\otimes S_{B}+R_{AB}^{\t_B}\geq 0$, it is clear that we can focus on the ``corners'' $2\times 2$ submatrix:%
	\begin{equation}%
		\begin{pmatrix}
			a+2 & a\\
			a & a+2\left(  1-\gamma\right)
		\end{pmatrix}.
	\end{equation}
	For $\gamma=0$ or $\gamma=1$, one can directly confirm that this corners submatrix is positive semi-definite. For $\gamma\in\left(0,1\right)$, the trace of the corners submatrix is%
	\begin{equation}
		3-\gamma+\sqrt{1-\gamma}>0,
	\end{equation}
	and its determinant is given by%
	\begin{equation}
		\left(  1-\gamma\right)  \left(  2+\gamma\right)  +\left(  2-\gamma\right) \sqrt{1-\gamma}>0
	\end{equation}
	for all $\gamma\in(0,1)$. It thus follows that $\mathbbm{1}_{A}\otimes S_{B}+R_{AB}^{\t_B}>0$ for all $\gamma \in\left(  0,1\right)  $. Thus, the proposed operators $R_{AB}$ and $S_{B}$ satisfy the given constraints in \eqref{eq:c_beta_quantity}, and we
conclude that%
	\begin{align}
		\beta(\mathcal{A}_{\gamma,0}) &  \leq\Tr[S_{B}]\\
		&  =1+a+1-\gamma+a\\
		&  =2+2a-\gamma\\
		&  =2+2\frac{1}{2}\left(\sqrt{1-\gamma}-\left(  1-\gamma\right)\right)-\gamma\\
		&  =1+\sqrt{1-\gamma}.
	\end{align}
	By the arguments presented at the beginning of the proof, we thus conclude that%
	\begin{equation}\label{eq:c_beta-GADC-lower-bnd}
		C_{\beta}(\mathcal{A}_{\gamma,N})\leq\log_{2}(1+\sqrt{1-\gamma})
	\end{equation}
	for all $\gamma,N\in\left[  0,1\right]  $.

	The SDP dual to the one in \eqref{eq:c_beta_quantity} is given by%
	\begin{equation}\label{eq:c-beta-dual}
		\hat{\beta}(\mathcal{N})\equiv\left\{\begin{array}{l l} \text{max.} & \Tr[\Gamma_{AB}^{\mathcal{N}}(K_{AB}-M_{AB})^{\t_B}]\\
		\text{subject to} & K_{AB}+M_{AB}\leq (E_{AB}-F_{AB})^{\t_B},\\
		& E_B+F_B\leq\mathbbm{1}_B,\\
		& K_{AB},M_{AB},E_{AB},F_{AB}\geq 0. \end{array}\right.
	\end{equation}
	From \eqref{eq-GADC_Choi_state} we have that the Choi matrix for the GADC is
	\begin{equation}
		\Gamma_{AB}^{\gamma,N}=%
		\begin{pmatrix}
			1-\gamma N & 0 & 0 & \sqrt{1-\gamma}\\
			0 & \gamma N & 0 & 0\\
			0 & 0 & \gamma\left(  1-N\right)   & 0\\
			\sqrt{1-\gamma} & 0 & 0 & 1-\gamma\left(  1-N\right)
		\end{pmatrix}.
	\end{equation}
	Let us now make the following choice for the operators $K_{AB},M_{AB},E_{AB},F_{AB}$:
	\begin{align}
		K_{AB} &  =\frac{1}{2}%
		\begin{pmatrix}
			1 & 0 & 0 & 0\\
			0 & 1 & 1 & 0\\
			0 & 1 & 1 & 0\\
			0 & 0 & 0 & 1
		\end{pmatrix},\qquad M_{AB}=0,\\
		E_{AB} &  =\frac{1}{2}%
		\begin{pmatrix}
			1 & 0 & 0 & 1\\
			0 & 1 & 0 & 0\\
			0 & 0 & 1 & 0\\
			1 & 0 & 0 & 1
		\end{pmatrix},\qquad F_{AB}=0.
	\end{align}
	We find that $K_{AB}=E_{AB}^{\t_B}$ and $E_{B}=\mathbbm{1}_{B}$, so that the constraints in \eqref{eq:c-beta-dual} are satisfied and%
	\begin{align}
		\Tr[\Gamma_{AB}^{\gamma,N}(K_{AB}-M_{AB})^{\t_B}] & =\Tr[\Gamma_{AB}^{\gamma,N}K_{AB}^{\t_B}]\\
		&  =\Tr[\Gamma_{AB}^{\gamma,N}E_{AB}].
	\end{align}
	We find that%
	\begin{widetext}
	\begin{align}
		\Gamma_{AB}^{\gamma,N}E_{AB}=\frac{1}{2}%
		\begin{pmatrix}
			\sqrt{1-\gamma}-N\gamma+1 & 0 & 0 & \sqrt{1-\gamma}-N\gamma+1\\
			0 & N\gamma & 0 & 0\\
			0 & 0 & -\gamma\left(  N-1\right)   & 0\\
			\sqrt{1-\gamma}+\gamma\left(  N-1\right)  +1 & 0 & 0 & \sqrt{1-\gamma}+\gamma\left(  N-1\right)  +1
		\end{pmatrix},
	\end{align}
	\end{widetext}
	so that%
	\begin{equation}
		\Tr[\Gamma_{AB}^{\gamma,N}E_{AB}]=1+\sqrt{1-\gamma}.
	\end{equation}
	This implies that $C_{\hat{\beta}}(\mathcal{A}_{\gamma,N})\equiv\log_2\hat{\beta}(\mathcal{A}_{\gamma,N})\geq\log_{2}(1+\sqrt{1-\gamma})$. By strong duality, it holds that $C_{\beta}(\mathcal{A}_{\gamma,N})=C_{\hat{\beta}}(\mathcal{A}_{\gamma,N})$. Therefore,
	\begin{equation}\label{eq:c_beta-GADC-upper-bnd}
		C_{\beta}(\mathcal{A}_{\gamma,N})\geq\log_{2}(1+\sqrt{1-\gamma}),
	\end{equation}
	for all $\gamma, N \in [0,1]$. Putting together \eqref{eq:c_beta-GADC-lower-bnd} and \eqref{eq:c_beta-GADC-upper-bnd}, we obtain $C_{\beta}(\mathcal{A}_{\gamma,N})=\log_2(1+\sqrt{1-\gamma})$, as required.
	
	Let us now show that $C_{\zeta}(\mathcal{A}_{\gamma,N})=\log_2(1+\sqrt{1-\gamma})$ for all $\gamma,N\in[0,1]$. Recall from \eqref{eq-C_zeta_primal} that
	\begin{equation}\label{eq-C_zeta_primal_2}
		\zeta(\mathcal{N})=\left\{\begin{array}{l l} \text{min.} & \Tr[S_B] \\ 
			\text{subject to} & V_{AB}\geq\Gamma_{AB}^{\mathcal{N}},\\
			& -\mathbbm{1}_A\otimes S_B\leq V_{AB}^{\t_B}\leq \mathbbm{1}_A\otimes S_B. \end{array}\right.
	\end{equation}
	
	The inequality $C_{\zeta}(\mathcal{A}_{\gamma,0})\leq\log_2(1+\sqrt{1-\gamma})$ has been proven in \cite[Theorem~14]{WXD18}. By inspecting the SDP in \eqref{eq-C_zeta_primal_2}, it is clear that the quantity $\zeta(\mathcal{N})$ is convex in the channel $\mathcal{N}$. Furthermore, it is invariant under unitary pre- and post-processing. Thus, proceeding in a way similar to the proof of the upper bound $C_{\beta}(\mathcal{A}_{\gamma,N})\leq\log_2(1+\sqrt{1-\gamma})$ above, we find that
	\begin{align}
		\zeta(\mathcal{A}_{\gamma,N})&=\zeta((1-N)\mathcal{A}_{\gamma,0}+N\mathcal{A}_{\gamma,1})\\
		&\leq (1-N)\zeta(\mathcal{A}_{\gamma,0})+N\zeta(\mathcal{A}_{\gamma,1})\\
		&=(1-N)\zeta(\mathcal{A}_{\gamma,0})+N\zeta(\mathcal{A}_{\gamma,0})\\
		&=\zeta(\mathcal{A}_{\gamma,0})\\
		&\leq 1+\sqrt{1-\gamma},
	\end{align}
	from which we conclude that
	\begin{equation}\label{eq-C_zeta_pf}
		C_{\zeta}(\mathcal{A}_{\gamma,N})\leq\log_2(1+\sqrt{1-\gamma})
	\end{equation}
	for all $\gamma,N\in[0,1]$.
	
	To arrive at the opposite inequality, consider that the SDP dual to the one in \eqref{eq-C_zeta_primal_2} is given by
	\begin{equation}
		\hat{\zeta}(\mathcal{N})\equiv\left\{\begin{array}{l l} \text{max.} & \Tr[K_{AB}\Gamma_{AB}^{\mathcal{N}}] \\
		\text{subject to} & \Tr_A[E_{AB}+F_{AB}]\leq\mathbbm{1}_B,\\
		& K_{AB}\leq (E_{AB}-F_{AB})^{\t_B},\\
		& K_{AB},E_{AB},F_{AB}\geq 0,\end{array}\right.
	\end{equation}
	where the optimization is with respect to the operators $K_{AB},E_{AB},F_{AB}$. Now, for the GADC, let us make the following choice for $K_{AB}$, $E_{AB}$, $F_{AB}$:
	\begin{align}
		K_{AB}&=\frac{1}{2}\begin{pmatrix} 1 & 0 & 0 & 1\\ 0 & 1 & 0 & 0 \\ 0 & 0 & 1 & 0 \\ 1 & 0 & 0 & 1 \end{pmatrix},\\
		E_{AB}&=\frac{1}{2}\begin{pmatrix} 1 & 0 & 0 & 0 \\ 0 & 1 & 1 & 0 \\ 0 & 1 & 1 & 0 \\ 0 & 0 & 0 & 1 \end{pmatrix},\\
		F_{AB}&=0.
	\end{align}
	Then, we find that the conditions $\Tr_A[E_{AB}+F_{AB}]\leq\mathbbm{1}_B$ and $K_{AB}\leq (E_{AB}-F_{AB})^{\t_B}$ are satisfied with equality. Now,
	\begin{widetext}
	\begin{equation}
		K_{AB}\Gamma_{AB}^{\gamma,N}=\frac{1}{2}\begin{pmatrix} \sqrt{1-\gamma}-N\gamma+1 & 0 & 0 & \sqrt{1-\gamma}-\gamma(1-N)+1 \\ 0 & N\gamma & 0 & 0 \\ 0 & 0 & \gamma(1-N) & 0 \\ \sqrt{1-\gamma}-N\gamma+1 & 0 & 0 & \sqrt{1-\gamma}-\gamma(1-N)+1 \end{pmatrix},
	\end{equation}
	\end{widetext}
	so that taking the trace yields
	\begin{equation}
		\Tr[K_{AB}\Gamma_{AB}^{\gamma,N}]=1+\sqrt{1-\gamma}.
	\end{equation}
	We thus conclude that
	\begin{align}
		C_{\hat{\zeta}}(\mathcal{A}_{\gamma,N})&\equiv \log_2\hat{\zeta}(\mathcal{A}_{\gamma,N})\\
		&\geq \log_2(1+\sqrt{1-\gamma}).
	\end{align}
	By strong duality, it holds that $C_{\zeta}(\mathcal{A}_{\gamma,N})=C_{\hat{\zeta}}(\mathcal{A}_{\gamma,N})$ for all $\gamma,N\in[0,1]$. Therefore, we have that $C_{\zeta}(\mathcal{A}_{\gamma,N})\geq \log_2(1+\sqrt{1-\gamma})$, and combining this with \eqref{eq-C_zeta_pf} means that $C_{\zeta}(\mathcal{A}_{\gamma,N})=\log_2(1+\sqrt{1-\gamma})$, as required.

\section{Covariance parameter for the GADC}\label{app-GADC_cov_parameter}

	Using the definition of the diamond norm in \eqref{eq-diamond_norm}, we can write the quantity $\varepsilon_{\text{cov}}(\mathcal{A}_{\gamma,N})$ as
	\begin{equation}
		\varepsilon_{\text{cov}}(\mathcal{A}_{\gamma,N})=\frac{1}{2}\max_{\psi_{RA}}\Norm{(\mathcal{A}_{\gamma,N}-\mathcal{A}_{\gamma,\frac{1}{2}})(\psi_{RA})}_1.
	\end{equation}
	
	We first show that the maximum is achieved by taking $\ket{\psi}_{RA}$ to be the maximally entangled state, i.e., taking $\ket{\psi}_{RA}=\ket{\Phi^+}_{RA}=\frac{1}{\sqrt{2}}(\ket{0,0}_{RA}+\ket{1,1}_{RA})$. We do this by making use of \cite[Lemma II.3]{LKDW18}. Let $\ket{\psi}_{RA}$ be an arbitrary pure state, and let $\rho_A\coloneqq\Tr_R[\psi_{RA}]$. We take the group $G=\mathbb{Z}_2\times\mathbb{Z}_2$ and the Pauli operators $\{\mathbbm{1},\sigma_x,\sigma_y,\sigma_z\}$ and note that
	\begin{equation}
		\overline{\rho}_A\coloneqq\frac{1}{4}(\rho_A+\sigma_x\rho_A\sigma_x+\sigma_y\rho_A\sigma_y+\sigma_z\rho_A\sigma_z)=\frac{\mathbbm{1}_A}{2}.
	\end{equation}
	Due to this fact, one purification of $\overline{\rho}$ is the maximally entangled state $\ket{\Phi^+}_{RA}$. Therefore, by applying \cite[Lemma II.3]{LKDW18} (with the generalized divergence therein taken to be the trace distance), we obtain
	\begin{align}
		&\Norm{\mathcal{A}_{\gamma,N}(\Phi_{RA}^+)-\mathcal{A}_{\gamma,\frac{1}{2}}(\Phi_{RA}^+)}_{1}\nonumber\\
		&\quad \geq \left\lVert\frac{1}{4}\sum_{g\in G}\ket{g}\bra{g}_P\otimes\mathcal{A}_{\gamma,N}^g(\psi_{RA})\right.\nonumber\\
		&\quad\qquad\left.-\frac{1}{4}\sum_{g\in G}\ket{g}\bra{g}_P\otimes\mathcal{A}_{\gamma,\frac{1}{2}}^g(\psi_{RA})\right\rVert_1,
	\end{align}
	where $\mathcal{A}_{\gamma,N}^g\coloneqq \mathcal{S}_g\circ\mathcal{A}_{\gamma,N}\circ\mathcal{S}_g$, with $\mathcal{S}_g(\cdot)=S_g(\cdot)S_g$ and $S_g\in\{\mathbbm{1},\sigma_x,\sigma_y,\sigma_z\}$. Then, recalling that
	\begin{align}
		\sigma_x\mathcal{A}_{\gamma,\frac{1}{2}}(\sigma_x(\cdot)\sigma_x)\sigma_x&=\mathcal{A}_{\gamma,\frac{1}{2}}(\cdot),\\
		\sigma_z\mathcal{A}_{\gamma,\frac{1}{2}}(\sigma_z(\cdot)\sigma_z)\sigma_z&=\mathcal{A}_{\gamma,\frac{1}{2}}(\cdot),\\
		\Rightarrow \sigma_y\mathcal{A}_{\gamma,\frac{1}{2}}(\sigma_y(\cdot)\sigma_y)\sigma_y&=\mathcal{A}_{\gamma,\frac{1}{2}}(\cdot),
	\end{align}
	we get that $\mathcal{A}_{\gamma,\frac{1}{2}}^g=\mathcal{A}_{\gamma,\frac{1}{2}}$ for all $g\in G$. Therefore,
	\begin{align}
		&\Norm{\mathcal{A}_{\gamma,N}(\Phi_{RA}^+)-\mathcal{A}_{\gamma,\frac{1}{2}}(\Phi_{RA}^+)}_{1}\\
		&\quad\geq \Norm{\frac{1}{4}\sum_{g\in G}\ket{g}\bra{g}_P\otimes (\mathcal{A}_{\gamma,N}^g-\mathcal{A}_{\gamma,\frac{1}{2}})(\psi_{RA})}_1\label{eq-GADC_cov_param_pf1}\\
		&\quad =\frac{1}{4}\sum_{g\in G}\Norm{(\mathcal{A}_{\gamma,N}^g-\mathcal{A}_{\gamma,\frac{1}{2}})(\psi_{RA})}_1,
	\end{align}
	where to obtain the last line we used the fact that all of the operators in the sum in \eqref{eq-GADC_cov_param_pf1} are supported on orthogonal spaces. Then, using \eqref{eq-GADC_Z_covariant} and \eqref{eq-GADC_N_symmetry}, which together imply that $\sigma_y\mathcal{A}_{\gamma,N}(\sigma_y(\cdot)\sigma_y)\sigma_y=\mathcal{A}_{\gamma,1-N}(\cdot)$, we get
	\begin{align}
		&\Norm{\mathcal{A}_{\gamma,N}(\Phi_{RA}^+)-\mathcal{A}_{\gamma,\frac{1}{2}}(\Phi_{RA}^+)}_{1}\\
		&\quad\geq \frac{1}{2}\Norm{(\mathcal{A}_{\gamma,1-N}-\mathcal{A}_{\gamma,\frac{1}{2}})(\psi_{RA})}_1\\
		&\qquad + \frac{1}{2}\Norm{(\mathcal{A}_{\gamma,N}-\mathcal{A}_{\gamma,\frac{1}{2}})(\psi_{RA})}_1.
	\end{align}
	Next, we use the fact that $\mathcal{A}_{\gamma,N}=(1-N)\mathcal{A}_{\gamma,0}+N\mathcal{A}_{\gamma,1}$ to get that
	\begin{align}
		& \Norm{(\mathcal{A}_{\gamma,N}-\mathcal{A}_{\gamma,\frac{1}{2}})(\psi_{RA})}_1\\
		&\quad = \Norm{\left(\left(\frac{1}{2}-N\right)\mathcal{A}_{\gamma,0}-\left(N-\frac{1}{2}\right)\mathcal{A}_{\gamma,1}\right)(\psi_{RA})}_1\\
		&\quad = \left|N-\frac{1}{2}\right|\Norm{(\mathcal{A}_{\gamma,0}-\mathcal{A}_{\gamma,1})(\psi_{RA})}_1,\label{eq-GADC_cov_param_pf3}
	\end{align}
	and
	\begin{align}
		& \Norm{(\mathcal{A}_{\gamma,1-N}-\mathcal{A}_{\gamma,\frac{1}{2}})(\psi_{RA})}_1\\
		&\quad = \Norm{\left(\left(N-\frac{1}{2}\right)\mathcal{A}_{\gamma,0}-\left(\frac{1}{2}-N\right)\mathcal{A}_{\gamma,1}\right)(\psi_{RA})}_1\\
		&\quad = \left|N-\frac{1}{2}\right|\Norm{(\mathcal{A}_{\gamma,0}-\mathcal{A}_{\gamma,1})(\psi_{RA})}_1\\
		&\quad =\Norm{(\mathcal{A}_{\gamma,N}-\mathcal{A}_{\gamma,\frac{1}{2}})(\psi_{RA})}_1
	\end{align}
	Therefore,
	\begin{align}
		&\Norm{\mathcal{A}_{\gamma,N}(\Phi_{RA}^+)-\mathcal{A}_{\gamma,\frac{1}{2}}(\Phi_{RA}^+)}_{1}\\
		&\quad \geq \Norm{(\mathcal{A}_{\gamma,N}-\mathcal{A}_{\gamma,\frac{1}{2}})(\psi_{RA})}_1
	\end{align}
	for all pure states $\psi_{RA}$, which implies that
	\begin{equation}
		\max_{\psi_{RA}}\Norm{(\mathcal{A}_{\gamma,N}-\mathcal{A}_{\gamma,\frac{1}{2}})(\psi_{RA})}_1\leq\Norm{\mathcal{A}_{\gamma,N}(\Phi_{RA}^+)-\mathcal{A}_{\gamma,\frac{1}{2}}(\Phi_{RA}^+)}_1.
	\end{equation}
	Combined with the inequality
	\begin{equation}
		\max_{\psi_{RA}}\Norm{(\mathcal{A}_{\gamma,N}-\mathcal{A}_{\gamma,\frac{1}{2}})(\psi_{RA})}_1\geq \Norm{(\mathcal{A}_{\gamma,N}-\mathcal{A}_{\gamma,\frac{1}{2}})(\Phi_{RA}^+)}_1,
	\end{equation}
	which holds simply by restricting the maximization to the state $\Phi_{RA}^+$, we obtain
	\begin{equation}\label{eq-GADC_cov_param_pf2}
		\varepsilon_{\text{cov}}(\mathcal{A}_{\gamma,N})=\frac{1}{2}\Norm{(\mathcal{A}_{\gamma,N}-\mathcal{A}_{\gamma,\frac{1}{2}})(\Phi_{RA}^+)}_1
	\end{equation}
	for all $\gamma,N\in[0,1]$.
	
	Finally, to calculate the right-hand side of \eqref{eq-GADC_cov_param_pf2}, we observe using \eqref{eq-GADC_cov_param_pf3} that
	\begin{align}
		&\Norm{(\mathcal{A}_{\gamma,N}-\mathcal{A}_{\gamma,\frac{1}{2}})(\Phi_{RA}^+)}_1\nonumber\\
		&\quad = \left|N-\frac{1}{2}\right|\Norm{(\mathcal{A}_{\gamma,0}-\mathcal{A}_{\gamma,1})(\Phi_{RA}^+)}_1\\
		&\quad = \left|N-\frac{1}{2}\right|\Norm{(\mathcal{A}_{\gamma,1}+\mathcal{A}_{\gamma,0}-2\mathcal{A}_{\gamma,0})(\Phi_{RA}^+)}_1\\
		&\quad = |2N-1|\Norm{\left(\frac{1}{2}\mathcal{A}_{\gamma,1}+\frac{1}{2}\mathcal{A}_{\gamma,0}-\mathcal{A}_{\gamma,0}\right)(\Phi_{RA}^+)}_1\\
		&\quad = |2N-1|\Norm{(\mathcal{A}_{\gamma,\frac{1}{2}}-\mathcal{A}_{\gamma,0})(\Phi_{RA}^+)}_1\\
		&\quad = 2|2N-1|\varepsilon_{\text{cov}}(\mathcal{A}_{\gamma,0}).
	\end{align}
	Now, it has been shown in \cite[Appendix C]{LKDW18} that $\varepsilon_{\text{cov}}(\mathcal{A}_{\gamma,0})=\frac{\gamma}{2}$. Therefore,
	\begin{equation}
		\varepsilon_{\text{cov}}(\mathcal{A}_{\gamma,N})=\frac{1}{2}\gamma|2N-1|=\gamma\left|N-\frac{1}{2}\right|,
	\end{equation}
	as required.

\section{Proof of Eq. (\ref{eq-GADC_Esq_UB_pf2})}\label{app-GADC_Esq_UB_pf}

	By restricting the optimization on the right-hand side of \eqref{eq-GADC_Esq_UB_pf2} to pure states $\ket{\theta^p}_{AA'}=\sqrt{1-p}\ket{0,0}_{AA'}+\sqrt{p}\ket{1,1}_{AA'}$, we obtain
	\begin{align}
		\frac{1}{2}\max_{\phi_{AA'}}I(A;B|E_1E_2)_{\tau}&\geq\frac{1}{2}\max_{\theta^p_{AA'}}I(A;B|E_1E_2)_{\tau^p}\\
		&=\frac{1}{2}\max_{p\in[0,1]}I(A;B|E_1E_2)_{\tau^p}.\label{eq-Esq_bd_pf}
	\end{align}
	The remainder of the proof is dedicated to proving the reverse inequality.

	Let $\phi_{AA'}$ be an arbitrary pure state, and let $\rho_{A'}\coloneqq\Tr_A[\phi_{AA'}]$. The state $\tau$ on which we evaluate the conditional mutual information on the left-hand side of \eqref{eq-Esq_bd_pf} is given by
	\begin{equation}
		\tau_{ABE_1E_2}=(\id_{AB}\otimes\mathcal{A}_{\frac{1}{2},0}\otimes\mathcal{A}_{\frac{1}{2},0})(\ket{\psi}\bra{\psi}_{ABE_1'E_2'}),
	\end{equation}
	where
	\begin{equation}\label{eq-Esq_bd_pf9}
		\ket{\psi}_{ABE_1'E_2'}=V_{B'\to BE_2'}^{\gamma N,1}V_{A'\to B'E_1'}^{\frac{\gamma(1-N)}{1-\gamma N},0}\ket{\phi}_{AA'}.
	\end{equation}
	Note that the GADC has only two Kraus operators when the second parameter is either zero or one. Consequently, for any $\gamma'\in[0,1]$, we can take the isometric extensions in \eqref{eq-Esq_bd_pf9} to be of the following form:
	\begin{align}
		V^{\gamma',0}&=A_1\otimes\ket{0}+A_2\otimes\ket{1},\\
		V^{\gamma',1}&=A_3\otimes\ket{0}+A_4\otimes\ket{1}.
	\end{align}
	By using an isometric extension of the same form for the channel $\mathcal{A}_{\frac{1}{2},0}$, we can write $\tau_{ABE_1E_2}$ explicitly as
	\begin{equation}\label{eq-Esq_bd_pf1}
		\begin{aligned}
		&\tau_{ABE_1E_2}=\Tr_{F_1F_2}[\ket{\varphi}\bra{\varphi}_{ABE_1E_2F_1F_2}],\\
		&\ket{\varphi}_{ABE_1E_2F_1F_2}\\
		&\quad =\left(V_{E_1'\to E_1F_1}^{\frac{1}{2},0}\otimes V_{E_2'\to E_2F_2}^{\frac{1}{2},0}\right)V_{B'\to BE_2'}^{\gamma N,1}V_{A'\to B'E_1'}^{\frac{\gamma(1-N)}{1-\gamma N},0}\ket{\phi}_{AA'},
		\end{aligned}
	\end{equation}
	
	Now, the Pauli-$z$ covariance of the GADC is equivalent to the relations $A_1\sigma_z=\sigma_zA_1$, $A_2\sigma_z=-\sigma_zA_2$, $A_3\sigma_z=\sigma_zA_3$, and $A_4\sigma_z=-\sigma_zA_4$. Therefore, writing $V^{\gamma',0}$ as $V^{\gamma',0}=A_1\otimes\sigma_z\ket{0}-A_2\otimes\sigma_z\ket{1}$, for any state $\ket{\psi}$, we obtain
	\begin{align}
		V^{\gamma',0}\sigma_z\ket{\psi}&=A_1\sigma_z\ket{\psi}\otimes\sigma_z\ket{0}-A_2\sigma_z\ket{\psi}\otimes\sigma_z\ket{1}\\
		&=\sigma_zA_1\ket{\psi}\otimes\sigma_z\ket{0}+\sigma_zA_2\ket{\psi}\otimes\sigma_z\ket{1}\\
		&=(\sigma_z\otimes\sigma_z)(A_1\ket{\psi}\otimes\ket{0}+A_2\ket{\psi}\otimes\ket{1})\\
		&=(\sigma_z\otimes\sigma_z)V^{\gamma',0}.\label{eq-Esq_bd_pf7}
	\end{align}
	Similarly, we have
	\begin{equation}\label{eq-Esq_bd_pf8}
		V^{\gamma',1}\sigma_z\ket{\psi}=(\sigma_z\otimes\sigma_z)V^{\gamma',1}\ket{\psi}
	\end{equation}
	for all states $\ket{\psi}$.
	
	Next, we observe that by using the definition of the conditional mutual information in \eqref{eq-QCMI}, along with the definition of the conditional entropy, we can write $I(A;B|E_1E_2)_{\tau}$ as
	\begin{align}
		I(A;B|E_1E_2)_{\tau}&=H(B|E_1E_2)_{\tau}-H(B|E_1E_2A)_{\tau}\\
		&=H(B|E_1E_2)_{\varphi}+H(B|F_1F_2)_{\varphi},\label{eq-Esq_bd_pf2}
	\end{align}
	where to obtain the last line we used the fact that the state $\ket{\varphi}_{ABE_1E_2F_1F_2}$ in \eqref{eq-Esq_bd_pf1} is pure; in particular,
	\begin{align}
		H(B|E_1E_2A)_{\tau}&=H(ABE_1E_2)_{\tau}-H(E_1E_2A)_{\tau}\\
		&=H(F_1F_2)_{\varphi}-H(BF_1F_2)_{\varphi}\\
		&=-H(B|F_1F_2)_{\varphi}.
	\end{align}
	Now, since the right-hand side of \eqref{eq-Esq_bd_pf2} does not contain the $A$ system, the quantity is a function solely of the state $\rho_{A'}$. For convenience, let us define a function $F$ by
	\begin{equation}
		F(\rho_{A'})=I(A;B|E_1E_2)_{\tau}=H(B|E_1E_2)_{\varphi}+H(B|F_1F_2)_{\varphi},
	\end{equation}
	where
	\begin{widetext}
	\begin{equation}
		\begin{aligned}
		\varphi_{BE_1E_2F_1F_2}&\equiv\varphi_{BE_1E_2F_1F_2}(\rho_{A'})\\
		&=\left(V^{\frac{1}{2},0}_{E_1'\to E_1F_1}\otimes V^{\frac{1}{2},0}_{E_2'\to E_2F_2}\right)V_{B'\to BE_2'}^{\gamma N,1}V_{A'\to B'E_1'}^{\frac{\gamma(1-N)}{1-\gamma N}}\rho_{A'}\left(V_{A'\to B'E_1'}^{\frac{\gamma(1-N)}{1-\gamma N}}\right)^\dagger \left(V_{B'\to BE_2'}^{\gamma N,1}\right)^\dagger\left(V_{E_1'\to E_1F_1}^{\frac{1}{2},0}\otimes V_{E_2'\to E_2F_2}^{\frac{1}{2},0}\right)^\dagger
		\end{aligned}
	\end{equation}
	\end{widetext}
	Using the relations in \eqref{eq-Esq_bd_pf7} and \eqref{eq-Esq_bd_pf8}, we get
	\begin{equation}
		\varphi_{BE_1E_2F_1F_2}(\sigma_z\rho_{A'}\sigma_z)=\sigma_z^{\otimes 5}\varphi_{BE_1E_2F_1F_2}(\rho_{A'})\sigma_z^{\otimes 5},
	\end{equation}
	which implies that $F(\sigma_z\rho_{A'}\sigma_z)=F(\rho_{A'})$. Furthermore, since the conditional entropy is concave, so is the function $F$. We thus obtain
	\begin{align}
		F\left(\frac{1}{2}\rho_{A'}+\frac{1}{2}\sigma_z\rho_{A'}\sigma_z\right)&\geq \frac{1}{2}F(\rho_{A'})+\frac{1}{2}F(\sigma_z\rho_{A'}\sigma_z)\\
		&=F(\rho_{A'}).
	\end{align}
	Now, observe that the state $\frac{1}{2}\rho_{A'}+\frac{1}{2}\sigma_z\rho_{A'}\sigma_z$ is diagonal in the standard basis, meaning that it has a purification of the form $\ket{\theta^p}_{AA'}=\sqrt{1-p}\ket{0,0}_{AA'}+\sqrt{p}\ket{1,1}_{AA'}$ for some $p\in[0,1]$, say $p^*$. Therefore, by restricting the optimization $\frac{1}{2}\max_{p\in[0,1]}I(A;B|E_1E_2)_{\tau^p}=\frac{1}{2}\max_{\theta_{AA'}^p}I(A;B|E_1E_2)_{\tau^p}$ to $p^*$, we get
	\begin{align}
		\frac{1}{2}\max_{p\in[0,1]}I(A;B|E_1E_2)_{\tau^p}&\geq F\left(\frac{1}{2}\rho_{A'}+\frac{1}{2}\sigma_z\rho_{A'}\sigma_z\right)\\
		&\geq \frac{1}{2}F(\rho_{A'})\\
		&=\frac{1}{2}I(A;B|E_1E_2)_{\tau}.
	\end{align}
	Since the state $\rho_{A'}$ was arbitrary, we get that
	\begin{equation}
		\frac{1}{2}\max_{p\in[0,1]}I(A;B|E_1E_2)_{\tau^p}\geq\frac{1}{2}\max_{\phi_{AA'}}I(A;B|E_1E_2)_{\tau}.
	\end{equation}
	Combining with the inequality in \eqref{eq-Esq_bd_pf}, we get
	\begin{equation}
		\frac{1}{2}\max_{\phi_{AA'}}I(A;B|E_1E_2)_{\tau}=\frac{1}{2}\max_{p\in[0,1]}I(A;B|E_1E_2)_{\tau^p},
	\end{equation}
	as required.

\section{Proof of Proposition \ref{prop-GADC_Emax}}\label{app-GADC_Emax}

	We start by showing that
	\begin{multline}
		E_{\max}(\mathcal{A}_{\gamma,N})\\=\log_2\left(1-\frac{\gamma}{2}+\frac{1}{2}\sqrt{(\gamma(2N-1))^2+4(1-\gamma)}\right)
		\label{eq:E-max-GADC-1}
	\end{multline}
	for all $\gamma,N$ such that the GADC $\mathcal{A}_{\gamma,N}$ is not entanglement breaking. If the channel $\mathcal{A}_{\gamma,N}$ is entanglement breaking, then the Choi matrix $\Gamma_{AB}^{\gamma,N}$ is separable and PPT, so that we can pick the variable $Y_{AB}$ in the SDP \eqref{eq-E_max_SDP_primal} to be $\Gamma_{AB}^{\gamma,N}$, for which we have $\norm{\Tr_B[Y_{AB}]}_{\infty}=1$. This means that $E_{\max}(\mathcal{A}_{\gamma,N})=0$ in this case. In what follows, we thus assume that $\mathcal{A}_{\gamma,N}$ is not entanglement breaking.
	
	We first establish an upper bound on $\Sigma(\mathcal{A}_{\gamma,N})$ by employing the SDP in \eqref{eq-E_max_SDP_primal}. To determine an ansatz for the variable $Y_{AB}$ therein, we first consider the positive partial transpose of the Choi matrix $\Gamma_{AB}^{\gamma,N}$ from \eqref{eq-GADC_Choi_state}:
	\begin{equation}
		\left(\Gamma_{AB}^{\gamma,N}\right)^{\t_B}=\begin{pmatrix} 1-\gamma N & 0 & 0 & 0 \\ 0 & \gamma N & \sqrt{1-\gamma} & 0 \\ 0 & \sqrt{1-\gamma} & \gamma(1-N) & 0 \\ 0 & 0 & 0 & 1-\gamma(1-N) \end{pmatrix}.
	\end{equation}
	To determine the positive semi-definiteness of this matrix, it suffices to focus on the inner $2\times 2$ matrix, given that $1-\gamma N\geq 0$ and $1-\gamma(1-N)\geq 0$ for all $\gamma,N\in[0,1]$. The eigenvalues of the inner $2\times 2$ matrix are given by
	\begin{equation}
		\lambda_{\pm}\equiv \frac{1}{2}\left(\gamma\pm\sqrt{(\gamma(2N-1))^2+4(1-\gamma)}\right).
	\end{equation}
	We have that $\lambda_+\geq 0$ for all $\gamma,N\in[0,1]$. The condition $\lambda_-\leq 0$ is equivalent to the channel not being entanglement breaking. If we add $-\lambda_-\mathbbm{1}$ to the inner $2\times 2$ matrix, then it becomes positive semi-definite. This leads to the following ansatz for the matrix $Y_{AB}$:
	\begin{align}
		Y_{AB}&=\Gamma_{AB}^{\gamma,N}-\begin{pmatrix} 0 & 0 & 0 & 0 \\ 0 & \lambda_- & 0 & 0 \\ 0 & 0 & \lambda_- & 0 \\ 0 & 0 & 0 & 0\end{pmatrix}\\
		&=\begin{pmatrix} 1-\gamma N & 0 & 0 & \sqrt{1-\gamma} \\ 0 & \gamma N-\lambda_- & 0 & 0 \\ 0 & 0 & \gamma(1-N)-\lambda_- & 0 \\ \sqrt{1-\gamma} & 0 & 0 & 1-\gamma(1-N) \end{pmatrix}.
	\end{align}
	By construction, we have that
	\begin{align}
		Y_{AB}-\Gamma_{AB}^{\gamma,N}&\geq 0,\\
		Y_{AB}^{\t_B}& \geq 0,
	\end{align}
	so that $Y_{AB}$ satisfies the constraints of the SDP in \eqref{eq-E_max_SDP_primal}. Now, computing $\Tr_B[Y_{AB}]$ gives
	\begin{equation}
		\Tr_B[Y_{AB}]=\begin{pmatrix} 1-\lambda_- & 0 \\ 0 & 1-\lambda_- \end{pmatrix},
	\end{equation}
	which implies that $\norm{\Tr_{B}[Y_{AB}]}_{\infty}=1-\lambda_-$. Therefore,
	\begin{equation}
		\Sigma(\mathcal{A}_{\gamma,N})\leq \frac{1}{2}\left(2-\gamma+\sqrt{(\gamma(2N-1))^2+4(1-\gamma)}\right).
	\end{equation}
	
	We now establish a lower bound on $\Sigma(\mathcal{A}_{\gamma,N})$ by considering the SDP dual to the one in \eqref{eq-E_max_SDP_primal}, namely,
	\begin{equation}\label{eq-E_max_SDP_dual}
		\hat{\Sigma}(\mathcal{N})\equiv\left\{\begin{array}{l l} \text{max}. & \Tr[\Gamma_{AB}^{\mathcal{N}}P_{AB}] \\ 
		\text{subject to} & P_{AB},Q_{AB}\geq 0,\\
		& P_{AB}+Q_{AB}^{\t_B}\leq \rho_A\otimes\mathbbm{1}_B,\\
		& \rho_A\geq 0,\\
		& \Tr[\rho_A]\leq 1.\end{array}\right.
	\end{equation}
	By strong duality, it follows that these optimization problems have equal solutions, i.e., $\hat{\Sigma}(\mathcal{N})=\Sigma(\mathcal{N})$ for all quantum channels $\mathcal{N}$.
	
	Now, let
	\begin{align}
		a&\equiv \sqrt{(\gamma(2N-1))^2+4(1-\gamma)},\\
		b&\equiv \frac{a-(2N-1)\gamma}{2a}.
	\end{align}
	Note that $b\in[0,1]$ for all $\gamma,N\in[0,1]$. Then, let
	\begin{align}
		\rho_A&=\begin{pmatrix} b & 0 \\ 0 & 1-b \end{pmatrix},\\
		P_{AB}&=\begin{pmatrix} b & 0 & 0 & \frac{1}{a}\sqrt{1-\gamma} \\ 0 & 0 & 0 & 0 \\ 0 & 0 & 0 & 0 \\ \frac{1}{a}\sqrt{1-\gamma} & 0 & 0 & 1-b \end{pmatrix},\\
		Q_{AB}&=\begin{pmatrix} 0 & 0 & 0 & 0 \\ 0 & b & -\frac{1}{a}\sqrt{1-\gamma} & 0 \\ 0 & -\frac{1}{a}\sqrt{1-\gamma} & 1-b & 0 \\ 0 & 0 & 0 & 0 \end{pmatrix}.
	\end{align}
	We have that $\rho_A\geq 0$ and $\Tr[\rho_A]=1$ for all $\gamma,N\in[0,1]$. Also, for all $\gamma,N\in[0,1]$, the eigenvalues of the corners submatrix of $P_{AB}$ are equal to zero and one, implying that $P_{AB}\geq 0$. Similarly, for all $\gamma,N\in[0,1]$, the eigenvalues of the inner submatrix of $Q_{AB}$ are equal to zero and one, implying that $Q_{AB}\geq 0$. Furthermore, we have that
	\begin{align}
		Q_{AB}^{\t_B}&=\begin{pmatrix} 0 & 0 & 0 & -\frac{1}{a}\sqrt{1-\gamma} \\ 0 & b & 0 & 0 \\ 0 & 0 & 1-b & 0 \\ -\frac{1}{a}\sqrt{1-\gamma} & 0 & 0 & 0 \end{pmatrix},\\
		\rho_A\otimes\mathbbm{1}_B&=\begin{pmatrix} b & 0 & 0 & 0 \\ 0 & b & 0 & 0 \\ 0 & 0 & 1-b & 0 \\ 0 & 0 & 0 & 1-b \end{pmatrix},
	\end{align}
	and so we have that $P_{AB}+Q_{AB}^{\t_B}\leq\rho_A\otimes\mathbbm{1}_B$ (in fact, this inequality is saturated). Thus, all the constraints in \eqref{eq-E_max_SDP_dual} are satisfied. Then, since
	\begin{multline}
		\Tr[\Gamma_{AB}^{\gamma,N}P_{AB}]\\=\frac{1}{2}\left(2-\gamma+\sqrt{(\gamma(2N-1))^2+4(1-\gamma)}\right),
	\end{multline}
	we have that
	\begin{equation}
		\hat{\Sigma}(\mathcal{A}_{\gamma,N})\geq \frac{1}{2}\left(2-\gamma+\sqrt{(\gamma(2N-1))^2+4(1-\gamma)}\right).
	\end{equation}
	This means that
	\begin{equation}
		\Sigma(\mathcal{A}_{\gamma,N})=1-\frac{\gamma}{2}+\frac{1}{2}\sqrt{(\gamma(2N-1))^2+4(1-\gamma)},
	\end{equation}
	thus establishing \eqref{eq:E-max-GADC-1}.
	
	We now show that
	\begin{multline}
		R_{\max}(\mathcal{A}_{\gamma,N})\\
		=\log_2\left(1-\frac{\gamma}{2}+\frac{1}{2}\sqrt{(\gamma(2N-1))^2+4(1-\gamma)}\right).
	\end{multline}
	Due to the inequality in \eqref{eq-Rmax_Emax_ineq}, namely, $R_{\max}(\mathcal{A}_{\gamma,N})\leq E_{\max}(\mathcal{A}_{\gamma,N})$, it suffices to show that
	\begin{equation}
		R_{\max}(\mathcal{A}_{\gamma,N})\geq \log_2\left(1-\frac{\gamma}{2}+\frac{1}{2}\sqrt{(\gamma(2N-1))^2+4(1-\gamma)}\right)
	\end{equation} 
	when $\mathcal{A}_{\gamma,N}$ is not entanglement breaking.
	
	When the channel $\mathcal{A}_{\gamma,N}$ is entanglement breaking, then the Choi matrix $\Gamma_{AB}^{\gamma,N}$ is separable and PPT. This means that we can pick $V_{AB}=(\Gamma_{AB}^{\gamma,N})^{\t_B}$ and $Y_{AB}=0$ in \eqref{eq-R_max}, for which $\norm{\Tr_B[V_{AB}+Y_{AB}]}_{\infty}=\norm{\Tr_B[V_{AB}]}_{\infty}=1$, implying that $R_{\max}(\mathcal{A}_{\gamma,N})=0$ in this case. In what follows, we thus assume that $\mathcal{A}_{\gamma,N}$ is not entanglement breaking.
	
	First, the SDP dual to the one in \eqref{eq-R_max} is
	\begin{equation}\label{eq:R-max-dual-SDP}
		\hat{\Delta}(\mathcal{N})=\left\{\begin{array}{l l} \text{max}. & \Tr[\Gamma_{AB}^{\gamma,N}R_{AB}]\\
		\text{subject to} & -\rho_A\otimes\mathbbm{1}_B\leq R_{AB}^{\t_B}\leq \rho_A\otimes\mathbbm{1}_B,\\
		& \rho_A\geq 0,
		 \Tr[\rho_A]\leq 1. \end{array}\right.
	\end{equation}
	By strong duality, it holds that $\hat{\Delta}(\mathcal{N})=\Delta(\mathcal{N})$.
	
	Let $a\in[0,1]$, which we will specify in more detail later as a function of $\gamma$ and $N$. We pick
	\begin{align}
		\rho_A&=\begin{pmatrix} a & 0 \\ 0 & 1-a \end{pmatrix},\\
		R_{AB}&=\begin{pmatrix} a & 0 & 0 & 2a(1-a) \\ 0 & a(1-2a) & 0 & 0 \\ 0 & 0 & -(1-a)(1-2a) & 0 \\ 2a(1-a) & 0 & 0 & 1-a \end{pmatrix}.
	\end{align}
	Note that $\rho_A\geq 0$ and  $\Tr[\rho_A]=1$. Also, consider that
	\begin{align}
		R_{AB}^{\t_B}&=\begin{pmatrix} a & 0 & 0 & 0 \\ 0 & a(1-2a) & 2a(1-a) & 0 \\ 0 & 2a(1-a) & -(1-a)(1-2a) & 0 \\ 0 & 0 & 0 & 1-a \end{pmatrix},\\
		\rho_A\otimes\mathbbm{1}_B&=\begin{pmatrix} a & 0 & 0 & 0 \\ 0 & a & 0 & 0 \\ 0 & 0 & 1-a & 0 \\ 0 & 0 & 0 & 1-a \end{pmatrix},
	\end{align}
	implying that
	\begin{equation}
		R_{AB}^{\t_B}+\rho_A\otimes\mathbbm{1}_B=2\begin{pmatrix} a & 0 & 0 & 0 \\ 0 & a(1-a) & a(1-a) & 0 \\ 0 & a(1-a) & a(1-a) & 0 \\ 0 & 0 & 0 & 1-a \end{pmatrix},
	\end{equation}
	which is positive semi-definite since $a\in[0,1]$. Also, we have that
	\begin{equation}
		\rho_A\otimes\mathbbm{1}_B-R_{AB}^{\t_B}=2\begin{pmatrix} 0 & 0 & 0 & 0 \\ 0 & a^2 & -a(1-a) & 0 \\ 0 & -a(1-a) & (1-a)^2 & 0 \\ 0 & 0 & 0 & 0 \end{pmatrix},
	\end{equation}
	which has eigenvalues equal to zero and $2(1-2a(1-a))$, the latter being nonnegative for all $a\in[0,1]$.  Thus, our choice of $\rho_A$ and $R_{AB}$ satisfies the constraints in \eqref{eq:R-max-dual-SDP}. Now, computing $\Tr[\Gamma_{AB}^{\gamma,N}R_{AB}]$, we find that
	\begin{multline}
		\Tr[\Gamma_{AB}^{\gamma,N}R_{AB}]= g(a,\gamma,N)\\
		\equiv 1-2(1-N)\gamma -2a^2\left(2\sqrt{1-\gamma}+\gamma\right)\\
		+4a(\sqrt{1-\gamma}+\gamma(1-N)).
	\end{multline}
	We now choose $a$ such that the equation
	\begin{equation}
		1-\frac{\gamma}{2}+\frac{1}{2}\sqrt{(\gamma(2N-1))^2+4(1-\gamma)}=g(a,\gamma,N)
	\end{equation}
	is satisfied. It has solutions
	\begin{equation}\label{eq-R_max_pf1}
		a=\frac{c_1\pm \sqrt{c_1^2+c_2((4N-3)\gamma- c_3)}}{c_2},
	\end{equation}
	where
	\begin{align}
		c_1&\equiv 4\left(\sqrt{1-\gamma}+\gamma(1-N)\right),\\
		c_2&\equiv 4\left(2\sqrt{1-\gamma}+\gamma\right),\\
		c_3&\equiv \sqrt{(\gamma(2N-1))^2+4(1-\gamma)}
	\end{align}
	Note that the solutions for $a$ in \eqref{eq-R_max_pf1} satisfy $a\in[0,1]$ for all $\gamma,N$ such that the GADC is not entanglement breaking. Thus, for this choice of $a$, we conclude that
	\begin{equation}
		\hat{\Delta}(\mathcal{N})\geq 1-\frac{\gamma}{2}+\frac{1}{2}\sqrt{(\gamma(2N-1))^2+4(1-\gamma)}.
	\end{equation}
	We thus have that
	\begin{multline}
		R_{\max}(\mathcal{A}_{\gamma,N})=E_{\max}(\mathcal{A}_{\gamma,N})\\
		=1-\frac{\gamma}{2}+\frac{1}{2}\sqrt{(\gamma(2N-1))^2+4(1-\gamma)},
	\end{multline}
	as required.

\bibliography{Ref}{}

\end{document}